\documentclass[10pt,aps,prb,notitlepage,superscriptaddress,letterpaper, longbibliography,twocolumn]{revtex4-2}
\usepackage{amsmath, amssymb, amsthm, stmaryrd, bbold}
\usepackage{tikz, pgfplots,tikz-cd}
\usetikzlibrary{calc, decorations.markings, decorations.pathmorphing,shapes}
\usepackage{diagbox}

\usepackage{natbib,mdframed}
\allowdisplaybreaks[4]
\pgfplotsset{compat=1.18}

\definecolor{YaleBlue} {HTML}{00356b}

\definecolor{YaleMidBlue} {HTML}{286dc0}
\definecolor{YaleLightBlue} {HTML}{b0d4ff}
\definecolor{YaleMidGrey} {HTML}{e0dcda}
\definecolor{YaleLightGrey} {HTML}{eeeeee}

\definecolor{ao}{rgb}{0.0, 0.5, 0.0}

\definecolor{YaleBlack} {HTML}{222222}
\definecolor{YaleDarkGrey} {HTML}{4a4a4a}
\definecolor{YaleWhite} {HTML}{f9f9f9}
\definecolor{YaleLightGreen}{HTML}{AFB896}
\definecolor{YaleGreen} {HTML}{5f712d}
\definecolor{YaleOrange} {HTML}{bd5319}

\newcommand{\spd}[1]{(#1+1)d}

\theoremstyle{definition}
\newtheorem{theorem}{Theorem}[section]
\newtheorem{corollary}[theorem]{Corollary}

\newtheorem{definition}[theorem]{Definition}
\newtheorem{proposition}[theorem]{Proposition}

\theoremstyle{remark}
\newtheorem{remark}{Remark}

\newtheorem{example}{Example}

\newcommand{\Z}{\mathbb{Z}}
\newcommand{\R}{\mathbb{R}}
\newcommand{\C}{\mathbb{C}}

\newcommand{\M}{\mathbb{M}}
\newcommand{\N}{\mathbb{N}}
\newcommand{\ii}{\mathrm{i}}
\newcommand{\ee}{\mathrm{e}}

\newcommand{\cC}{ {\cal C} } 
\newcommand{\cD}{ {\cal D} }

\newcommand{\cG}{ {\cal G} } 
\newcommand{\cH}{ {\cal H} } 

\newcommand{\cJ}{{\cal J}}
 
\newcommand{\cL}{ {\cal L} } 
\newcommand{\cM}{ {\cal M} } 
\newcommand{\cN}{ {\cal N} } 
 
\newcommand{\cP}{ {\cal P} } 
\newcommand{\cQ}{ {\cal Q} }

\newcommand{\sH}{ {\mathsf H} }

\newcommand{\fA}{\mathfrak{A}}
\newcommand{\fS}{\mathfrak{S}}
\newcommand{\fZ}{\mathfrak{Z}}

\newcommand{\BI}{\mathbf{I}}
\newcommand{\Ve}{\mathsf{Vec}}
\newcommand{\sVec}{\mathsf{sVec}}
\newcommand{\Hilb}{\mathsf{Hilb}}

\newcommand{\one}{\mathbf{1}}

\newcommand{\Tr}{\mathrm{Tr}}

\newcommand{\ev}{\mathrm{ev}}
\newcommand{\coev}{\mathrm{coev}}

\newcommand{\Hom}{\mathrm{Hom}}
\newcommand{\End}{\mathrm{End}}
\newcommand{\categoricalEnd}{{\mathsf{End}}}

\newcommand{\Fun}{\mathsf{Fun}}
\newcommand{\rev}{\mathrm{rev}}

\newcommand{\Rep}{\mathsf{Rep}}

\newcommand{\id}{\mathrm{id}}
\newcommand{\im}{\mathrm{im}}
\newcommand{\TY}{\mathrm{TY}}
\newcommand{\Irr}{\mathrm{Irr}}

\newcommand{\chargef}{\mathrm{fgt}}

\newcommand{\xt}{\times}
\newcommand{\ot}[1][]{\underset{#1}{\otimes}}

\newcommand{\rt}[1][]{\underset{#1}{\triangleright}}
\newcommand{\lt}[1][]{\underset{#1}{\triangleleft}}

\newcommand{\la}{\langle} 
\newcommand{\ra}{\rangle}

\newcommand{\scr}{\mathrm{scr}}
\newcommand{\bl}[1]{{\hat{#1}}}
\newcommand{\fgt}{\mathrm{fgt}}
\newcommand{\spann}{\mathrm{span}}

\tikzstyle string=[thin,postaction={decorate},decoration={markings,
    mark=at position 0.5*\pgfdecoratedpathlength+.5*4pt with {\arrow[scale=1]{stealth}}}]
\newcommand{\diagram}[2]{
\begin{tikzpicture}[baseline=(current bounding box),scale=#1]
#2
\end{tikzpicture}
}

\usepackage[colorlinks, citecolor=YaleBlue, linkcolor=YaleBlue, urlcolor=YaleBlue, breaklinks=true]{hyperref}

\begin{document}
\title{Non-invertible SPTs: an on-site realization of \spd{1} anomaly-free fusion category symmetry}
\author{Chenqi Meng}
\affiliation{Department of Physics, The Chinese University of Hong Kong, Shatin, New Territories, Hong Kong SAR}
\author{Xinping Yang}
\affiliation{Department of Physics, Yale University, New Haven, CT 06520-8120, USA}
\author{Tian Lan}
\affiliation{Department of Physics, The Chinese University of Hong Kong, Shatin, New Territories, Hong Kong SAR}
\author{Zhengcheng Gu}
\affiliation{Department of Physics, The Chinese University of Hong Kong, Shatin, New Territories, Hong Kong SAR}

\begin{abstract}
We investigate (1+1)d symmetry-protected topological (SPT) phases with fusion category symmetries. We emphasize that the UV description of an anomaly-free fusion category symmetry must include the fiber functor, giving rise to a local symmetry action, a charge category, and a trivial phase. We construct an ``onsite'' matrix-product-operator (MPO) version of the Hopf algebra symmetry operators in a lattice model with tensor-product Hilbert space. In particular, we propose a systematic framework for classifying and constructing SPTs with non-invertible symmetries. An SPT phase corresponds to a Q-system in the charge category, such that the Q-system becomes a matrix algebra when the symmetry is forgotten. As an example, we provide an explicit microscopic realization of all three $\mathsf{Rep}^\dagger(D_8)$ SPT phases, including a trivial phase, and further demonstrate the $S_3$-duality among these three SPT phases on lattice. 

\end{abstract}

\maketitle
\tableofcontents

\section{Introduction}
The discovery of symmetry-protected topological (SPT) phases sheds light on our understanding of quantum matter over the past decade \cite{Gu_2009,Chen_2013,Chen_2011,Levin_2012,Chen_2011_spinchain}. Starting with the primordial SPT order, the Haldane phase, numerous efforts have been made in systematically constructing SPTs with invertible symmetries ($G$-SPTs), such as the group-cohomology construction and domain wall decoration approach \cite{Chen_2013,wang2021domainwalldecorationsanomalies}. Along the way, mathematical tools are employed to describe and classify $G$-SPTs \cite{Kong_2020,Kong_2022,kapustin2014symmetryprotectedtopologicalphases,ogata2021classificationgappedgroundstate,Kong_2020_math}. Recently, with an increasing interest in generalized symmetries, the generalization of $G$-SPTs to those with non-invertible symmetries, such as fusion categories, has drawn attention across high-energy physics \cite{diatlyk2023gaugingnoninvertiblesymmetriestopological, shao2024whatsundonetasilectures,schafernameki2023ictplecturesnoninvertiblegeneralized,thorngren2019fusioncategorysymmetryi,Seiberg_2024,Seiberg_2024_Majorana,choi2022noninvertibleglobalsymmetriesstandard,Chang_2019,Roumpedakis_2023,thorngren2021fusioncategorysymmetryii}, 
condensed matter physics \cite{Chatterjee_2024,ning2023building1dlatticemodels,Lootens_2023,lootens2024dualitiesonedimensionalquantumlattice,fechisin2024noninvertiblesymmetryprotectedtopologicalorder}, 
and mathematics \cite{jones2024indexquantumcellularautomata,freed2024topologicalsymmetryquantumfield}, where independent progresses have been made in understanding fusion category symmetries and their associated anomalies. 
In particular, the categorical approach to classifying phases \cite{Kong_2020,Kong_2020_math, bhardwaj2023categoricallandauparadigmgapped,bhardwaj2023generalizedchargesiinoninvertible, lan2024categorysetorders, Ji_2020, Chatterjee_2023,Xu_24} has given pivotal insights into the mathematical structures underlying these symmetries. Despite the decent understanding of topological phases enriched by fusion category symmetries at the macroscopic level, the systematic realization of such phases in a microscopic lattice model remains relatively under-explored, with some constructions given by \cite{ning2023building1dlatticemodels,seifnashri2024clusterstatenoninvertiblesymmetry, Inamura_2022,inamura202411dsptphasesfusion,molnar2022matrixproductoperatoralgebras,jia2024weakhopfnoninvertiblesymmetryprotected,Jia_2024,Garre_Rubio_2023,fechisin2024noninvertiblesymmetryprotectedtopologicalorder,jia2024quantumclusterstatemodel}. 

Similar to $G$-SPTs, the UV description of SPTs with non-invertible symmetries should implement the already-known categorical data, realizing the symmetry action on local Hilbert space. Thus in this work, we aim to complement the existing efforts by developing a comprehensive framework to classify and construct (1+1)d SPT phases enriched by fusion category symmetries. We would like to take a further step beyond the abstract categorical data, emphasizing the role of the fiber functor in defining and realizing fusion category symmetries in a tensor-product Hilbert space lattice model. Explicitly, we define the anomalous-free fusion category symmetry by a pair $(\cC, f)$, where $\cC$ is a unitary fusion category and $f:\cC\rightarrow \Hilb$ is a fiber functor. In topological quantum field theory (TQFT), fiber functors are employed to classify SPT phases based on its defining property, the unique symmetric ground state \cite{thorngren2019fusioncategorysymmetryi,davydov2011fieldtheoriesdefectscentre}. In the UV, fiber functor also implements the locality structure of the fusion category symmetry on the tensor product Hilbert spaces kinematically. Most importantly, once a fiber functor is fixed, the existence of a trivial phase (symmetric product state) is guaranteed as an integral part in describing the fusion category symmetry microscopically. On the other hand, the existence of a trivial phase for non-invertible SPTs ensures the physical response of an SPT phase can be well defined, i.e. the non-trivial edge mode. 

A key distinction between SPTs with non-invertible symmetries and those with invertible symmetries is the absence of the stacking structure. Unlike group symmetries, where stacking defines a group structure between phases, fusion category symmetries do not inherently possess such an operation. Thus the choice of the trivial phase is not canonical despite its existence. Therefore, dualities between SPT phases are crucial for relating different SPT phases \cite{Li_2023}. 

The rest of the paper is organized as follows: in Section \ref{sec:RepD8} we begin with explicit examples of SPT phases with $\Rep^\dagger(D_8)$ and $\Rep^\dagger(\sH_8)$ symmetries. For $\Rep^\dagger(D_8)$, we construct a model realizing three distinct SPT phases and demonstrate how the $S_3$-automorphism relates them in construction. In Section \ref{sec:MPO}, we give an ``onsite" construction of the fusion category symmetry acting on the Hilbert space. After assigning an onsite symmetry action, we generalize the construction in \cite{Lan_2024} and define the general fusion category symmetric lattice model by requiring the local Hilbert space at each lattice site to be objects in the charge category, with interactions represented as morphisms in the charge category. With the generic symmetric lattice model established, in Section \ref{section:QSys} we proceed to classify SPT phases by analyzing the properties of the fixed-point model constructed from the Q-system in the charge category \cite{Inamura_2022, Lan_2024}. The Q-system is a unitary version of the separable algebra and provides a mathematical abstraction of fixed-point tensors \cite{Gu_2009,chen2024cftdtqftd1holographictensor} under the tensor network renormalization. Equivalently, it can be interpreted as the defining datum of the open string TQFT \cite{Fukuma1994,LAUDA_2007,davydov2011fieldtheoriesdefectscentre}. In Section \ref{sec:auto}, we study the ground state degeneracy of the lattice model constructed from the Q-system on an infinite open chain. We show that they realize an SPT phase when the Q-system is simple (i.e. a matrix algebra) when mapped to the category of Hilbert spaces by the forgetful functor. In particular, we prove a correspondence between fiber functors from the symmetry category and Q-systems in the charge categories (which are mapped to matrix algebras). The existence of isocategorical groups \cite{etingof2000isocategoricalgroups} provides an exotic example, thus requiring a refined classification of topological phases enriched by fusion category symmetries. We should consider not only the fiber functors but also the fusion of defects. In such an example, the Q-system lattice model has a unique symmetric ground state, which ought to be a symmetric phase by usual reasoning. However, the fusion of point-like defects can differ from that in the trivial phase. As the fusion of defects is physically detectable, in these examples the symmetry is ``less'' preserved than in a usual SPT phase. Such phases are peculiar to non-invertible symmetries and are absent in SPTs with invertible symmetries.

\section{Preliminary}\label{sec:RepD8}
Given a unitary fusion category $\cC$, it has been proposed in the context of TQFT that $\cC$-symmetry is anomaly-free if $\cC$ admits fiber functor, and the classification of 1D $\cC$-SPT phases is determined by fiber functors on $\cC$ \cite{thorngren2019fusioncategorysymmetryi}. When $\cC = \Ve_G$ \footnote{The category $\Hilb_G$ will be used in the rest of the paper to emphasize unitarity.}, the fiber functors on $\Ve_G$ are classified by $H^2(G, U(1))$, recovering the group cohomology classification for 1D $G$-SPTs \cite{Chen_2011_spinchain,Chen_2013}.

In this work, we focus on the case of SPTs with $\Rep^\dagger(D_8)$ symmetry. As a unitary fusion category, $\Rep^\dagger(D_8)$ is equivalent to the Tambara-Yamagami category \cite{TAMBARA1998692,article} $\mathrm{TY}_{\Z_2 \times \Z_2}^{ \chi(g,h) = (-1)^{g_1h_2+g_2h_1},\epsilon=1}$, characterized by the symmetric non-degenerate bicharacter $\chi$ and the Frobenius-Schur indicator $\epsilon \in \{ \pm 1\}$. It is of particular interest as $\Rep^\dagger(D_8)$ admits three fiber functors \cite{TAMBARA1998692}, i.e. corresponding to three distinct SPT phases \cite{thorngren2019fusioncategorysymmetryi}. Furthermore, these three fiber functors form a torsor under the monoidal equivalences in $\Rep^\dagger(D_8)$. In the literature, there has been work realizes three $\Rep^\dagger(D_8)$ SPT phases on lattice \cite{seifnashri2024clusterstatenoninvertiblesymmetry}, while claiming that there does not exist a trivial phase for SPTs with non-invertible symmetries due to the lack of stacking structure. 

Regardless of the stacking structure, a trivial phase should still exist as one essential physical response of an SPT phase is the non-trivial edge mode: the boundary between a non-trivial SPT phase and the vacuum. The lack of stacking structure of non-invertible SPTs simply suggests the choice of the trivial phase is not canonical.

More importantly, an IR description of $\cC$-SPT phases might not be complete. The example of isocategorical groups \cite{etingof2000isocategoricalgroups} is a contradiction if one considers fusion categories alone as symmetry. It states that different fiber functors of the fusion category $\Rep^\dagger(G)$ give rise to different Hopf algebras. They give distinct fusion rules for symmetry charges, which are directly physically observable.

As we will demonstrate in this paper, there exists a product state once we assign the symmetry operator in an ``onsite'' MPO form, also introduced in \cite{inamura202411dsptphasesfusion} called the completeness condition. The onsiteness might sound contradictory as non-invertible symmetry operators or MPOs act non-locally in tensor product Hilbert spaces. Here, as in \cite{inamura202411dsptphasesfusion}, we generalize the onsite condition of $G$-symmetries to $\cC$-symmetries by requiring that the virtual bond dimension $v_c$ of a MPO for an object $c \in \mathcal{C}$ must equal its quantum dimension $d_c$. This condition directly follows from the dualizability condition on the associated MPOs. Realizing objects in $\mathcal{C}$ by dualizable MPOs is characterized by a fiber functor $f: \mathcal{C} \to \Hilb$.
Therefore, the existence of a trivial phase is the consequence of the microscopic description (i.e. fiber functor) of $\cC$ symmetry and thus is an integral part in characterizing non-invertible SPTs:
\begin{mdframed}
    An anomaly-free fusion category symmetry is a pair $(\mathcal{C}, f)$, where $\mathcal{C}$ is an abstract unitary fusion category, and $f: \mathcal{C} \to\Hilb$ is a fiber functor.
\end{mdframed}

A related condition based on whether a (1+1)d defect TQFT admits a state-sum construction is given in \cite{davydov2011fieldtheoriesdefectscentre}.

Below we summarize our main result within a specific example of SPTs with $\Rep^\dagger(D_8)$ symmetry. 
After fixing the $\Rep^\dagger(D_8)$ symmetry charges, $\Rep^\dagger(D_8)$ symmetry is imposed by \textit{dualizable} MPOs. Then applying the general construction in \cite{Inamura_2022}, we write down a lattice model realizing three $\Rep^\dagger(D_8)$-SPT phases as well as phase transitions between them, with one of them being a product state. 
Moreover, three models proposed in \cite{seifnashri2024clusterstatenoninvertiblesymmetry} can be transformed to ours by a local unitary once we reduce their MPOs to the onsite form.

Furthermore, in Theorem \ref{thm:autoequivalence} we show that the monoidal equivalence $\pi$ on the symmetry category $\mathcal{C}$ induces an algebra $A_\pi$ in the charge category $\mathcal{C}_{\mathcal{M}}^\vee$ such that $_{A_\pi}(\mathcal{C}_{\mathcal{M}}^\vee)_{A_\pi} \cong \mathcal{C}_{\mathcal{M}}^\vee$. This suggests that $A_\pi$ in $\cC_\cM^\vee$ realizes a phase, the defects of which form the category $\cC_\cM^\vee$. Guided by the theorem, we explicitly realize the autoequivalence between different fiber functors of $\Rep^\dagger(D_8)$ on the lattice via a duality transformation.

We fix the dictionary throughout this paper:
\begin{center}
{\renewcommand{\arraystretch}{1.7}
\begin{tabular}{ | m{4.8cm} | m{3.7cm}|} 
  \hline
  (dual) Hopf C$^\star $-algebra & ($\mathsf{H^*}$) $\mathsf{H}$\\ 
  \hline
  Symmetry category $\mathcal{C}$ (unitary fusion category) & $ \Rep^\dagger(\mathsf{H^*})$ \\ 
  \hline
  Charge category $\mathcal{C}_{\Hilb_f}^\vee$ & $\Rep^\dagger(\mathsf{H})$ \\ 
  \hline
  Commuting projector fixed point model $(\Z,\,\sH,\,A,\,-m^\dagger m)$ & $A=W \ot W^*$ Q-system in $\mathcal{C}_{\Hilb_f}^\vee$, projective charge (edge mode) $W$, 
   Hilbert space $\mathcal{H}_i = A$ $\forall \, i \in \Z$, Hamiltonian $H = -\sum_i(m^\dagger m)_i$.
 \\
  \hline
\end{tabular} 
}
\end{center}

\subsection{\texorpdfstring{$\Rep^\dagger(D_8)$}{Rep(D8)} symmetry operator}
The sites of 1D $(\cC,f) = (\Rep^\dagger(D_8),\mathrm{fgt})$-symmetric chain are labeled by integers. 
For a system with $\Rep^\dagger(D_8)$ symmetry, we have the charge category $\cC_{\Hilb_f}^\vee = \Hilb_{D_8}$. The Hilbert space on each site $i\in\Z$ is an object in the charge category, i.e. a finite-dimensional $D_8$-graded Hilbert space $V_i = \oplus_{g\in D_8}(V_i)_g\in \Hilb_{D_8}$, and the total Hilbert space on the entire chain is the tensor product of local Hilbert spaces with respect to the ordering of lattice sites, see Figure \ref{fig:RepD8 lattice}. One can view the charge category $\Hilb_{D_8}$ as the category of $D_8$-domain walls.
\begin{figure}
    \centering
    \diagram{1}{
        \foreach \i in {-2,-1,...,2}
        {
            \fill (\i*1.5,0) node[above]{$\i$} node[below]{$V_{\i}$} circle [radius = .05];
        }
        \foreach \j in {-2,-1,...,1}
        {
            \node[below] at (\j*1.5+.75,0) {$\otimes$};
        }
        \node[below] at (-4,0) {$\cH =$};
    }
    \caption{A $\Rep^\dagger(D_8)$-symmetric chain, where black dots are lattice sites indexed by integers. On each site $i\in\Z$, there is a $\Rep^\dagger(D_8)$-charge, i.e. a $D_8$-graded Hilbert space $V_i$. The total Hilbert space is the tensor product of $D_8$-graded Hilbert spaces at each site respecting the ordering of lattice sites.}
    \label{fig:RepD8 lattice}
\end{figure}
The irreducible charges are 1-dimensional Hilbert spaces labeled by $D_8$ elements $\Big \{\C_g \,| \,g \in D_8 \Big\}$. 

We fix the presentation of $D_8$ as $\la s, r|s^2 = r^4 = e, srs = r^{-1}\ra$ and label the simple objects in the symmetry category $\Rep^\dagger(D_8)$, i.e. irreducible representations of $D_8$ by
\begin{multline*}
    \Irr(\Rep^\dagger(D_8)) = \left\{\left(\C,\rho_\one\right), \left(\C,\rho_a\right),\left(\C,\rho_b\right),\left(\C,\rho_c\right),\right.\\
    \left.\left(\C^2,\rho_\sigma\right)\right\},
\end{multline*}
where $\rho_\one$, $\rho_a$, $\rho_b$ and $\rho_c$ are one-dimensional representations and $\rho_\sigma$ is a two-dimensional representation, see Table \ref{tab:RepD8}.
\begin{table}
    \centering
    \begin{tabular}{|c|c|c|c|c|c|}
    \hline
         \diagbox{$D_8$}{Irr} & $\one$ & $a$ & $b$ & $c$ & $\sigma$\\
         \hline
        $s$ & $1$ & $-1$ & $1$ & $-1$ & $X$\\
        \hline
        $r$ & $1$ & $1$ & $-1$ & $-1$ & $\ii Z$\\
        \hline
    \end{tabular}
    \caption{Irreducible representations of $D_8$ on generators.}
    \label{tab:RepD8}
\end{table}
A local tensor of an MPO is graphically represented as
\[
T_{ij}^{\alpha\beta}=
\begin{tikzpicture}[baseline=(current bounding box),scale=.6]
    \draw[string] (0,0) node[below]{$j$} -- (0,.7);
    \draw[string] (0,1.3) -- (0,2) node[above]{$i$};
    \draw[YaleMidBlue, string] (-.3,1)  --  (-1,1) node[left]{$\alpha$};
    \draw[YaleMidBlue, string] (1,1)node[right]{$\beta$} -- (.3,1) ;
    \filldraw[fill = white] (-.3,.7) rectangle node {$T$} (.3,1.3);
\end{tikzpicture},
\]
where the blue line represents the virtual bond while the black line represents the physical bond. We use lower indices to denote the physical bonds, and upper indices for virtual bonds. We will also label the type of virtual bonds by simple objects in $\Rep^\dagger(D_8)$ and the space of physical bonds by objects in $\Hilb_{D_8}$.

$\Rep^\dagger(D_8)$ symmetry operators are represented by MPOs.
We write the local tensor of the $\Rep^\dagger(D_8)$-MPOs that corresponds to the simple object $s$ in $\Rep^\dagger(D_8)$ on irreducible charge $\C_g$ as \cite{inamura202411dsptphasesfusion,molnar2022matrixproductoperatoralgebras}
\[
(T_s)_{gg} = \rho_{s}(g).
\]
More explicitly,
\begin{equation*}
    \begin{split}
        &(T_\one)_{gg} = 1,\; (T_{a})_{gg} = (-1)^{n_s(g)},\; (T_{b})_{gg} = (-1)^{n_r(g)},\\& (T_{c})_{gg} = (-1)^{n_r(g)}(-1)^{n_s(g)}\,,
    \end{split}
\end{equation*}
where $(-1)^{n_s(g)}$ ($(-1)^{n_r(g)}$) is the number parity of $s$ ($r$) contained in $g$. Here the notation is well-defined as the multiplication order of group elements does not affect the number parity. The non-zero components of the local tensor of the non-invertible symmetry operator are written as
\begin{equation*}
    \begin{alignedat}{2}
        & (T_\sigma)_{ee} =  \one,\quad &&(T_\sigma)_{r^2r^2} = -\one,\\
        &(T_\sigma)_{rr} =  \ii Z,\quad && (T_\sigma)_{r^3r^3} = -\ii Z,\\
        & (T_\sigma)_{ss} = X,\quad &&  (T_\sigma)_{(sr^2)(sr^2)}= -X,\\
        &  (T_\sigma)_{(sr)(sr)} =  Y,\quad &&(T_\sigma)_{(sr^3)(sr^3)} = -Y.
    \end{alignedat}
\end{equation*}

To construct the local tensor of $\Rep^\dagger(D_8)$-MPOs on other choices of symmetry charges, one can decompose them into the irreducible charges, graphically
\[
\begin{tikzpicture}[baseline=(current bounding box),scale=1.5]
    \draw[string] (0,0) -- node[right]{$V$} (0,.8);
    \draw[string] (0,1.2) -- node[right]{$V$} (0,2);
    \draw[YaleMidBlue, string] (-.2,1)  -- node[above]{$s$} (-1,1) ;
    \draw[YaleMidBlue, string] (1,1) -- node[above]{$s$} (.2,1) ;
    \filldraw[fill = white] (-.2,.8) rectangle node {$T_s$} (.2,1.2);
\end{tikzpicture}
  :=\sum_{g\in D_8}\sum_{\alpha\in B(V_g)}
\begin{tikzpicture}[baseline=(current bounding box),scale=1.5]
    \draw[string] (0,0) -- node[right]{$V$} (0,.2);
    \draw[string] (0,.2) -- node[left]{$\C_g$} (0,.8);
    \draw[string] (0,1.2) -- node[left]{$\C_g$} (0,1.8);
    \draw[string] (0,1.8) -- node[right]{$V$} (0,2);
    \draw[YaleMidBlue, string] (-.2,1)  -- node[above]{$s$} (-1,1) ;
    \draw[YaleMidBlue, string] (1,1) -- node[above]{$s$} (.2,1) ;
    \filldraw[fill = white] (-.2,.8) rectangle node {$T_s$} (.2,1.2);
    \filldraw[fill = white] (0,.4) node[right]{$\alpha$} -- (.1,.3) -- (-.1,.3) -- cycle;
    \filldraw[fill = white] (0,1.6) node[right]{$\alpha$} -- (.1,1.7) -- (-.1,1.7) -- cycle;
\end{tikzpicture} 
\]
where $B(V_g)$ represents an orthonormal basis of $V_g$. 
We also use the same notation $T_s$ to denote the local tensor of the $s$-MPO on any local Hilbert spaces.
These MPOs clearly preserve the grading of charges.

For any $D_8$-grading preserving linear map $O:V\rightarrow W$, the $\Rep^\dagger(D_8)$-MPO satisfies the pulling-through condition \cite{Cirac_2021,Lootens_2023,lootens2024dualitiesonedimensionalquantumlattice}:
\[
\begin{tikzpicture}[baseline=(current bounding box),scale=1]
    \draw[string] (0,-.5) -- node[right]{$W$} (0,0);
    \draw[string] (0,.4) -- node[right]{$V$} (0,.8);
    \draw[string] (0,1.2) -- node[right]{$V$} (0,2);
    \draw[YaleMidBlue, string] (-.2,1)  -- node[above]{$s$} (-1,1) ;
    \draw[YaleMidBlue, string] (1,1) -- node[above]{$s$} (.2,1) ;
    \filldraw[fill = white] (-.2,.8) rectangle node {$T_s$} (.2,1.2);
    \filldraw[fill = white] (-.2,0) rectangle node {$O$} (.2,.4);
\end{tikzpicture}=
\begin{tikzpicture}[baseline=(current bounding box),scale=1]
    \draw[string] (0,2) -- node[right]{$V$} (0,2.5);
    \draw[string] (0,0) -- node[right]{$W$} (0,.8);
    \draw[string] (0,1.2) -- node[right]{$W$} (0,1.6);
    \draw[YaleMidBlue, string] (-.2,1)  -- node[above]{$s$} (-1,1) ;
    \draw[YaleMidBlue, string] (1,1) -- node[above]{$s$} (.2,1) ;
    \filldraw[fill = white] (-.2,.8) rectangle node {$T_s$} (.2,1.2);
    \filldraw[fill = white] (-.2,1.6) rectangle node {$O$} (.2,2);
\end{tikzpicture}.
\]

The contraction of MPOs is given by
\begin{multline*}
    (T^{\{1,\cdots ,n\}}_{s})^{\alpha\beta}:= \\\sum_{i_1\cdots i_n}\sum_{j_1\cdots j_n}(T_{s})_{i_1j_1}^\alpha(T_{s})_{i_2j_2}\cdots (T_{s})_{i_nj_n}^{\;\;\;\,\beta}\\
    |i_1\cdots i_n\ra\la j_1\cdots j_n|.
\end{multline*}
When MPO composes, we have
\[
    (T_{s})^{\alpha\beta}_{gg}(T_{r})^{\gamma \delta}_{gg} = \rho_s(g)^{\alpha\beta}\rho_{r}(g)^{\gamma\delta} = \rho_{s\ot r}(g)^{(\alpha\gamma)(\beta\delta)}.
\]
Thus the composed MPO can be decomposed by several rank-3 tensors $\mu$,  
\[
    \begin{tikzpicture}[baseline=(current bounding box),scale=.4]
    \draw[string] (0,0) node[below]{$g$} -- (0,.7);
    \draw[string] (0,1.3) -- (0,2);
    \draw[YaleMidBlue, string] (-.3,1)  -- node[above]{$r$} (-2,1) node[left]{$\gamma$};
    \draw[YaleMidBlue, string] (2,1)node[right]{$\delta$} -- node[above]{$r$} (.3,1) ;
    \filldraw[fill = white] (-.3,.7) rectangle node[font = \tiny] {$T_r$} (.3,1.3);
    \draw[string] (0,2) -- (0,2.7);
    \draw[string] (0,3.3) -- (0,4) node[above]{$g$};
    \draw[YaleMidBlue, string] (-.3,3)  -- node[above]{$s$}  (-2,3) node[left]{$\alpha$};
    \draw[YaleMidBlue, string] (2,3)node[right]{$\beta$} -- node[above]{$s$} (.3,3) ;
    \filldraw[fill = white] (-.3,2.7) rectangle node[font = \tiny] {$T_s$} (.3,3.3);
    \end{tikzpicture}
    =
    \sum_{l\in \Irr(\Rep^\dagger(D_8))}\sum_{\mu = 1}^{N^{rs}_l}
    \begin{tikzpicture}[baseline=(current bounding box),scale=.7]
    \draw[string] (0,0) node[below]{$g$} -- (0,1.7);
    \draw[string] (0,2.3) -- (0,4) node[above]{$g$};
    \draw[YaleMidBlue,string] (1,2) -- node[above]{$l$} (0,2);
    \draw[YaleMidBlue,string] (0,2) --node[above]{$l$} (-1,2);
    \draw[YaleMidBlue,string] (-1,2) -- node[above]{$s$} (-2,3);
    \draw[YaleMidBlue,string] (-1,2) -- node[below]{$r$} (-2,1);
    \draw[YaleMidBlue,string] (2,3) -- node[above]{$s$} (1,2);
    \draw[YaleMidBlue,string] (2,1) -- node[below]{$r$} (1,2);
    \filldraw[fill = white] (-.3,1.7) rectangle node {$T_l$} (.3,2.3);
    \filldraw[fill = white] (-1.2,1.8) rectangle node {$\mu$} (-.8,2.2);
    \filldraw[fill = white] (.8,1.8) rectangle node {$\overline{\mu}$} (1.2,2.2);
    \end{tikzpicture}.
\]
which are intertwiners of $\Rep^\dagger(D_8)$:
\[
    \mu\rho_l(g) = \rho_{s\ot r}(g)\mu,\quad \forall g\in D_8.
\]
More concretely, the invertible parts of the MPO satisfy the group multiplication rule:
\begin{equation*}
    \begin{split}
       &(T_a)^{11}(T_a)^{11} = \id,\quad (T_b)^{11}(T_b)^{11} =\id,\\
      &(T_a)^{11}(T_b)^{11} =(T_b)^{11}(T_a)^{11} = (T_c)^{11}. 
    \end{split}
\end{equation*}

And the invertible part multiplies with the non-invertible part by
\begin{align*}
    \begin{tikzpicture}[baseline=(current bounding box),scale=.4]
        \draw[string] (0,0) -- (0,.7);
        \draw[string] (0,1.3) -- (0,2);
        \draw[YaleMidBlue, string] (-.3,1)  -- node[above]{$\sigma$} (-2,1);
        \draw[YaleMidBlue, string] (2,1) -- node[above]{$\sigma$} (.3,1);
        \filldraw[fill = white] (-.3,.7) rectangle node[font = \tiny] {$T_\sigma$} (.3,1.3);
        \draw[string] (0,2) -- (0,2.7);
        \draw[string] (0,3.3) -- (0,4);
        \draw[YaleMidBlue,dashed, string] (-.3,3)  -- node[above]{$a$} (-2,3);
        \draw[YaleMidBlue,dashed, string] (2,3) -- node[above]{$a$} (.3,3);
        \filldraw[fill = white] (-.3,2.7) rectangle node[font = \tiny] {$T_a$} (.3,3.3);
    \end{tikzpicture}
    &=
    \begin{tikzpicture}[baseline=(current bounding box),scale=.7]
        \draw[string] (0,1) -- (0,1.7);
        \draw[string] (0,2.3) -- (0,3);
        \draw[YaleMidBlue,string] (1,2) -- node[above]{$\sigma$} (0,2);
        \draw[YaleMidBlue,string] (0,2) --node[above]{$\sigma$} (-1,2);
        \draw[YaleMidBlue,dashed,string] (-1,2) -- node[above]{$a$} (-2,3);
        \draw[YaleMidBlue,string] (-1,2) -- node[below]{$\sigma$} (-2,1);
        \draw[YaleMidBlue,dashed,string] (2,3) -- node[above]{$a$} (1,2);
        \draw[YaleMidBlue,string] (2,1) -- node[below]{$\sigma$} (1,2);
        \filldraw[fill = white] (-.3,1.7) rectangle node {$T_\sigma$} (.3,2.3);
        \filldraw[fill = white] (-1.2,1.8) rectangle node {$Z$} (-.8,2.2);
        \filldraw[fill = white] (.8,1.8) rectangle node {$Z$} (1.2,2.2);
    \end{tikzpicture},\\
    \begin{tikzpicture}[baseline=(current bounding box),scale=.4]
        \draw[string] (0,0) -- (0,.7);
        \draw[string] (0,1.3) -- (0,2);
        \draw[YaleMidBlue, string] (-.3,1)  -- node[above]{$\sigma$} (-2,1);
        \draw[YaleMidBlue, string] (2,1) -- node[above]{$\sigma$} (.3,1) ;
        \filldraw[fill = white] (-.3,.7) rectangle node[font = \tiny] {$T_\sigma$} (.3,1.3);
        \draw[string] (0,2) -- (0,2.7);
        \draw[string] (0,3.3) -- (0,4);
        \draw[YaleMidBlue,dashed, string] (-.3,3)  -- node[above]{$b$} (-2,3);
        \draw[YaleMidBlue,dashed, string] (2,3) -- node[above]{$b$} (.3,3);
        \filldraw[fill = white] (-.3,2.7) rectangle node[font = \tiny] {$T_b$} (.3,3.3);
    \end{tikzpicture}
    &=
    \begin{tikzpicture}[baseline=(current bounding box),scale=.7]
        \draw[string] (0,1) -- (0,1.7);
        \draw[string] (0,2.3) -- (0,3);
        \draw[YaleMidBlue,string] (1,2) -- node[above]{$\sigma$} (0,2);
        \draw[YaleMidBlue,string] (0,2) --node[above]{$\sigma$} (-1,2);
        \draw[YaleMidBlue,dashed,string] (-1,2) -- node[above]{$b$} (-2,3);
        \draw[YaleMidBlue,string] (-1,2) -- node[below]{$\sigma$} (-2,1);
        \draw[YaleMidBlue,dashed,string] (2,3) -- node[above]{$b$} (1,2);
        \draw[YaleMidBlue,string] (2,1) -- node[below]{$\sigma$} (1,2);
        \filldraw[fill = white] (-.3,1.7) rectangle node {$T_\sigma$} (.3,2.3);
        \filldraw[fill = white] (-1.2,1.8) rectangle node {$X$} (-.8,2.2);
        \filldraw[fill = white] (.8,1.8) rectangle node {$X$} (1.2,2.2);
    \end{tikzpicture},\\
    \begin{tikzpicture}[baseline=(current bounding box),scale=.4]
        \draw[string] (0,0) -- (0,.7);
        \draw[string] (0,1.3) -- (0,2);
        \draw[YaleMidBlue, dashed,string] (-.3,1)  -- node[above]{$a$} (-2,1) ;
        \draw[YaleMidBlue,dashed, string] (2,1) -- node[above]{$a$} (.3,1) ;
        \filldraw[fill = white] (-.3,.7) rectangle node[font = \tiny] {$T_a$} (.3,1.3);
        \draw[string] (0,2) -- (0,2.7);
        \draw[string] (0,3.3) -- (0,4);
        \draw[YaleMidBlue, string] (-.3,3)  -- node[above]{$\sigma$} (-2,3);
        \draw[YaleMidBlue, string] (2,3) -- node[above]{$\sigma$} (.3,3) ;
        \filldraw[fill = white] (-.3,2.7) rectangle node[font = \tiny] {$T_\sigma$} (.3,3.3);
    \end{tikzpicture}
    &=
    \begin{tikzpicture}[baseline=(current bounding box),scale=.7]
        \draw[string] (0,1) -- (0,1.7);
        \draw[string] (0,2.3) -- (0,3);
        \draw[YaleMidBlue,string] (1,2) -- node[above]{$\sigma$} (0,2);
        \draw[YaleMidBlue,string] (0,2) --node[above]{$\sigma$} (-1,2);
        \draw[YaleMidBlue,string] (-1,2) -- node[above]{$\sigma$} (-2,3);
        \draw[YaleMidBlue,dashed,string] (-1,2) -- node[below]{$a$} (-2,1);
        \draw[YaleMidBlue,string] (2,3) -- node[above]{$\sigma$} (1,2);
        \draw[YaleMidBlue,dashed,string] (2,1) -- node[below]{$a$} (1,2);
        \filldraw[fill = white] (-.3,1.7) rectangle node {$T_\sigma$} (.3,2.3);
        \filldraw[fill = white] (-1.2,1.8) rectangle node {$Z$} (-.8,2.2);
        \filldraw[fill = white] (.8,1.8) rectangle node {$Z$} (1.2,2.2);
      \end{tikzpicture},\\
    \begin{tikzpicture}[baseline=(current bounding box),scale=.4]
        \draw[string] (0,0) -- (0,.7);
        \draw[string] (0,1.3) -- (0,2);
        \draw[YaleMidBlue, dashed,string] (-.3,1)  -- node[above]{$b$} (-2,1);
        \draw[YaleMidBlue, dashed,string] (2,1) -- node[above]{$b$} (.3,1) ;
        \filldraw[fill = white] (-.3,.7) rectangle node[font = \tiny] {$T_b$} (.3,1.3);
        \draw[string] (0,2) -- (0,2.7);
        \draw[string] (0,3.3) -- (0,4);
        \draw[YaleMidBlue, string] (-.3,3)  -- node[above]{$\sigma$} (-2,3);
        \draw[YaleMidBlue, string] (2,3) -- node[above]{$\sigma$} (.3,3);
        \filldraw[fill = white] (-.3,2.7) rectangle node[font = \tiny] {$T_\sigma$} (.3,3.3);
    \end{tikzpicture}
    &=
    \begin{tikzpicture}[baseline=(current bounding box),scale=.7]
        \draw[string] (0,1) -- (0,1.7);
        \draw[string] (0,2.3) -- (0,3);
        \draw[YaleMidBlue,string] (1,2) -- node[above]{$\sigma$} (0,2);
        \draw[YaleMidBlue,string] (0,2) --node[above]{$\sigma$} (-1,2);
        \draw[YaleMidBlue,string] (-1,2) -- node[above]{$\sigma$} (-2,3);
        \draw[YaleMidBlue,dashed,string] (-1,2) -- node[below]{$b$} (-2,1);
        \draw[YaleMidBlue,string] (2,3) -- node[above]{$\sigma$} (1,2);
        \draw[YaleMidBlue,dashed,string] (2,1) -- node[below]{$b$} (1,2);
        \filldraw[fill = white] (-.3,1.7) rectangle node {$T_\sigma$} (.3,2.3);
        \filldraw[fill = white] (-1.2,1.8) rectangle node {$X$} (-.8,2.2);
        \filldraw[fill = white] (.8,1.8) rectangle node {$X$} (1.2,2.2);
    \end{tikzpicture}.
\end{align*}
Note that by requiring the orthogonality condition, all these rank-3 tensors are fixed up to a phase factor.
Let
\begin{align*}
    |\psi_e\ra &= \frac{1}{\sqrt{2}}(|01\ra+|10\ra),\\
    |\psi_a\ra &= \frac{1}{\sqrt{2}}(|01\ra-|10\ra),\\
    |\psi_b\ra &=\frac{1}{\sqrt{2}}(|00\ra + |11\ra),\\
    |\psi_c\ra &=\frac{1}{\sqrt{2}}(|00\ra - |11\ra).
\end{align*}
Then the non-invertible tensor composes to
\begin{align*}
    \begin{tikzpicture}[baseline=(current bounding box),scale=.5]
        \draw[string] (0,0) -- (0,.8);
        \draw[string] (0,1.2) -- (0,2);
        \draw[YaleMidBlue, string] (-.3,1)  -- node[above]{$\sigma$} (-2,1);
        \draw[YaleMidBlue, string] (2,1) -- node[above]{$\sigma$} (.3,1);
        \filldraw[fill = white] (-.3,.7) rectangle node[font = \tiny] {$T_\sigma$} (.3,1.3);
        \draw[string] (0,2) -- (0,2.7);
        \draw[string] (0,3.3) -- (0,4);
        \draw[YaleMidBlue, string] (-.3,3)  -- node[above]{$\sigma$} (-2,3);
        \draw[YaleMidBlue, string] (2,3) -- node[above]{$\sigma$} (.3,3);
        \filldraw[fill = white] (-.3,2.7) rectangle node[font = \tiny] {$T_\sigma$} (.3,3.3);
  \end{tikzpicture}
  &=
  \begin{tikzpicture}[baseline=(current bounding box),scale=.9]
        \draw[string] (0,1) -- (0,1.8);
        \draw[string] (0,2.2) -- (0,3);
        \draw[YaleMidBlue,dashed,string] (1,2) -- node[above]{$e$} (0,2);
        \draw[YaleMidBlue,dashed,string] (0,2) --node[above]{$e$} (-1,2);
        \draw[YaleMidBlue,string] (-1,2) -- node[above]{$\sigma$} (-2,3);
        \draw[YaleMidBlue,string] (-1,2) -- node[below]{$\sigma$} (-2,1);
        \draw[YaleMidBlue,string] (2,3) -- node[above]{$\sigma$} (1,2);
        \draw[YaleMidBlue,string] (2,1) -- node[below]{$\sigma$} (1,2);
        \filldraw[fill = white] (-.2,1.8) rectangle node {$T_e$} (.2,2.2);
        \filldraw[fill = white] (-1.2,1.8) rectangle node {$\psi_e$} (-.8,2.2);
        \filldraw[fill = white] (.8,1.8) rectangle node {$\psi_e^\dagger$} (1.2,2.2);
      \end{tikzpicture}\\
      &+
      \begin{tikzpicture}[baseline=(current bounding box),scale=.9]
      \draw[string] (0,1) -- (0,1.8);
      \draw[string] (0,2.2) -- (0,3);
        \draw[YaleMidBlue,dashed,string] (1,2) -- node[above]{$a$} (0,2);
        \draw[YaleMidBlue,dashed,string] (0,2) --node[above]{$a$} (-1,2);
        \draw[YaleMidBlue,string] (-1,2) -- node[above]{$\sigma$} (-2,3);
        \draw[YaleMidBlue,string] (-1,2) -- node[below]{$\sigma$} (-2,1);
        \draw[YaleMidBlue,string] (2,3) -- node[above]{$\sigma$} (1,2);
        \draw[YaleMidBlue,string] (2,1) -- node[below]{$\sigma$} (1,2);
        \filldraw[fill = white] (-.2,1.8) rectangle node {$T_a$} (.2,2.2);
        \filldraw[fill = white] (-1.2,1.8) rectangle node {$\psi_a$} (-.8,2.2);
        \filldraw[fill = white] (.8,1.8) rectangle node {$\psi_a^\dagger$} (1.2,2.2);
      \end{tikzpicture}\\
      &+
      \begin{tikzpicture}[baseline=(current bounding box),scale=.9]
      \draw[string] (0,1) -- (0,1.8);
      \draw[string] (0,2.2) -- (0,3);
        \draw[YaleMidBlue,dashed,string] (1,2) -- node[above]{$b$} (0,2);
        \draw[YaleMidBlue,dashed,string] (0,2) --node[above]{$b$} (-1,2);
        \draw[YaleMidBlue,string] (-1,2) -- node[above]{$\sigma$} (-2,3);
        \draw[YaleMidBlue,string] (-1,2) -- node[below]{$\sigma$} (-2,1);
        \draw[YaleMidBlue,string] (2,3) -- node[above]{$\sigma$} (1,2);
        \draw[YaleMidBlue,string] (2,1) -- node[below]{$\sigma$} (1,2);
        \filldraw[fill = white] (-.2,1.8) rectangle node {$T_b$} (.2,2.2);
        \filldraw[fill = white] (-1.2,1.8) rectangle node {$\psi_b$} (-.8,2.2);
        \filldraw[fill = white] (.8,1.8) rectangle node {$\psi_b^\dagger$} (1.2,2.2);
      \end{tikzpicture}\\
      &+
      \begin{tikzpicture}[baseline=(current bounding box),scale=.9]
      \draw[string] (0,1) -- (0,1.8);
      \draw[string] (0,2.2) -- (0,3);
        \draw[YaleMidBlue,dashed,string] (1,2) -- node[above]{$c$} (0,2);
        \draw[YaleMidBlue,dashed,string] (0,2) --node[above]{$c$} (-1,2);
        \draw[YaleMidBlue,string] (-1,2) -- node[above]{$\sigma$} (-2,3);
        \draw[YaleMidBlue,string] (-1,2) -- node[below]{$\sigma$} (-2,1);
        \draw[YaleMidBlue,string] (2,3) -- node[above]{$\sigma$} (1,2);
        \draw[YaleMidBlue,string] (2,1) -- node[below]{$\sigma$} (1,2);
        \filldraw[fill = white] (-.2,1.8) rectangle node {$T_c$} (.2,2.2);
        \filldraw[fill = white] (-1.2,1.8) rectangle node {$\psi_c$} (-.8,2.2);
        \filldraw[fill = white] (.8,1.8) rectangle node {$\psi_c^\dagger$} (1.2,2.2);
      \end{tikzpicture}.
\end{align*}

Those rank-3 tensors further give the $F$-symbol of $\Rep^\dagger(D_8)$:
\[
    \begin{tikzpicture}[baseline={($(current bounding box)$)}]
        \draw[string,YaleMidBlue] (0,-.6) node[below]{$l$} -- (0,0) node[above]{$\beta$};
        \draw[string,YaleMidBlue] (0,0) -- node[below left]{$m$} (-.6,.6);
        \draw[string,YaleMidBlue] (0,0) -- (1.2,1.2) node[above]{$k$};
        \draw[string,YaleMidBlue] (-.6,.6) node[above]{$\alpha$} -- (-1.2,1.2) node[above]{$i$};
        \draw[string,YaleMidBlue] (-.6,.6) -- (0,1.2) node[above]{$j$};
        \filldraw[fill=white,thick] (0,0) circle [radius=.05];
        \filldraw[fill=white,thick] (-.6,.6) circle [radius=.05];
    \end{tikzpicture}
    =\sum_{n\gamma\delta}[F^{ijk}_{l}]_{m\alpha\beta}^{n\gamma\delta}
    \begin{tikzpicture}[baseline={($(current bounding box)$)}]
        \draw[string,YaleMidBlue] (0,-.6) node[below]{$l$} -- (0,0) node[above]{$\delta$};
        \draw[string,YaleMidBlue] (0,0) --  (-1.2,1.2)node[above]{$i$};
        \draw[string,YaleMidBlue] (0,0) --node[below right]{$n$} (.6,.6) ;
        \draw[string,YaleMidBlue] (.6,.6) node[above]{$\gamma$} -- (1.2,1.2) node[above]{$k$};
        \draw[string,YaleMidBlue] (.6,.6) -- (0,1.2) node[above]{$j$};
        \filldraw[fill=white,thick] (0,0) circle [radius=.05];
        \filldraw[fill=white,thick] (.6,.6) circle [radius=.05];
    \end{tikzpicture}
\]

For each dual object $s^*$ of $s$, there is a corresponding dual MPO tensor
\[
    (T_{s^*})_{gg} = \rho_s(g^{-1})^T
\]
for each symmetry charge, denoted by reversing the arrow
\[
\begin{tikzpicture}[baseline=(current bounding box),scale=1]
    \draw[string] (0,0) -- (0,.7);
    \draw[string] (0,1.3) -- (0,2);
    \draw[YaleMidBlue, string] (-1,1)  -- node[above]{$s$}  (-.3,1);
    \draw[YaleMidBlue, string] (.3,1) -- node[above]{$s$}  (1,1) ;
    \filldraw[fill = white] (-.3,.7) rectangle node {$T_{s^*}$} (.3,1.3);
  \end{tikzpicture}.
\]
Furthermore, MPO tensors and dual tensors satisfy the dualizable property:
\begin{equation}\label{eq:duality condition}
    \begin{tikzpicture}[baseline=(current bounding box),scale=.8]
        \draw[string] (0,0) -- (0,.7);
        \draw[YaleMidBlue, string] (-.3,1)  -- node[above]{$s$} (-1,1) node[left]{$\alpha$};
        \draw[string] (0,1.3)  -- node[left]{$V$} (0,2.7);
        \draw[string] (0,3.3) -- (0,4);
        \draw[YaleMidBlue, string] (-1,3) node[left]{$\beta$}  -- node[above]{$s$}  (-.3,3);
        \draw[YaleMidBlue, string] (0,3) arc  (90:-90:1);
        \filldraw[fill = white] (-.3,.7) rectangle node {$T_s$} (.3,1.3);
        \filldraw[fill = white] (-.3,2.7) rectangle node {$T_{s^*}$} (.3,3.3);
  \end{tikzpicture} 
   =\delta_{\alpha\beta}\id_V
   =
   \begin{tikzpicture}[baseline=(current bounding box),scale=.8]
        \draw[string] (0,0) -- (0,.7);
        \draw[YaleMidBlue, string] (.3,1)  -- node[above]{$s$} (1,1) node[right]{$\alpha$};
        \draw[string] (0,1.3) --  node[right]{$V$} (0,2.7);
        \draw[string] (0,3.3) -- (0,4);
        \draw[YaleMidBlue, string] (1,3) node[right]{$\beta$}  -- node[above]{$s$}  (.3,3);
        \draw[YaleMidBlue, string] (0,3) arc  (90:270:1);
        \filldraw[fill = white] (-.3,.7) rectangle node {$T_{s^*}$} (.3,1.3);
        \filldraw[fill = white] (-.3,2.7) rectangle node {$T_{s}$} (.3,3.3);
  \end{tikzpicture}.
\end{equation}

Here, the bond dimension of the $s$-MPO equals the quantum dimension of $s$ by construction. This representation is, in fact, the minimal representation of the $\Rep^\dagger(D_8)$-MPO. 
However, given a set of MPOs that reproduces the correct fusion rule and F-symbols of $\Rep^\dagger(D_8)$, the dualizability condition serves as a criterion. When it is satisfied, the closure of the $s$-MPO yields a virtual bond dimension $v_s \in \mathbb{Z}_+$. Combined with the correct $\Rep^\dagger(D_8)$-fusion rule, the set of positive integers $\{v_s\}_{s\in\Irr(\Rep^\dagger(D_8))}$ satisfies the standard fusion equation
\[
    v_rv_s = \sum_{l}N^{rs}_{l}v_l.
\]
The only solution to this equation is 
\begin{equation}\label{eq:onsite}
    v_s = d_s.
\end{equation}
This equation defines the \textit{on-site} criterion of fusion category symmetry MPOs \cite{inamura202411dsptphasesfusion}. As we will see in Appendix \ref{sec:reduction}, once we reduce the $\Rep^\dagger(D_8)$ operators in \cite{seifnashri2024clusterstatenoninvertiblesymmetry} to its dualizable MPO form, the cluster state becomes the trivial $\Rep^\dagger(D_8)$ SPT phase. 

To understand the dualizability condition and the on-siteness, consider the partition function of a $\Rep^\dagger(D_8)$-symmetric Hamiltonian at zero temperature. The partition function can be expressed as a $\Rep^\dagger(D_8)$-symmetric tensor network by performing a Suzuki-Trotter expansion. The $\Rep^\dagger(D_8)$-MPOs correspond to defect lines across this tensor network. The pulling-through condition and the dualizability condition ensure that the MPOs can be arbitrarily deformed on the tensor network without altering the partition function. In the continuum limit, these $\Rep^\dagger(D_8)$-MPOs correspond to the topological defect lines (TDLs) shown in Figure \ref{fig:TDL}. As shown in \cite{thorngren2019fusioncategorysymmetryi}, the expectation value of a TDL is equal to its quantum dimension. Due to the pulling-through condition and dualizability condition, the corresponding MPO can be shrunk to the bond of a single local tensor, which becomes the virtual bond dimension of the MPO. We thus obtain that the quantum dimension of the topological defect line is equal to the virtual bond dimension of the MPO. 

In Section \ref{sec:MPO}, we give a detailed discussion on constructing anomaly-free fusion category symmetry operators as onsite MPOs from Hopf algebra. A Hopf algebra $\sH$ is equivalent to the representation category $\Rep(\sH^*)$ equipped with a forgetful functor (a fiber functor) \cite{Kitaev_2012}. It provides the complete data needed to describe an anomaly-free symmetry microscopically. One counter-example is the non-invertible line in the Ising category, the quantum dimension of which is $\sqrt{2}$. On the lattice, it is the Krammers-Wannier operator in the transverse-field Ising model. In this case, the Hopf algebra may be replaced with the weak Hopf algebra \cite{jia2024weakhopfnoninvertiblesymmetryprotected,molnar2022matrixproductoperatoralgebras} and the dualizability condition is not satisfied.

\begin{figure}[htbp]
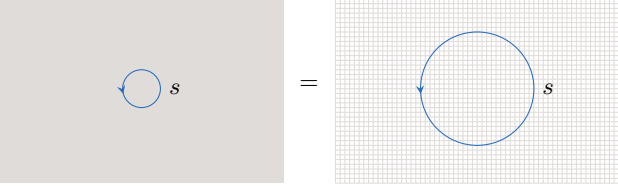

    \centering
    \diagram{2.5}{
        \fill[YaleMidGrey] (0,0) rectangle (1.5,1);
        \draw[YaleMidBlue,string] (.75,.5) circle[radius = .1];
        \node[right] at (.85,.5) {$s$};
        
    }
    =
    \diagram{2.5}{
        \foreach \i in {0,.025,...,1.5}
        {
            \draw[YaleMidGrey] (\i,0) -- (\i,1);
        }
        \foreach \j in {0,.025,...,1}
        {
            \draw[YaleMidGrey] (0,\j) -- (1.5,\j);
        }
        \draw[YaleMidBlue,string] (.75,.5) circle[radius = .3];
        \node[right] at (1.05,.5) {$s$};
    }
    \caption{A TDL is provided by an MPO in the continuum limit. Each crossing at the right-hand side represents a local tensor of the tensor network representation of the partition function.}
    \label{fig:TDL}
\end{figure}

\subsection{Dynamics}\label{sec:D8 lattices}
In this section, we present a lattice model protected by $\Rep^\dagger(D_8)$ symmetry. Specifically, we construct a $\Rep^\dagger(D_8)$-symmetric model with tuning parameters that realizes three distinct $\Rep^\dagger(D_8)$-SPTs at different points of the parameter space, with continuous phase transitions between these phases.

A generic $\Rep^\dagger(D_8)$-symmetric Hamiltonian is of the form
\[
    H = \sum_{X\subset \Z}\Phi(X),
\]
where each $X$ is a finite subset of $\Z$ in the form $\{m,m+1,\cdots, m+n\}$, where $m$ is a integer and $n$ is a non-negative integer, and $\Phi(X)$ is a $D_8$-grading preserving Hermitian operator:
\[
    \Phi(X):V_{m}\ot V_{m+1}\ot\cdots\ot V_{m+n}\rightarrow V_{m}\ot V_{m+1}\ot\cdots\ot V_{m+n}.
\]

In our model, the local Hilbert space on each site is chosen to be
\[
    \C[D_8] := \spann(|g\ra)_{g\in D_8}.
\]
Such a choice of local Hilbert space in our construction is non-canonical and is only for convenience to host three phases in the same model.

The total Hilbert space is
\[
\cH = \cdots \ot \C[D_8]\otimes \C[D_8]\otimes \cdots \otimes \C[D_8]\ot \cdots.
\]
In Section \ref{sec:O-type gs}, we give a more rigorous treatment of the infinite chain using operator algebra.

Here we emphasize that the tensor product order must obey the linear order of the lattice sites, as the tensor product $|g\ra\ot |h\ra$ is graded $gh$, but $|h\ra\ot |g\ra$ is graded $hg$, which does not always equal to $gh$ since $D_8$ is non-Abelian.

Define the projection maps
\begin{align*}
    P_0&:\C[D_8]\to A_0,\\
    P_1&:\C[D_8]\to A_1,\\
    P_2&:\C[D_8]\to A_2,
\end{align*}
with
\begin{align*}
    A_0 &= \C_e := \spann(|e\ra),\\
    A_1 &= \C\la s, r^2 \ra:=\spann(|e\ra,|r^2\ra,|s\ra,|sr^2\ra),\\
    A_2 &= \C\la sr, r^2 \ra:= \spann(|e\ra,|r^2\ra,|sr\ra,|sr^3\ra).
\end{align*}
On basis vectors, projections are explicitly defined by
\begin{align*}
    P_0|g\ra &= \delta_{g,e}|e\ra,\\
    P_1|g\ra &= \begin{cases}
        |g\ra & g \in \{e,r^2,s,sr^2\}\\
        0 & \text{otherwise}
    \end{cases},\\
    P_2|g\ra &= \begin{cases}
        |g\ra & g \in \{e,r^2,sr,sr^3\}\\
        0 & \text{otherwise}
    \end{cases}.
\end{align*}

We also introduce two rank-3 tensors which will be used to describe the nearest-neighbor interaction, they are the multiplication maps
\begin{align*}
    m_1:A_1\ot A_1\rightarrow A_1,\quad (m_1)^i_{jk} = 
    \diagram{2}{
        \draw[string] (-.5,-.5) node[below]{$j$} -- (0,0);
        \draw[string] (.5,-.5) node[below]{$k$} -- (0,0);
        \draw[string] (0,0) -- (0,.5) node[above]{$i$};
        \node[draw = black,fill = white] at (0,0) {$m_1$};
    },\\
    m_2:A_2\ot A_2\rightarrow A_2,\quad (m_2)^i_{jk} = 
    \diagram{2}{
        \draw[string] (-.5,-.5) node[below]{$j$} -- (0,0);
        \draw[string] (.5,-.5) node[below]{$k$} -- (0,0);
        \draw[string] (0,0) -- (0,.5) node[above]{$i$};
        \node[draw = black,fill = white] at (0,0) {$m_2$};
    }.
\end{align*}
with
\begin{align}\label{eq:multiplication}
    (m_1)^g_{hk} &= \frac{1}{2}\delta_{g,hk}\omega(h,k),\quad \forall g,h,k\in \la r^2,s\ra,\\
    (m_2)^g_{hk} &= \frac{1}{2}\delta_{g,hk}\omega(h,k),\quad \forall g,h,k\in \la r^2,sr\ra.
\end{align}
The $\omega$ takes values in Table \ref{tab:omega}.
\begin{table}[htbp]
    \centering
    \begin{tabular}{|c|c|c|c|c|}
    \hline
         &  $e$ & $a$ & $b$ & $ab$\\
         \hline
        $e$ & $1$ & $1$ & $1$ & $1$\\
        \hline
        $a$ & $1$ & $1$ & $1$ & $1$\\
        \hline
        $b$ & $1$ & $-1$ & $1$ & $-1$\\
        \hline
        $ab$ & $1$ & $-1$ & $1$ & $-1$\\
        \hline
    \end{tabular}
    \caption{$\omega(g,h)$ for $g,h \in\Z_2\times \Z_2 = \la a, b|a^2 = b^2 = e,\ ab = ba\ra$. For $\la r^2,s\ra$, $a = r^2$ and $b = s$, for $\la r^2,sr\ra$, $a = r^2$ and $b = sr$.}
    \label{tab:omega}
\end{table}
Under this multiplication map, $A_1$, $A_2$ are isomorphic to the matrix algebra $\mathbb{M}_2$. 

The $\Rep^\dagger(D_8)$-symmetric model is
\begin{equation*}
    \begin{split}
        H = &-\lambda_0\sum_{i}P_0^{(i)} \\
         &- \lambda_1 \sum_{i}\left(m_{1}^{(i,i+1)}P_1^{(i)}P_1^{(i+1)}\right)^\dagger m_1^{(i,i+1)}P_1^{(i)}P_1^{(i+1)}\\
         &- \lambda_2\sum_{i}\left(m_2^{(i,i+1)}P_2^{(i)}P_2^{(i+1)}\right)^\dagger m_2^{(i,i+1)}P_2^{(i)}P_2^{(i+1)}\,,
    \end{split}
\end{equation*}
where $A_\bullet$ is locally embedded into the eight-dimensional Hilbert space $\C[D_8]$. The model is $\Rep^\dagger(D_8)$-symmetric as all the terms are $D_8$-grading-preserving maps. $\lambda_0,\lambda_1,\lambda_2$ are three tuning parameters. When $(\lambda_0,\lambda_1,\lambda_2) = (1,0,0),\  (0,1,0),\  (0,0,1)$ respectively, the models are exactly solvable. 

\subsubsection{Phase 0: trivial phase}
When $(\lambda_0,\lambda_1,\lambda_2) = (1,0,0)$, the Hamiltonian reduces to
\[
    H = -\sum_{i}P^{(i)}_1,
\]
the ground state of which is the product state
\[
    |\psi\ra = |\cdots ee\cdots e\cdots \ra.
\]
And the value of local tensor of the $\Rep^\dagger(D_8)$-MPO are
\[
    (T_{\one})_{ee} = (T_{a})_{ee} = (T_{b})_{ee} = (T_{c})_{ee} = 1,\quad (T_{\sigma})_{ee} = \one.
\]
If the lattice is defined on $S^1$, the eigenvalues of $\Rep^\dagger(D_8)$-MPOs on ground state are equal to their quantum dimensions.

\subsubsection{Phase 1: \texorpdfstring{$W_{(a,c)} \ot W^*_{(a,c)}$}{Wac x Wac*}}
When $(\lambda_0,\lambda_1,\lambda_2) = (0,1,0)$, the Hamiltonian is defined on a the low energy effective Hilbert space
\[
    \cH^{(1)}_{\text{eff}}:=\cdots \otimes A_1\ot A_1\ot \cdots \ot A_1\otimes \cdots,
\]
and the effective Hamiltonian is
\begin{equation}\label{eq:effective Hamiltonian 1}
    H^{(1)}_{\text{eff}} = -\sum_{i}\left(m_{1}^{(i,i+1)}\right)^\dagger m_1^{(i,i+1)}.
\end{equation}

This Hamiltonian can be further simplified by dividing a single site $i$ into two sublattice sites $(2i-1,2i)$ which forms the unit cell in the new lattice. The local Hilbert space is then decomposed into two ``projective'' charge spaces:

\begin{alignat*}{2}
    &|e\ra = \frac{1}{\sqrt{2}}(|00'\ra + |11'\ra),\quad &&|s\ra = \frac{1}{\sqrt{2}}(|01'\ra + |10'\ra),\\
    &|r^2\ra = \frac{1}{\sqrt{2}}(|00'\ra - |11'\ra),\quad &&|sr^2\ra = \frac{1}{\sqrt{2}}(|01'\ra - |10'\ra),
\end{alignat*}
such that $\C[\la r^2,s\ra] =W_{(a,c)} \otimes W_{(a,c)}^*$ with $W_{(a,c)} = \spann(|0\ra,\,|1\ra )$ and $W_{(a,c)}^* = \spann(|0'\ra,\,|1'\ra )$, where $\{|0'\ra,\,|1'\ra\}$ is the dual basis of $\{|0\ra,\,|1\ra\}$. The reason for the subscript in $W_{(a,c)}$ will become clear shortly.
\[
    \diagram{1}{
        \foreach \i in {0,1,...,5}
        {
        \fill (\i*.5,0) circle [radius = .05];
        }
    }
    \longrightarrow 
    \diagram{1}{
        \foreach \i in {0,1,...,5}
        {
        \fill[YaleLightBlue] (\i*.5,0) ellipse [x radius=.2,y radius=0.15];
        \fill[gray] (\i*.5 -.1,0) circle [radius = .05];
        \fill[gray] (\i*.5 +.1,0) circle [radius = .05];
        }
    }
\]
Under such a basis, $\Rep^\dagger(D_8)$ MPOs are correspondingly factorized into two ``projective" MPOs
\[
    (T_s)_{(ab)(cd)} = (L_s)_{ac} (R_s)_{bd},
\]
graphically
\[
    \diagram{1}{
        \draw[gray] (0-.1,0) node[below]{$c$} -- (0,.1) -- (.1,0) node[below]{$d$};
        \draw[gray] (-.1,2) node[above]{$a$} -- (0,1.9) -- (.1,2) node[above]{$b$};
        \draw (0,.1) -- (0,1.9);
        \draw[YaleMidBlue] (1,1) -- (-1,1);
        \filldraw[fill = white] (-.2,.8) rectangle node {$T_s$} (.2,1.2);
    }
    =
    \diagram{1}{
        \draw[gray] (0,0) node[below]{$c$} -- (0,2) node[above]{$a$};
        \draw[gray] (1,0) node[below]{$d$} -- (1,2) node[above]{$b$};
        \draw[YaleMidBlue] (2,1) -- (-1,1);
        \filldraw[fill = white] (-.2,.8) rectangle node{$L_s$} (.2,1.2);
        \filldraw[fill = white] (1-.2,.8) rectangle node{$R_s$} (1.2,1.2);
    }.
\]
More concretely,
invertible symmetry operators can be factorized as
\begin{gather*}
    (L_e)^{11} = (R_e)^{11} =(L_b)^{11} = (R_b)^{11}= \one,\\
    (L_a)^{11} = (R_a)^{11} =(L_c)^{11} = (R_c)^{11}= Z.
\end{gather*}
and non-zero tensor elements for non-invertible symmetry $\sigma$ are
\begin{alignat*}{2}
    &(L_\sigma)_{01} = \frac{1}{\sqrt{2}}(\one - \ii X),\quad &&(R_\sigma)_{0'1'} = \frac{1}{\sqrt{2}}(\one + \ii X),\\
    &(L_\sigma)_{10} = \frac{1}{\sqrt{2}}(-\ii \one + X),\quad && (R_\sigma)_{1'0'}=\frac{1}{\sqrt{2}}(\ii\one + X).
\end{alignat*}

Let
\begin{equation}\label{eq:maximally entangled}
    |\psi\ra = \frac{1}{\sqrt{2}}\Big(|0'0\ra + |1'1\ra\Big) =\frac{1}{\sqrt{2}}\times
    \diagram{.7}{
        \draw[gray,string] (0,1) -- (0,0) -- (1,0) -- (1,1);
    }
    \in W_{(a,c)}^*\ot W_{(a,c)}
\end{equation}
be a maximally entangled state,
the effective Hamiltonian \eqref{eq:effective Hamiltonian 1} can then be simplified to
\[
    H^{(1)}_{\text{eff}} = -\sum_{i}|\psi\ra\la \psi|_{(2i-1,2i)}.
\]
It can be directly checked that each local term in the Hamiltonian is symmetric by
\[
    \diagram{1}{
        \draw[gray,string] (0,1) -- (0,0) -- (1,0) -- (1,1);
        \draw[YaleMidBlue,string] (1.5,.5) -- (-.5,.5);
        \filldraw[fill = white] (-.2,.3) rectangle node{$R_s$} (.2,.7);
        \filldraw[fill = white] (.8,.3) rectangle node{$L_s$} (1.2,.7);
    }
    =
    \diagram{1}{
        \draw[gray,string] (0,1) -- (0,0) -- (1,0) -- (1,1);
        \draw[YaleMidBlue,string] (1.5,-.2) -- (-.5,-.2);
    }.
\]

On a closed chain with $N$-sites, where site $N+1$ and site $1$ are identified, and the Hamiltonian is in the similar form.
The unique ground state on an infinite chain is similar to the standard injective MPS:
\[
    |\Psi\ra = \otimes_{i}|\psi\ra_{(2i-1,2i)} \propto 
    \diagram{.7}{
        \foreach \i in {1,2,...,5}
        {
        \draw[gray,string] (\i+.1,.8) -- (\i+.1,0) -- (\i+.9,0) -- (\i+.9,.8);
        }
        \node at (.5,0) {$\cdots$};
        \node at (6.5,0) {$\cdots$};
    }.
\]
When defined on an open chain with free boundary condition, 
two edge states of which are $2$-fold degenerate. 
\[
    \diagram{1}{
        \foreach \i in {0,1,...,5}
        {
        \fill[YaleLightBlue] (\i*.5,0) ellipse [x radius=.2,y radius=0.15];
        \fill[gray] (\i*.5 -.1,0) circle [radius = .05];
        \fill[gray] (\i*.5 +.1,0) circle [radius = .05];
        }
        \foreach \j in {0,1,...,4}
        {
        \draw (\j*.5+.1,0) -- (\j*.5+.4,0);
        }
    }
\]
We can see the projective algebra on the edge state by computing the compositions of factorized MPOs $L_s$ or $R_s$. As $T_s$ is decomposed into two parts, we need to be careful with the labeling of the \textit{internal} virtual bonds, i.e. the virtual bonds that are pulled out during the factorization process. We denote the internal bonds of the projective MPOs as
\[
    \begin{tikzpicture}[baseline=(current bounding box),scale=.6]
        \draw[gray, string] (0,0) -- (0,.7);
        \draw[gray, string] (0,1.3) -- (0,2);
        \draw[YaleMidBlue,dashed, string] (-.3,1)  -- node[above]{$e$}  (-1,1);
        \draw[YaleMidBlue,dashed, string] (1,1) -- node[above]{$e$} (.3,1) ;
        \filldraw[fill = white] (-.3,.7) rectangle node {$L_e$} (.3,1.3);
      \end{tikzpicture},\ 
    \begin{tikzpicture}[baseline=(current bounding box),scale=.6]
        \draw[gray, string] (0,0) -- (0,.7);
        \draw[gray, string] (0,1.3) -- (0,2);
        \draw[YaleMidBlue,dashed, string] (-.3,1)  -- node[above]{$a$}  (-1,1);
        \draw[YaleMidBlue, dashed,string] (1,1) -- node[above]{$c$} (.3,1) ;
        \filldraw[fill = white] (-.3,.7) rectangle node {$L_a$} (.3,1.3);
    \end{tikzpicture},\ 
    \begin{tikzpicture}[baseline=(current bounding box),scale=.6]
        \draw[gray, string] (0,0) -- (0,.7);
        \draw[gray, string] (0,1.3) -- (0,2);
        \draw[YaleMidBlue,dashed, string] (-.3,1)  -- node[above]{$b$}  (-1,1);
        \draw[YaleMidBlue,dashed, string] (1,1) -- node[above]{$b$} (.3,1) ;
        \filldraw[fill = white] (-.3,.7) rectangle node {$L_b$} (.3,1.3);
    \end{tikzpicture},\ 
    \begin{tikzpicture}[baseline=(current bounding box),scale=.6]
        \draw[gray, string] (0,0) -- (0,.7);
        \draw[gray, string] (0,1.3) -- (0,2);
        \draw[YaleMidBlue,dashed, string] (-.3,1)  -- node[above]{$c$}  (-1,1);
        \draw[YaleMidBlue,dashed, string] (1,1) -- node[above]{$a$} (.3,1) ;
        \filldraw[fill = white] (-.3,.7) rectangle node {$L_c$} (.3,1.3);
    \end{tikzpicture},\  
    \begin{tikzpicture}[baseline=(current bounding box),scale=.6]
        \draw[gray, string] (0,0) -- (0,.7);
        \draw[gray, string] (0,1.3) -- (0,2);
        \draw[YaleMidBlue, string] (-.3,1)  -- node[above]{$\sigma$}  (-1,1);
        \draw[YaleMidBlue, string] (1,1) -- node[above]{$\sigma$} (.3,1) ;
        \filldraw[fill = white] (-.3,.7) rectangle node {$L_\sigma$} (.3,1.3);
    \end{tikzpicture}
\]

The invertible ones satisfy the group multiplication rule. The compositions between the invertible ones and the non-invertible one are
\begin{align*}
    \begin{tikzpicture}[baseline=(current bounding box),scale=.4]
        \draw[gray, string] (0,0) -- (0,.7);
        \draw[gray, string] (0,1.3) -- (0,2);
        \draw[YaleMidBlue, string] (-.3,1)  -- node[above]{$\sigma$} (-2,1);
        \draw[YaleMidBlue, string] (2,1) -- node[above]{$\sigma$} (.3,1);
        \filldraw[fill = white] (-.3,.7) rectangle node[font = \tiny] {$L_\sigma$} (.3,1.3);
        \draw[gray, string] (0,2) -- (0,2.7);
        \draw[gray, string] (0,3.3) -- (0,4);
        \draw[YaleMidBlue,dashed, string] (-.3,3)  -- node[above]{$a$} (-2,3);
        \draw[YaleMidBlue,dashed, string] (2,3) -- node[above]{$c$} (.3,3);
        \filldraw[fill = white] (-.3,2.7) rectangle node[font = \tiny] {$L_a$} (.3,3.3);
  \end{tikzpicture}
  &=
    \begin{tikzpicture}[baseline=(current bounding box),scale=1.5]
        \draw[gray, string] (0,1) -- (0,1.7);
        \draw[gray, string] (0,2.3) -- (0,3);
        \draw[YaleMidBlue,string] (1,2) -- node[above]{$\sigma$} (0,2);
        \draw[YaleMidBlue,string] (0,2) --node[above]{$\sigma$} (-1,2);
        \draw[YaleMidBlue,dashed,string] (-1,2) -- node[above]{$a$} (-2,3);
        \draw[YaleMidBlue,string] (-1,2) -- node[below]{$\sigma$} (-2,1);
        \draw[YaleMidBlue,dashed,string] (2,3) -- node[above]{$c$} (1,2);
        \draw[YaleMidBlue,string] (2,1) -- node[below]{$\sigma$} (1,2);
        \filldraw[fill = white] (-.3,1.7) rectangle node {$L_\sigma$} (.3,2.3);
        \filldraw[fill = white] (-1.2,1.8) rectangle node {\small $Z$} (-.8,2.2);
        \filldraw[fill = white] (.8,1.8) rectangle node {\small $Y$} (1.2,2.2);
    \end{tikzpicture},\\
    \begin{tikzpicture}[baseline=(current bounding box),scale=.4]
        \draw[gray, string] (0,0) -- (0,.7);
        \draw[gray, string] (0,1.3) -- (0,2);
        \draw[YaleMidBlue, string] (-.3,1)  -- node[above]{$\sigma$} (-2,1);
        \draw[YaleMidBlue, string] (2,1) -- node[above]{$\sigma$} (.3,1) ;
        \filldraw[fill = white] (-.3,.7) rectangle node[font = \tiny] {$L_\sigma$} (.3,1.3);
        \draw[gray, string] (0,2) -- (0,2.7);
        \draw[gray, string] (0,3.3) -- (0,4);
        \draw[YaleMidBlue,dashed, string] (-.3,3)  -- node[above]{$b$} (-2,3);
        \draw[YaleMidBlue,dashed, string] (2,3) -- node[above]{$b$} (.3,3);
        \filldraw[fill = white] (-.3,2.7) rectangle node[font = \tiny] {$L_b$} (.3,3.3);
    \end{tikzpicture}
  &=
    \begin{tikzpicture}[baseline=(current bounding box),scale=1.5]
        \draw[gray, string] (0,1) -- (0,1.7);
        \draw[gray, string] (0,2.3) -- (0,3);
        \draw[YaleMidBlue,string] (1,2) -- node[above]{$\sigma$} (0,2);
        \draw[YaleMidBlue,string] (0,2) --node[above]{$\sigma$} (-1,2);
        \draw[YaleMidBlue,dashed,string] (-1,2) -- node[above]{$b$} (-2,3);
        \draw[YaleMidBlue,string] (-1,2) -- node[below]{$\sigma$} (-2,1);
        \draw[YaleMidBlue,dashed,string] (2,3) -- node[above]{$b$} (1,2);
        \draw[YaleMidBlue,string] (2,1) -- node[below]{$\sigma$} (1,2);
        \filldraw[fill = white] (-.3,1.7) rectangle node {$L_\sigma$} (.3,2.3);
        \filldraw[fill = white] (-1.2,1.8) rectangle node {\small $X$} (-.8,2.2);
        \filldraw[fill = white] (.8,1.8) rectangle node {\small $X$} (1.2,2.2);
    \end{tikzpicture},\\
    \begin{tikzpicture}[baseline=(current bounding box),scale=.4]
        \draw[gray, string] (0,0) -- (0,.7);
        \draw[gray, string] (0,1.3) -- (0,2);
        \draw[YaleMidBlue, dashed,string] (-.3,1)  -- node[above]{$a$} (-2,1) ;
        \draw[YaleMidBlue,dashed, string] (2,1) -- node[above]{$c$} (.3,1) ;
        \filldraw[fill = white] (-.3,.7) rectangle node[font = \tiny] {$L_a$} (.3,1.3);
        \draw[gray, string] (0,2) -- (0,2.7);
        \draw[gray, string] (0,3.3) -- (0,4);
        \draw[YaleMidBlue, string] (-.3,3)  -- node[above]{$\sigma$} (-2,3);
        \draw[YaleMidBlue, string] (2,3) -- node[above]{$\sigma$} (.3,3) ;
        \filldraw[fill = white] (-.3,2.7) rectangle node[font = \tiny] {$L_\sigma$} (.3,3.3);
    \end{tikzpicture}
  &=
    \begin{tikzpicture}[baseline=(current bounding box),scale=1.5]
        \draw[gray, string] (0,1) -- (0,1.7);
        \draw[gray, string] (0,2.3) -- (0,3);
        \draw[YaleMidBlue,string] (1,2) -- node[above]{$\sigma$} (0,2);
        \draw[YaleMidBlue,string] (0,2) --node[above]{$\sigma$} (-1,2);
        \draw[YaleMidBlue,string] (-1,2) -- node[above]{$\sigma$} (-2,3);
        \draw[YaleMidBlue,dashed,string] (-1,2) -- node[below]{$a$} (-2,1);
        \draw[YaleMidBlue,string] (2,3) -- node[above]{$\sigma$} (1,2);
        \draw[YaleMidBlue,dashed,string] (2,1) -- node[below]{$c$} (1,2);
        \filldraw[fill = white] (-.3,1.7) rectangle node {$L_\sigma$} (.3,2.3);
        \filldraw[fill = white] (-1.2,1.8) rectangle node {\small $Z$} (-.8,2.2);
        \filldraw[fill = white] (.8,1.8) rectangle node {\small $-Y$} (1.2,2.2);
    \end{tikzpicture},\\
    \begin{tikzpicture}[baseline=(current bounding box),scale=.4]
        \draw[gray, string] (0,0) -- (0,.7);
        \draw[gray, string] (0,1.3) -- (0,2);
        \draw[YaleMidBlue, dashed,string] (-.3,1)  -- node[above]{$b$} (-2,1);
        \draw[YaleMidBlue, dashed,string] (2,1) -- node[above]{$b$} (.3,1) ;
        \filldraw[fill = white] (-.3,.7) rectangle node[font = \tiny] {$L_b$} (.3,1.3);
        \draw[gray, string] (0,2) -- (0,2.7);
        \draw[gray, string] (0,3.3) -- (0,4);
        \draw[YaleMidBlue, string] (-.3,3)  -- node[above]{$\sigma$} (-2,3);
        \draw[YaleMidBlue, string] (2,3) -- node[above]{$\sigma$} (.3,3);
        \filldraw[fill = white] (-.3,2.7) rectangle node[font = \tiny] {$L_\sigma$} (.3,3.3);
    \end{tikzpicture}
  &=
    \begin{tikzpicture}[baseline=(current bounding box),scale=1.5]
        \draw[gray, string] (0,1) -- (0,1.7);
        \draw[gray, string] (0,2.3) -- (0,3);
        \draw[YaleMidBlue,string] (1,2) -- node[above]{$\sigma$} (0,2);
        \draw[YaleMidBlue,string] (0,2) --node[above]{$\sigma$} (-1,2);
        \draw[YaleMidBlue,string] (-1,2) -- node[above]{$\sigma$} (-2,3);
        \draw[YaleMidBlue,dashed,string] (-1,2) -- node[below]{$b$} (-2,1);
        \draw[YaleMidBlue,string] (2,3) -- node[above]{$\sigma$} (1,2);
        \draw[YaleMidBlue,dashed,string] (2,1) -- node[below]{$b$} (1,2);
        \filldraw[fill = white] (-.3,1.7) rectangle node {$L_\sigma$} (.3,2.3);
        \filldraw[fill = white] (-1.2,1.8) rectangle node {\small $X$} (-.8,2.2);
        \filldraw[fill = white] (.8,1.8) rectangle node {\small $X$} (1.2,2.2);
    \end{tikzpicture}.
    \end{align*}
And the non-invertible tensor composes with itself as
\begin{align*}
    \begin{tikzpicture}[baseline=(current bounding box),scale=.5]
        \draw[gray, string] (0,0) -- (0,.8);
        \draw[gray, string] (0,1.2) -- (0,2);
        \draw[YaleMidBlue, string] (-.3,1)  -- node[above]{$\sigma$} (-2,1);
        \draw[YaleMidBlue, string] (2,1) -- node[above]{$\sigma$} (.3,1);
        \filldraw[fill = white] (-.3,.7) rectangle node[font = \tiny] {$L_\sigma$} (.3,1.3);
        \draw[gray, string] (0,2) -- (0,2.7);
        \draw[gray, string] (0,3.3) -- (0,4);
        \draw[YaleMidBlue, string] (-.3,3)  -- node[above]{$\sigma$} (-2,3);
        \draw[YaleMidBlue, string] (2,3) -- node[above]{$\sigma$} (.3,3);
        \filldraw[fill = white] (-.3,2.7) rectangle node[font = \tiny] {$L_a$} (.3,3.3);
  \end{tikzpicture}
  &=
    \begin{tikzpicture}[baseline=(current bounding box),scale=1.5]
        \draw[gray, string] (0,1) -- (0,1.8);
        \draw[gray, string] (0,2.2) -- (0,3);
        \draw[YaleMidBlue,dashed,string] (1,2) -- node[above]{$e$} (0,2);
        \draw[YaleMidBlue,dashed,string] (0,2) --node[above]{$e$} (-1,2);
        \draw[YaleMidBlue,string] (-1,2) -- node[above]{$\sigma$} (-2,3);
        \draw[YaleMidBlue,string] (-1,2) -- node[below]{$\sigma$} (-2,1);
        \draw[YaleMidBlue,string] (2,3) -- node[above]{$\sigma$} (1,2);
        \draw[YaleMidBlue,string] (2,1) -- node[below]{$\sigma$} (1,2);
        \filldraw[fill = white] (-.2,1.8) rectangle node {$L_e$} (.2,2.2);
        \filldraw[fill = white] (-1.2,1.8) rectangle node {$\psi_e$} (-.8,2.2);
        \node[draw = black,fill = white] at (1,2){$-\ii\psi_e^\dagger$};
    \end{tikzpicture}\\
      &+
    \begin{tikzpicture}[baseline=(current bounding box),scale=1.5]
        \draw[gray, string] (0,1) -- (0,1.8);
        \draw[gray, string] (0,2.2) -- (0,3);
        \draw[YaleMidBlue,dashed,string] (1,2) -- node[above]{$c$} (0,2);
        \draw[YaleMidBlue,dashed,string] (0,2) --node[above]{$a$} (-1,2);
        \draw[YaleMidBlue,string] (-1,2) -- node[above]{$\sigma$} (-2,3);
        \draw[YaleMidBlue,string] (-1,2) -- node[below]{$\sigma$} (-2,1);
        \draw[YaleMidBlue,string] (2,3) -- node[above]{$\sigma$} (1,2);
        \draw[YaleMidBlue,string] (2,1) -- node[below]{$\sigma$} (1,2);
        \filldraw[fill = white] (-.2,1.8) rectangle node {$L_a$} (.2,2.2);
        \filldraw[fill = white] (-1.2,1.8) rectangle node {$\psi_a$} (-.8,2.2);
        \node[draw = black,fill = white] at (1,2) {$-\psi_c^\dagger$};
    \end{tikzpicture}\\
      &+
    \begin{tikzpicture}[baseline=(current bounding box),scale=1.5]
        \draw[gray, string] (0,1) -- (0,1.8);
        \draw[gray, string] (0,2.2) -- (0,3);
        \draw[YaleMidBlue,dashed,string] (1,2) -- node[above]{$b$} (0,2);
        \draw[YaleMidBlue,dashed,string] (0,2) --node[above]{$b$} (-1,2);
        \draw[YaleMidBlue,string] (-1,2) -- node[above]{$\sigma$} (-2,3);
        \draw[YaleMidBlue,string] (-1,2) -- node[below]{$\sigma$} (-2,1);
        \draw[YaleMidBlue,string] (2,3) -- node[above]{$\sigma$} (1,2);
        \draw[YaleMidBlue,string] (2,1) -- node[below]{$\sigma$} (1,2);
        \filldraw[fill = white] (-.2,1.8) rectangle node {$L_b$} (.2,2.2);
        \filldraw[fill = white] (-1.2,1.8) rectangle node {$\psi_b$} (-.8,2.2);
        \node[draw = black,fill = white] at (1,2) {$-\ii\psi_b^\dagger$} (1.2,2.2);
    \end{tikzpicture}\\
      &+
    \begin{tikzpicture}[baseline=(current bounding box),scale=1.5]
        \draw[gray, string] (0,1) -- (0,1.8);
        \draw[gray, string] (0,2.2) -- (0,3);
        \draw[YaleMidBlue,dashed,string] (1,2) -- node[above]{$a$} (0,2);
        \draw[YaleMidBlue,dashed,string] (0,2) --node[above]{$c$} (-1,2);
        \draw[YaleMidBlue,string] (-1,2) -- node[above]{$\sigma$} (-2,3);
        \draw[YaleMidBlue,string] (-1,2) -- node[below]{$\sigma$} (-2,1);
        \draw[YaleMidBlue,string] (2,3) -- node[above]{$\sigma$} (1,2);
        \draw[YaleMidBlue,string] (2,1) -- node[below]{$\sigma$} (1,2);
        \filldraw[fill = white] (-.2,1.8) rectangle node {$L_c$} (.2,2.2);
        \filldraw[fill = white] (-1.2,1.8) rectangle node {$\psi_c$} (-.8,2.2);
        \node[draw = black,fill = white] at (1,2) {$-\psi_a^\dagger$};
    \end{tikzpicture}.
\end{align*}
We can see that when MPO passes through the projective charge, types of the virtual bond $a$ and $c$ are swapped, which we denote by $(a,c)$.

\subsubsection{Phase 2: \texorpdfstring{$W_{(a,b)} \ot W^*_{(a,b)}$}{Wab x Wab*}}
When $(\lambda_0,\lambda_1,\lambda_2) = (0,0,1)$, the Hamiltonian is effectively defined on the low energy effective Hilbert space
\[
    \cH^{(2)}_{\text{eff}}:=\cdots \otimes A_2\ot A_2\ot \cdots \ot A_2\otimes \cdots,
\]
and the effective Hamiltonian is
\begin{equation}\label{eq:effective Hamiltonian 2}
    H^{(2)}_{\text{eff}} = -\sum_{i}\left(m_{2}^{(i,i+1)}\right)^\dagger m_2^{(i,i+1)}.
\end{equation}
Similarly, we separate the local Hilbert space into projective charge spaces
\begin{alignat}{2}\label{eq:phase2 fractionalization}
    &|e\ra = \frac{1}{\sqrt{2}}\left(|00'\ra + |11'\ra\right), \quad &&|r^2\ra = \frac{1}{\sqrt{2}}\left(|00'\ra - |11'\ra\right),\\
    &|sr\ra =\frac{1}{\sqrt{2}}\left( |01'\ra + |10'\ra\right),\quad && |sr^3\ra = \frac{1}{\sqrt{2}}\left(|01'\ra - |10'\ra\right),
\end{alignat}
such that $\C[\la r^2,sr\ra] =W_{(a,b)} \otimes W_{(a,b)}^*$ with $W_{(a,b)} = \spann(|0\ra,\,|1\ra )$ and $W_{(a,b)}^* = \spann(|0'\ra,\,|1'\ra )$, where $\{|0'\ra,\,|1'\ra \}$ is the dual basis of $\{|0'\ra,\,|1'\ra \}$.
\[
    \diagram{1}{
        \foreach \i in {0,1,...,5}
        {
            \fill (\i*.5,0) circle [radius = .05];
        }
    }
    \longrightarrow 
    \diagram{1}{
        \foreach \i in {0,1,...,5}
        {
            \fill[YaleLightBlue] (\i*.5,0) ellipse [x radius=.2,y radius=0.15];
            \fill[gray] (\i*.5 -.1,0) circle [radius = .05];
            \fill[gray] (\i*.5 +.1,0) circle [radius = .05];
        }
    }
\]
More concretely,
invertible symmetry operators can be factorized as
\begin{gather*}
    (L_e)^{11} = (R_e)^{11} =(L_c)^{11} = (R_c)^{11}= \one,\\
    (L_a)^{11} = (R_a)^{11} =(L_b)^{11} = (R_b)^{11}= Z.
\end{gather*}
and non-zero tensor elements for non-invertible symmetry $\sigma$ are
\begin{alignat*}{2}
    &(L_\sigma)_{01} = \frac{1}{\sqrt{2}}(\one + \ii Y),\quad &&(R_\sigma)_{0'1'} = \frac{1}{\sqrt{2}}(\one - \ii Y),\\
    &(L_\sigma)_{10} = \frac{1}{\sqrt{2}}(\ii \one  + Y),\quad && (R_\sigma)_{1'0'}=\frac{1}{\sqrt{2}}(-\ii\one + Y).
\end{alignat*}

Similarly, we denote the \textit{internal} virtual bonds of local tensors on projective charges as
\[
    \begin{tikzpicture}[baseline=(current bounding box),scale=.6]
        \draw[gray, string] (0,0) -- (0,.7);
        \draw[gray, string] (0,1.3) -- (0,2);
        \draw[YaleMidBlue,dashed, string] (-.3,1)  -- node[above]{$e$}  (-1,1);
        \draw[YaleMidBlue,dashed, string] (1,1) -- node[above]{$e$} (.3,1) ;
        \filldraw[fill = white] (-.3,.7) rectangle node {$L_e$} (.3,1.3);
    \end{tikzpicture},\ 
    \begin{tikzpicture}[baseline=(current bounding box),scale=.6]
        \draw[gray, string] (0,0) -- (0,.7);
        \draw[gray, string] (0,1.3) -- (0,2);
        \draw[YaleMidBlue,dashed, string] (-.3,1)  -- node[above]{$a$}  (-1,1);
        \draw[YaleMidBlue, dashed,string] (1,1) -- node[above]{$b$} (.3,1) ;
        \filldraw[fill = white] (-.3,.7) rectangle node {$L_a$} (.3,1.3);
    \end{tikzpicture},\ 
    \begin{tikzpicture}[baseline=(current bounding box),scale=.6]
        \draw[gray, string] (0,0) -- (0,.7);
        \draw[gray, string] (0,1.3) -- (0,2);
        \draw[YaleMidBlue,dashed, string] (-.3,1)  -- node[above]{$b$}  (-1,1);
        \draw[YaleMidBlue,dashed, string] (1,1) -- node[above]{$a$} (.3,1) ;
        \filldraw[fill = white] (-.3,.7) rectangle node {$L_b$} (.3,1.3);
    \end{tikzpicture},\ 
    \begin{tikzpicture}[baseline=(current bounding box),scale=.6]
        \draw[gray, string] (0,0) -- (0,.7);
        \draw[gray, string] (0,1.3) -- (0,2);
        \draw[YaleMidBlue,dashed, string] (-.3,1)  -- node[above]{$c$}  (-1,1);
        \draw[YaleMidBlue,dashed, string] (1,1) -- node[above]{$c$} (.3,1) ;
        \filldraw[fill = white] (-.3,.7) rectangle node {$L_c$} (.3,1.3);
    \end{tikzpicture},\  
    \begin{tikzpicture}[baseline=(current bounding box),scale=.6]
        \draw[gray, string] (0,0) -- (0,.7);
        \draw[gray, string] (0,1.3) -- (0,2);
        \draw[YaleMidBlue, string] (-.3,1)  -- node[above]{$\sigma$}  (-1,1);
        \draw[YaleMidBlue, string] (1,1) -- node[above]{$\sigma$} (.3,1) ;
        \filldraw[fill = white] (-.3,.7) rectangle node {$L_\sigma$} (.3,1.3);
    \end{tikzpicture}
\]
for later consideration.

We compute the composition rules for $\Rep^\dagger(D_8)$-MPOs on the projective charges. The compositions of invertible ones satisfy the group multiplication rule while those between the invertible ones and the non-invertible one are
\begin{align*}
    \begin{tikzpicture}[baseline=(current bounding box),scale=.4]
        \draw[gray, string] (0,0) -- (0,.7);
        \draw[gray, string] (0,1.3) -- (0,2);
        \draw[YaleMidBlue, string] (-.3,1)  -- node[above]{$\sigma$} (-2,1);
        \draw[YaleMidBlue, string] (2,1) -- node[above]{$\sigma$} (.3,1);
        \filldraw[fill = white] (-.3,.7) rectangle node[font = \tiny] {$L_\sigma$} (.3,1.3);
        \draw[gray, string] (0,2) -- (0,2.7);
        \draw[gray, string] (0,3.3) -- (0,4);
        \draw[YaleMidBlue,dashed, string] (-.3,3)  -- node[above]{$a$} (-2,3);
        \draw[YaleMidBlue,dashed, string] (2,3) -- node[above]{$b$} (.3,3);
        \filldraw[fill = white] (-.3,2.7) rectangle node[font = \tiny] {$L_a$} (.3,3.3);
    \end{tikzpicture}
  &=
    \begin{tikzpicture}[baseline=(current bounding box),scale=1.5]
        \draw[gray, string] (0,1) -- (0,1.7);
        \draw[gray, string] (0,2.3) -- (0,3);
        \draw[YaleMidBlue,string] (1,2) -- node[above]{$\sigma$} (0,2);
        \draw[YaleMidBlue,string] (0,2) --node[above]{$\sigma$} (-1,2);
        \draw[YaleMidBlue,dashed,string] (-1,2) -- node[above]{$a$} (-2,3);
        \draw[YaleMidBlue,string] (-1,2) -- node[below]{$\sigma$} (-2,1);
        \draw[YaleMidBlue,dashed,string] (2,3) -- node[above]{$b$} (1,2);
        \draw[YaleMidBlue,string] (2,1) -- node[below]{$\sigma$} (1,2);
        \filldraw[fill = white] (-.3,1.7) rectangle node {$L_\sigma$} (.3,2.3);
        \filldraw[fill = white] (-1.2,1.8) rectangle node {$Z$} (-.8,2.2);
        \filldraw[fill = white] (.8,1.8) rectangle node {$X$} (1.2,2.2);
    \end{tikzpicture},\\
    \begin{tikzpicture}[baseline=(current bounding box),scale=.4]
        \draw[gray, string] (0,0) -- (0,.7);
        \draw[gray, string] (0,1.3) -- (0,2);
        \draw[YaleMidBlue, string] (-.3,1)  -- node[above]{$\sigma$} (-2,1);
        \draw[YaleMidBlue, string] (2,1) -- node[above]{$\sigma$} (.3,1) ;
        \filldraw[fill = white] (-.3,.7) rectangle node[font = \tiny] {$L_\sigma$} (.3,1.3);
        \draw[gray, string] (0,2) -- (0,2.7);
            \draw[gray, string] (0,3.3) -- (0,4);
        \draw[YaleMidBlue,dashed, string] (-.3,3)  -- node[above]{$b$} (-2,3);
        \draw[YaleMidBlue,dashed, string] (2,3) -- node[above]{$a$} (.3,3);
        \filldraw[fill = white] (-.3,2.7) rectangle node[font = \tiny] {$L_b$} (.3,3.3);
    \end{tikzpicture}
  &=
    \begin{tikzpicture}[baseline=(current bounding box),scale=1.5]
        \draw[gray, string] (0,1) -- (0,1.7);
        \draw[gray, string] (0,2.3) -- (0,3);
        \draw[YaleMidBlue,string] (1,2) -- node[above]{$\sigma$} (0,2);
        \draw[YaleMidBlue,string] (0,2) --node[above]{$\sigma$} (-1,2);
        \draw[YaleMidBlue,dashed,string] (-1,2) -- node[above]{$b$} (-2,3);
        \draw[YaleMidBlue,string] (-1,2) -- node[below]{$\sigma$} (-2,1);
        \draw[YaleMidBlue,dashed,string] (2,3) -- node[above]{$a$} (1,2);
        \draw[YaleMidBlue,string] (2,1) -- node[below]{$\sigma$} (1,2);
        \filldraw[fill = white] (-.3,1.7) rectangle node {$L_\sigma$} (.3,2.3);
        \filldraw[fill = white] (-1.2,1.8) rectangle node {$X$} (-.8,2.2);
        \filldraw[fill = white] (.8,1.8) rectangle node {$-Z$} (1.2,2.2);
    \end{tikzpicture},\\
    \begin{tikzpicture}[baseline=(current bounding box),scale=.4]
        \draw[gray, string] (0,0) -- (0,.7);
        \draw[gray, string] (0,1.3) -- (0,2);
        \draw[YaleMidBlue, dashed,string] (-.3,1)  -- node[above]{$a$} (-2,1) ;
        \draw[YaleMidBlue,dashed, string] (2,1) -- node[above]{$b$} (.3,1) ;
        \filldraw[fill = white] (-.3,.7) rectangle node[font = \tiny] {$L_a$} (.3,1.3);
        \draw[gray, string] (0,2) -- (0,2.7);
        \draw[gray, string] (0,3.3) -- (0,4);
        \draw[YaleMidBlue, string] (-.3,3)  -- node[above]{$\sigma$} (-2,3);
        \draw[YaleMidBlue, string] (2,3) -- node[above]{$\sigma$} (.3,3) ;
        \filldraw[fill = white] (-.3,2.7) rectangle node[font = \tiny] {$L_\sigma$} (.3,3.3);
    \end{tikzpicture}
  &=
    \begin{tikzpicture}[baseline=(current bounding box),scale=1.5]
        \draw[gray, string] (0,1) -- (0,1.7);
        \draw[gray, string] (0,2.3) -- (0,3);
        \draw[YaleMidBlue,string] (1,2) -- node[above]{$\sigma$} (0,2);
        \draw[YaleMidBlue,string] (0,2) --node[above]{$\sigma$} (-1,2);
        \draw[YaleMidBlue,string] (-1,2) -- node[above]{$\sigma$} (-2,3);
        \draw[YaleMidBlue,dashed,string] (-1,2) -- node[below]{$a$} (-2,1);
        \draw[YaleMidBlue,string] (2,3) -- node[above]{$\sigma$} (1,2);
        \draw[YaleMidBlue,dashed,string] (2,1) -- node[below]{$b$} (1,2);
        \filldraw[fill = white] (-.3,1.7) rectangle node {$L_\sigma$} (.3,2.3);
        \filldraw[fill = white] (-1.2,1.8) rectangle node {$Z$} (-.8,2.2);
        \filldraw[fill = white] (.8,1.8) rectangle node {$-X$} (1.2,2.2);
    \end{tikzpicture},\\
    \begin{tikzpicture}[baseline=(current bounding box),scale=.4]
        \draw[gray, string] (0,0) -- (0,.7);
        \draw[gray, string] (0,1.3) -- (0,2);
        \draw[YaleMidBlue, dashed,string] (-.3,1)  -- node[above]{$b$} (-2,1);
        \draw[YaleMidBlue, dashed,string] (2,1) -- node[above]{$a$} (.3,1) ;
        \filldraw[fill = white] (-.3,.7) rectangle node[font = \tiny] {$L_b$} (.3,1.3);
        \draw[gray, string] (0,2) -- (0,2.7);
        \draw[gray, string] (0,3.3) -- (0,4);
        \draw[YaleMidBlue, string] (-.3,3)  -- node[above]{$\sigma$} (-2,3);
        \draw[YaleMidBlue, string] (2,3) -- node[above]{$\sigma$} (.3,3);
        \filldraw[fill = white] (-.3,2.7) rectangle node[font = \tiny] {$L_\sigma$} (.3,3.3);
    \end{tikzpicture}
  &=
    \begin{tikzpicture}[baseline=(current bounding box),scale=1.5]
        \draw[gray, string] (0,1) -- (0,1.7);
        \draw[gray, string] (0,2.3) -- (0,3);
        \draw[YaleMidBlue,string] (1,2) -- node[above]{$\sigma$} (0,2);
        \draw[YaleMidBlue,string] (0,2) --node[above]{$\sigma$} (-1,2);
        \draw[YaleMidBlue,string] (-1,2) -- node[above]{$\sigma$} (-2,3);
        \draw[YaleMidBlue,dashed,string] (-1,2) -- node[below]{$b$} (-2,1);
        \draw[YaleMidBlue,string] (2,3) -- node[above]{$\sigma$} (1,2);
        \draw[YaleMidBlue,dashed,string] (2,1) -- node[below]{$a$} (1,2);
        \filldraw[fill = white] (-.3,1.7) rectangle node {$L_\sigma$} (.3,2.3);
        \filldraw[fill = white] (-1.2,1.8) rectangle node {$X$} (-.8,2.2);
        \filldraw[fill = white] (.8,1.8) rectangle node {$Z$} (1.2,2.2);
    \end{tikzpicture}.
\end{align*}
    
And the non-invertible tensor composes with itself as
\begin{align*}
    \begin{tikzpicture}[baseline=(current bounding box),scale=.5]
        \draw[gray, string] (0,0) -- (0,.8);
        \draw[gray, string] (0,1.2) -- (0,2);
        \draw[YaleMidBlue, string] (-.3,1)  -- node[above]{$\sigma$} (-2,1);
        \draw[YaleMidBlue, string] (2,1) -- node[above]{$\sigma$} (.3,1);
        \filldraw[fill = white] (-.3,.7) rectangle node[font = \tiny] {$L_\sigma$} (.3,1.3);
        \draw[gray, string] (0,2) -- (0,2.7);
        \draw[gray, string] (0,3.3) -- (0,4);
        \draw[YaleMidBlue, string] (-.3,3)  -- node[above]{$\sigma$} (-2,3);
        \draw[YaleMidBlue, string] (2,3) -- node[above]{$\sigma$} (.3,3);
        \filldraw[fill = white] (-.3,2.7) rectangle node[font = \tiny] {$L_a$} (.3,3.3);
    \end{tikzpicture}
  &=
    \begin{tikzpicture}[baseline=(current bounding box),scale=1.5]
        \draw[gray, string] (0,1) -- (0,1.8);
        \draw[gray, string] (0,2.2) -- (0,3);
        \draw[YaleMidBlue,dashed,string] (1,2) -- node[above]{$e$} (0,2);
        \draw[YaleMidBlue,dashed,string] (0,2) --node[above]{$e$} (-1,2);
        \draw[YaleMidBlue,string] (-1,2) -- node[above]{$\sigma$} (-2,3);
        \draw[YaleMidBlue,string] (-1,2) -- node[below]{$\sigma$} (-2,1);
        \draw[YaleMidBlue,string] (2,3) -- node[above]{$\sigma$} (1,2);
        \draw[YaleMidBlue,string] (2,1) -- node[below]{$\sigma$} (1,2);
        \filldraw[fill = white] (-.2,1.8) rectangle node {$L_e$} (.2,2.2);
        \filldraw[fill = white] (-1.2,1.8) rectangle node {$\psi_e$} (-.8,2.2);
        \filldraw[fill = white] (.8,1.8) rectangle node {$\ii\psi_e^\dagger$} (1.2,2.2);
    \end{tikzpicture}\\
      &+
    \begin{tikzpicture}[baseline=(current bounding box),scale=1.5]
        \draw[gray, string] (0,1) -- (0,1.8);
        \draw[gray, string] (0,2.2) -- (0,3);
        \draw[YaleMidBlue,dashed,string] (1,2) -- node[above]{$b$} (0,2);
        \draw[YaleMidBlue,dashed,string] (0,2) --node[above]{$a$} (-1,2);
        \draw[YaleMidBlue,string] (-1,2) -- node[above]{$\sigma$} (-2,3);
        \draw[YaleMidBlue,string] (-1,2) -- node[below]{$\sigma$} (-2,1);
        \draw[YaleMidBlue,string] (2,3) -- node[above]{$\sigma$} (1,2);
        \draw[YaleMidBlue,string] (2,1) -- node[below]{$\sigma$} (1,2);
        \filldraw[fill = white] (-.2,1.8) rectangle node {$L_a$} (.2,2.2);
        \filldraw[fill = white] (-1.2,1.8) rectangle node {$\psi_a$} (-.8,2.2);
        \node[draw = black,fill = white] at (1,2) {$-\ii\psi_b^\dagger$} (1.2,2.2);
    \end{tikzpicture}\\
      &+
    \begin{tikzpicture}[baseline=(current bounding box),scale=1.5]
        \draw[gray, string] (0,1) -- (0,1.8);
        \draw[gray, string] (0,2.2) -- (0,3);
        \draw[YaleMidBlue,dashed,string] (1,2) -- node[above]{$a$} (0,2);
        \draw[YaleMidBlue,dashed,string] (0,2) --node[above]{$b$} (-1,2);
        \draw[YaleMidBlue,string] (-1,2) -- node[above]{$\sigma$} (-2,3);
        \draw[YaleMidBlue,string] (-1,2) -- node[below]{$\sigma$} (-2,1);
        \draw[YaleMidBlue,string] (2,3) -- node[above]{$\sigma$} (1,2);
        \draw[YaleMidBlue,string] (2,1) -- node[below]{$\sigma$} (1,2);
        \filldraw[fill = white] (-.2,1.8) rectangle node {$L_b$} (.2,2.2);
        \filldraw[fill = white] (-1.2,1.8) rectangle node {$\psi_b$} (-.8,2.2);
        \filldraw[fill = white] (.8,1.8) rectangle node {$\ii\psi_a^\dagger$} (1.2,2.2);
    \end{tikzpicture}\\
      &+
    \begin{tikzpicture}[baseline=(current bounding box),scale=1.5]
        \draw[gray, string] (0,1) -- (0,1.8);
        \draw[gray, string] (0,2.2) -- (0,3);
        \draw[YaleMidBlue,dashed,string] (1,2) -- node[above]{$c$} (0,2);
        \draw[YaleMidBlue,dashed,string] (0,2) --node[above]{$c$} (-1,2);
        \draw[YaleMidBlue,string] (-1,2) -- node[above]{$\sigma$} (-2,3);
        \draw[YaleMidBlue,string] (-1,2) -- node[below]{$\sigma$} (-2,1);
        \draw[YaleMidBlue,string] (2,3) -- node[above]{$\sigma$} (1,2);
        \draw[YaleMidBlue,string] (2,1) -- node[below]{$\sigma$} (1,2);
        \filldraw[fill = white] (-.2,1.8) rectangle node {$L_c$} (.2,2.2);
        \filldraw[fill = white] (-1.2,1.8) rectangle node {$\psi_c$} (-.8,2.2);
        \filldraw[fill = white] (.8,1.8) rectangle node {$\ii\psi_c^\dagger$} (1.2,2.2);
    \end{tikzpicture}.
\end{align*}
We can see that the projective charge effectively swaps $a$ and $b$ MPOs, which we denote by $(a,b)$.

Similarly, the effective Hamiltonian can be simplified to
\[
    H^{(2)}_{\text{eff}} = -\sum_{i}|\psi\ra\la \psi|_{(2i-1,2i)}.
\]
Again, one can check that each local term in the Hamiltonian is symmetric by
\[
    \diagram{1}{
        \draw[gray,string] (0,1) -- (0,0) -- (1,0) -- (1,1);
        \draw[YaleMidBlue,string] (1.5,.5) -- (-.5,.5);
        \filldraw[fill = white] (-.2,.3) rectangle node{$R_s$} (.2,.7);
        \filldraw[fill = white] (.8,.3) rectangle node{$L_s$} (1.2,.7);
    }
    =
    \diagram{1}{
        \draw[gray,string] (0,1) -- (0,0) -- (1,0) -- (1,1);
        \draw[YaleMidBlue,string] (1.5,-.2) -- (-.5,-.2);
    }.
\]
The unique ground state on an infinite chain is
\[
    |\Psi\ra = \otimes_{i}|\psi\ra_{(2i-1,2i)} \propto 
    \diagram{.7}{
        \foreach \i in {1,2,...,5}
        {
            \draw[gray,string] (\i+.1,.8) -- (\i+.1,0) -- (\i+.9,0) -- (\i+.9,.8);
        }
        \node at (.5,0) {$\cdots$};
        \node at (6.5,0) {$\cdots$};
    }.
\]
Similarly, when defined on an open chain with free boundary conditions, the edge state exhibits a 2-fold degeneracy.

\subsection{Other projective charges}
We observe that $W_{(a,c)}$ and $W_{(a,b)}$ are two projective charges that realize the permutations $(a, c)$ and $(a, b)$ respectively. Since any $S_3$-permutation of $\{a, b, c\}$ is generated by $(a, c)$ and $(a, b)$, we can tensor the projective charges $W_{(a,c)}$ and $W_{(a,b)}$ and then reduce it into irreducible component such that all possible $S_3$-permutations of $\{a, b, c\}$ can be realized via those projective charges.

In our construction, the projective charges are the input for the commuting projector Hamiltonian of the SPT phase. As we will demonstrate, these newly constructed projective charges do not give rise to new $\Rep^\dagger(D_8)$-SPT phases but must realize one of the three previously identified phases.
The subgroup $K:= \langle (b,c) \rangle \subset S_3$ realizes the trivial phase, and the left cosets in the quotient $S_3 / K = \{K, \;(a,c)K,\; (a,b)K\}$ represent three distinct phases. We label the trivial phase as phase 0  $:= K$, two non-trivial phases as phase 1 $:= (a,c)K$, and phase 2 $:= (a,b)K$ in the following.

\subsubsection{\texorpdfstring{$(c,b,a)\sim (a,c)$}{(c, b, a) ~ (a, c)} as phase}
The tensor product space $W_{(a,b)}\ot W_{(a,c)}$ is equipped with the  projective $\Rep^\dagger(D_8)$-actions
\[
    \diagram{2}{
        \draw[gray,string] (0,0) -- node [left] {$W_{(a,b)}$}(0,.8);
        \draw[gray,string] (0,1.2) -- node[left]{$W_{(a,b)}$} (0,2);
        \draw[gray,string] (1.5,0) --  node [left] {$W_{(a,c)}$}(1.5,.8);
        \draw[gray,string] (1.5,1.2) -- node[left]{$W_{(a,c)}$} (1.5,2);
        \draw[YaleMidBlue,string] (2.5,1) -- node[above]{$s$}(1.7,1);
        \draw[YaleMidBlue,string] (1.3,1) -- node[above]{$(a,c)(s)$} (.2,1);
        \draw[YaleMidBlue,string] (-.2,1) -- node[above left]{$(c,b,a)(s)$} (-1,1);
        \node[fill = white, draw = black] at (0,1) {\small $L_{(c,b,a)(s)}$};
        \node[fill = white, draw = black] at (1.5,1) {\small $L_{(a,c)(s)}$};
    },
\]

The two invariant subspaces of the projective MPOs are $\spann(|00\ra,|11\ra)$ and $\spann(|01\ra, |10\ra)$. When restricted to $W_{(c,b,a)}:=\spann(|00\ra,|11\ra)$, 
the projective actions are reduced to
\begin{align*}
    (\tilde{L}_{a})^{11} &= \tilde{Z},\quad (\tilde{L}_b)^{11} = \one,\\
    (\tilde{L}_\sigma)_{01} &= \frac{1}{\sqrt{2}}(\one + \ii Y)\frac{1}{\sqrt{2}}(\one - \ii X),\\
    (\tilde{L}_\sigma)_{10} &= \frac{1}{\sqrt{2}}(\ii \one +  Y)\frac{1}{\sqrt{2}}(-\ii \one+ X),
\end{align*}
where we denote $|\tilde{0}\ra = |00\ra$, $|\tilde{1}\ra = |11\ra$. The reduced local tensors are
\begin{align*}
    &(T_{a})^{11} = \tilde{Z}\ot \tilde{Z},\quad  (T_{b})^{11} = \one \ot \one,\\
    &(T_\sigma)_{(\tilde{0}\tilde{0})(\tilde{1}\tilde{1})} = (T_\sigma)_{(\tilde{1}\tilde{1})(\tilde{0}\tilde{0})} = \one,\\
    &(T_\sigma)_{(\tilde{0}\tilde{1})(\tilde{1}\tilde{0})}= -\ii X,\quad (T_\sigma)_{(\tilde{1}\tilde{0})(\tilde{0}\tilde{1})}= \ii X.
\end{align*}
We identify the basis as
\begin{alignat*}{2}
    &|\tilde{0}\tilde{0}\ra = \frac{1}{\sqrt{2}}\left(|\tilde{e}\ra + |\tilde{r^2}\ra\right),\quad && |\tilde{1}\tilde{1}\ra = \frac{1}{\sqrt{2}}\left(|\tilde{e}\ra - |\tilde{r^2}\ra\right),\\
    &|\tilde{0}\tilde{1}\ra = \frac{1}{\sqrt{2}}\left(|\tilde{s}\ra + |\tilde{sr^2}\ra\right),\quad &&
    |\tilde{1}\tilde{0}\ra = \frac{-\ii}{\sqrt{2}}\left(|\tilde{s}\ra - |\tilde{sr^2}\ra\right).
\end{alignat*}
The induced rank-3 tensor is
\[
    \widetilde{m}\,|\tilde{i}\tilde{j}\ra\,|\tilde{k}\tilde{l}\ra = \frac{1}{\sqrt{2}}\,\delta_{jk}\,|\tilde{i}\tilde{l}\ra.
\]
This choice of gauge differs from the standard gauge of $m_1$ in table \ref{tab:omega} up to a unitary transformation
\[
    |\tilde{g}\ra = \varphi(g)|g\ra,
\]
with
\begin{alignat*}{2}
    &\varphi(e) = 1,\quad &&\varphi(r^2) = -1,\\
    &\varphi(s) = \ee^{\frac{\pi}{4}\ii},\quad &&\varphi(sr^2) = -\ee^{\frac{\pi}{4}\ii}.
\end{alignat*}
Thus the projective charge $W_{(c,b,a)}$ realizes the same phase as $W_{(a,c)}$.

\subsubsection{\texorpdfstring{$(a,b,c)\sim (a,b)$}{(a, b, c) ~ (a, b)} as phase}
The tensor product space $W_{(a,c)}\ot W_{(a,b)}$ is equipped with the  projective $\Rep^\dagger(D_8)$-actions
\[
\diagram{2}{
    \draw[gray,string] (0,0) -- node[left]{$W_{(a,c)}$} (0,.8);
    \draw[gray,string] (0,1.2) -- node[left]{$W_{(a,c)}$} (0,2);
    \draw[gray,string] (1.5,0) --  node[left]{$W_{(a,b)}$} (1.5,.8);
    \draw[gray,string] (1.5,1.2) -- node[left]{$W_{(a,b)}$} (1.5,2);
    \draw[YaleMidBlue,string] (2.5,1) -- node[above]{$s$}(1.7,1);
    \draw[YaleMidBlue,string] (1.3,1) -- node[above]{$(a,b)(s)$} (.2,1);
    \draw[YaleMidBlue,string] (-.2,1) -- node[above left]{$(a,b,c)(s)$} (-1,1);
    \node[fill = white, draw = black] at (0,1) {\small $L_{(a,b,c)(s)}$};
    \node[fill = white, draw = black] at (1.5,1) {\small $L_{(a,b)(s)}$};
},
\]

The invariant subspaces of the projective MPOs are $\spann(|00\ra,|11\ra)$ and $ \spann(|01\ra ,|10\ra)$. When restricted to $W_{(a,b,c)} = \spann(|00\ra,|11\ra)$,  
the projective actions $L_s$ are reduced to
\begin{align*}
    &(\tilde{L}_a)^{11} = \tilde{Z},\quad  (\tilde{L}_b)^{11} = \tilde{Z},\\
    &(\tilde{L}_\sigma)_{01} = \frac{1}{\sqrt{2}}(\one -\ii X)\frac{1}{\sqrt{2}}(\one + \ii Y),\\
    &(\tilde{L}_\sigma)_{10} =\frac{1}{\sqrt{2}}(-\ii \one + X)\frac{1}{\sqrt{2}}(\ii\one +  Y),
\end{align*}
where we denote $|\tilde{0}\ra = |00\ra$ and $|\tilde{1}\ra = |11\ra$. The reduced $\Rep^\dagger(D_8)$-MPOs are
\begin{align*}
    &(T_{a})^{11} = (T_{b})^{11} = \tilde{Z}\ot \tilde{Z},\\
    &(T_\sigma)_{(\tilde{0}\tilde{0})(\tilde{1}\tilde{1})} = (T_\sigma)_{(\tilde{1}\tilde{1})(\tilde{0}\tilde{0})} = \one,\\
    &(T_\sigma)_{(\tilde{0}\tilde{1})(\tilde{1}\tilde{0})} =\ii Y,\quad  (T_\sigma)_{(\tilde{1}\tilde{0})(\tilde{0}\tilde{1})} = -\ii Y.
\end{align*}
Thus the space is $\C[\la r^2,sr\ra]$ with the basis defined as
\begin{alignat*}{2}
    &|\tilde{0}\tilde{0}\ra = \frac{1}{\sqrt{2}}\left(|e\ra + |r^2\ra\right),\quad &&|\tilde{1}\tilde{1}\ra = \frac{1}{\sqrt{2}}\left(|e\ra - |r^2\ra\right),\\
    &|\tilde{0}\tilde{1}\ra = \frac{1}{\sqrt{2}}\left(|sr\ra + |sr^3\ra\right),\quad &&
    |\tilde{1}\tilde{0}\ra = \frac{\ii}{\sqrt{2}}\left(|sr\ra -|sr^3\ra\right).
\end{alignat*}
The induced rank-3 tensor is
\[
    \widetilde{m}\,|\tilde{i}\tilde{j}\ra\,|\tilde{k}\tilde{l}\ra = \frac{1}{\sqrt{2}}\,\delta_{jk}\,|\tilde{i}\tilde{l}\ra.
\]
This choice of gauge differs from the standard gauge of $m_2$ in table \ref{tab:omega} up to a unitary transformation
\[
    |\tilde{g}\ra = \varphi(g)|g\ra,
\]
with
\begin{alignat*}{2}
    &\varphi(e) = 1,\quad &&\varphi(r^2) = -1,\\
    &\varphi(sr) = \ee^{\frac{\pi}{4}\ii},\quad &&\varphi(sr^3) = -\ee^{\frac{\pi}{4}\ii}.
\end{alignat*}
Thus the projective charge $W_{(a,b,c)}$ realizes the same phase as $W_{(a,b)}$.

\subsubsection{\texorpdfstring{$(b,c)\sim \id$}{(b, c) ~ id} as phase}
Now we consider the tensor product $W_{(c,b,a)}\ot W_{(a,b)}$ and its reduction. Similarly, when restricted to the invariant subspace $W_{(b,c)} = \spann(|00\ra,|11\ra)$,
the reduced projective MPOs on $W_{(b,c)} = \spann(|00\ra,|11\ra)$ are
\begin{align*}
    &\tilde{L}_{a} = \tilde{L}_b = \one,\\
    &(\tilde{L}_\sigma)_{01} = \frac{\ii}{\sqrt{2}}(-X+Y),\quad (\tilde{L}_\sigma)_{10} = \frac{1}{\sqrt{2}}(-X+Y),
\end{align*}
and the reduced MPOs are
\begin{align*}
    &(T_a)^{11} =(T_b)^{11} = \one \ot \one,\\
    &(T_{a})^{11} = (T_{b})^{11} = \tilde{Z}\ot \tilde{Z},\\
    &(T_\sigma)_{(\tilde{0}\tilde{0})(\tilde{1}\tilde{1})} = (T_\sigma)_{(\tilde{1}\tilde{1})(\tilde{0}\tilde{0})} = \one,\\
    &(T_\sigma)_{(\tilde{0}\tilde{1})(\tilde{1}\tilde{0})} =\ii \one,\quad  (T_\sigma)_{(\tilde{1}\tilde{0})(\tilde{0}\tilde{1})} = -\ii \one.
\end{align*}
Thus the space $W_{(b,c)}\ot W_{(b,c)}^*$ is isomorphic to $\C_e^2\oplus \C_{r^2}^2$, with the basis transformation 
\begin{alignat*}{2}
    &|\tilde{0}\tilde{0}\ra = \frac{1}{\sqrt{2}}(|e,1\ra + |r^2,1\ra),\quad && |\tilde{1}\tilde{1}\ra = \frac{1}{\sqrt{2}}(|e,1\ra - |r^2,1\ra),\\
    &|\tilde{0}\tilde{1}\ra = \frac{1}{\sqrt{2}}(|e,2\ra +\ii |r^2,2\ra),\quad &&|\tilde{1}\tilde{0}\ra = \frac{1}{\sqrt{2}}(\ii|e,2\ra + |r^2,2\ra).\\
\end{alignat*}
This space can be further reduced to $\C_e$ using a similar approach, which is spanned by
\[
    |e\ra = \frac{1}{\sqrt{2}}\Big(|e,1\ra + \ee^{\frac{\pi}{4}\ii}|e,2\ra \Big).
\]
We conclude that the $W_{(b,c)}$ realizes the same phase as $\C_{e}$.

\subsection{\texorpdfstring{$S_3$}{S3}-duality between three SPT phases}\label{sec: S3 duality}
In this section, we propose a duality procedure to relate three distinct SPT phases with the $\Rep^\dagger({D_8})$ symmetry. This duality transformation fills the blank of the permutations between different SPTs that could have been realized via stacking structure and physically realizes the monoidal equivalence connecting different fiber functors:
\begin{itemize}
    \item the stacking structure of $G$-SPTs naturally facilitates permutations between different SPT phases, however, SPTs with non-invertible symmetries lack a natural stacking structure;
    \item three fiber functors of $\Rep^\dagger({D_8})$ are related by the $S_3$-automorphism. We would like to understand the microscopic realization of such automorphism.
\end{itemize}

Given a $\Rep^\dagger(D_8)$-symmetric lattice model $(\mathcal{H}, H)$, the duality procedure produces another $\Rep^\dagger(D_8)$-symmetric lattice model $(\mathcal{H}', H')$, which is parameterized by a projective $\Rep^\dagger(D_8)$-charge $W_\pi$, where $\pi\in S_3$. More explicitly, given the Hilbert space
\[
    \cH = \cdots \ot V_n\ot V_{n+1}\ot \cdots \ot V_{n+m}\ot \cdots,
\]
and Hamiltonian
\[
    H = \sum_{X\subset \Z}\Phi(X),
\]
we will perform the duality between SPTs parameterized by $W_\pi$.

In the first step, we enlarge the Hilbert space by introducing extra degrees of freedom (defects) on the links. A similar approach used in gauging generalized symmetries has been discussed in \cite{Seifnashri_2024}. In our language, we sandwich the original lattice site with additional sublattice sites. We separate the extra degrees of freedom into left and right components. The left sublattice site, labeled by $(i,L)$, hosts a projective $\Rep^\dagger(D_8)$-charge $W_\pi$, while the right sublattice site, labeled by $(i,R)$, carries an opposite projective $\Rep^\dagger(D_8)$-charge $W_\pi^*$. This construction is diagrammatically represented as
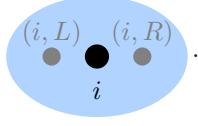
\[
    \diagram{4}{
        \fill[YaleLightBlue] (0,0) ellipse [x radius=.3,y radius=0.2];
        \fill[gray] (-.15,0) node[above]{$(i,L)$} circle [radius = .03];
        \fill (0,0) node[below] at (0,-.05) {$i$} circle [radius = .04];
        \fill[gray] (+.15,0) node[above]{$(i,R)$} circle [radius = .03];
    }\,.
\]
This sandwich construction ensures that the total charge in the extended unit cell remains a $\Rep^\dagger(D_8)$-charge, meanwhile the $\Rep^\dagger(D_8)$-MPO is permuted by $\pi$
\[
    \diagram{1}{
        \foreach \i in {0,1,...,5}
        {
            \fill (\i*.5,0) circle [radius = .05];
        }
    }
    \longrightarrow 
    \diagram{1}{
        \foreach \i in {0,1,...,5}
        {
            \fill[YaleLightBlue, ] (\i*.5,0) ellipse [x radius=.2,y radius=0.15];
            \fill[gray] (\i*.5 -.15,0) circle [radius = .03];
            \fill (\i*.5,0) circle [radius = .05];
            \fill[gray] (\i*.5 +.15,0) circle [radius = .03];
        }
    }.
\]
More concretely, the enlarged Hilbert space on site $i$ is $W_\pi\ot V_i\ot W_\pi^*$, acted upon by the permuted $\Rep^\dagger(D_8)$-MPO 
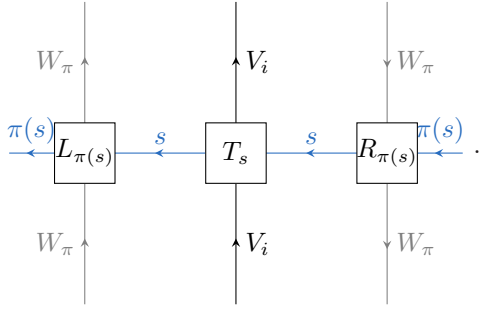
\[
    \diagram{2}{
        \draw[string, gray] (-.5,0) -- node[left]{$W_\pi$} (-.5,1-.2);
        \draw[string, gray] (-.5,1+.2) -- node[left]{$W_\pi$} (-.5,2);
        \draw[string,gray] (1.5,1-.2) -- node[right]{$W_\pi$} (1.5,0);
        \draw[string,gray] (1.5,2) -- node[right]{$W_\pi$} (1.5,1+.2);
        \draw[string] (.5,0) -- node[right]{$V_i$} (.5,1-.2);
        \draw[string] (.5,1+.2) -- node[right]{$V_i$} (.5,2);
        \draw[YaleMidBlue,string] (2,1) -- node[above]{$\pi(s)$} (1.7,1);
        \draw[YaleMidBlue,string] (1.3,1) -- node[above]{$s$} (.7,1);
        \draw[YaleMidBlue,string] (.3,1) -- node[above]{$s$} (-.3,1);
        \draw[YaleMidBlue,string] (-.7,1) -- node[above]{$\pi(s)$}  (-1,1);
        \filldraw[fill = white] (-.7,.8) rectangle node{$L_{\pi(s)}$} (-.3,1.2);
        \filldraw[fill = white] (1.3,.8) rectangle node{$R_{\pi(s)}$} (1.7,1.2);
        \filldraw[fill = white] (.3,.8) rectangle node{$T_s$} (.7,1.2);
    }.
\]
Clearly, the permuted MPOs still satisfy the composition rules of the original $\Rep^\dagger(D_8)$-MPOs.

For each local $D_8$-grading preserving operator $O$ acting on multiple sites, the duality transformation maps
\begin{equation}\label{eq:dual mini coup}
    \diagram{1}{
        \draw[string] (0,0) -- (0,1-.2);
        \draw[string] (0,1+.2) -- (0,2);
        \draw[string] (1,0) -- (1,1-.2);
        \draw[string] (1,1+.2) -- (1,2);
        \draw[string] (2,0) -- (2,1-.2);
        \draw[string] (2,1+.2) -- (2,2);
        \filldraw[fill = white] (-.1,.8) rectangle node{$O$} (2.1,1.2);
    }\mapsto
    \diagram{1}{
        \draw[string] (0,0) -- (0,1-.2);
        \draw[string] (0,1+.2) -- (0,2);
        \draw[string] (1,0) -- (1,1-.2);
        \draw[string] (1,1+.2) -- (1,2);
        \draw[string] (2,0) -- (2,1-.2);
        \draw[string] (2,1+.2) -- (2,2);
        \draw[string, gray] (.2, 2) -- (.2,1.4) -- (.8,1.4) -- (.8,2);
        \draw[string, gray] (1.2, 2) -- (1.2,1.4) -- (1.8,1.4) -- (1.8,2);
        \draw[string, gray] (.8, 0) -- (.8,.6) -- (.2,.6) -- (.2,0);
        \draw[string, gray] (1.8, 0) -- (1.8,.6) -- (1.2,.6) -- (1.2,0);
        \draw[string, gray] (-.2,0) -- (-.2,2);
        \draw[string, gray] (2.2,2) -- (2.2,0);
        \filldraw[fill = white] (-.1,.8) rectangle node{$O$} (2.1,1.2);
    }:=\tilde{O}.
\end{equation}
Note that such a map, when ignoring the symmetry, maps $O$ to $O\ot |\psi\ra\la \psi|\ot |\psi\ra\la \psi|\ot \cdots $, i.e. the original operator tensor-producting with the interactions on the projective charges, where $|\psi\ra$ is the maximally entangled state defined in Equation \eqref{eq:maximally entangled}.

In the second step, we project the extended Hilbert space into the invariant subspace, analogous to imposing the Gauss law constraint locally. To ensure a tensor product Hilbert space, we impose the Gauss law energetically \footnote{Note that the Gauss law is a kinematic constraint. Strictly speaking, we should impose the constraint at the Hilbert space level and identify the gauge-invariant generators to recover the tensor product Hilbert space after gauging. Here for convenience, we impose the Gauss law energetically to achieve the same effect.}
\[
    \tilde{H} = \sum_{X\subset \Z}\widetilde{\Phi(X)} -J\sum_{i\in\Z}|\psi\ra\la \psi|_{(i,R),(i+1,L)},
\]
where $\widetilde{\Phi(X)}$ is obtained from $\Phi(X)$ via Equation \eqref{eq:dual mini coup}, and $J\gg 1$ is a parameter that enforces $|\psi\ra\la \psi|_{(i,R),(i+1,L)} = 1$ for any $i\in\Z$ at the low energy subspace. The transformed Hamiltonian $\tilde{H}$ remains to be $\Rep^\dagger(D_8)$-symmetric, but with the $\Rep^\dagger(D_8)$-MPO permuted.
In the low energy subspace, where the constraint $|\psi\ra\la \psi|_{(i,R),(i+1,L)} = 1$ is satisfied, the Hamiltonian $\tilde{H}$ effectively reduces to the original one $H$ when the symmetry is ignored. Thus the reduced Hamiltonian shares the same energy spectrum as $H$. Consequently, if $H$ is gapped and has the unique ground state, $\tilde{H}$ inherits its properties and remains gapped with unique ground state.

Moreover, the adiabatic path in the original model can be mapped to the one in the dual lattice model such that no phase transition occurs along that path. If $H$ and $H'$ belong to the same phase, which means that there exists a symmetric adiabatic path connecting these two Hamiltonians, then their corresponding lattice models $\tilde{H}$ and $\tilde{H}'$ after the duality transformation belong to the same phase as well. 

Note that the deformed MPO differs from the one that might act diagonally on two layers of $\Rep^\dagger(D_8)$ SPTs, thus it does not work as the stacking. The newly-introduced projective charges cannot be located in another layer of the system as we are not allowed to introduce interactions between two projective charges without surpassing the symmetry charge sandwiched in between.

Now, we can examine how the phases in $S_3/K$ are permuted under the duality parameterized by the projective charge $W_\pi$.
Suppose the original lattice model is constructed from $W_{\sigma}$, corresponding to the phase described by the coset $\sigma K$. After the duality transformation, the resulting model is equivalent to the one constructed from $W_{\pi} \otimes W_{\sigma}$. The projection of the projective charge onto the invariant subspace can then be added to the Hamiltonian adiabatically, we discuss this in more detail in \ref{section: deformation}. After the projection, this model is equivalent to the one constructed by the projective charge $W_{\pi\sigma}$, corresponding to the phase described by the coset $\pi\sigma K$.
Thus, the duality parameterized by $W_\pi$ maps the phase described by the coset $\sigma K$ to the phase described by $\pi\sigma K$. The $S_3$ duality acts on phases in the same way as the group acts on cosets via group multiplication.

\subsection{Comparison with the cluster state SPT}
When comparing our construction with the models in \cite{seifnashri2024clusterstatenoninvertiblesymmetry}, we notice that the $\Rep^\dagger(D_8)$ MPO $D$ proposed there is not in its onsite form. The virtual bond dimension of $D$ is four. In fact, we can conjugate $D$ by CZ to reduce its virtual bond dimension to two. 
We group two nearby sites of the cluster model as a unit cell and identify the basis as
\[
    |++\ra = |e\ra,\quad |-+\ra = |s\ra,\quad |--\ra = |r\ra,\quad |+-\ra = |sr\ra,\,
\]
thus the local Hilbert space in \cite{seifnashri2024clusterstatenoninvertiblesymmetry} is
\[
    \C^2\ot \C^2 = \mathrm{span}( |e\ra,\ |s\ra,\ |r\ra,\ |sr\ra )\,.
\]

Upon conjugating the cluster Hamiltonian with CZ, cluster state becomes the product state, suggesting that if we transform the MPO $D$ into its onsite form, we can choose the produce state as a representative of the trivial phase. In fact, the trivial phase is constructed from $\C_e$.

In their second model, after applying the CZ gate, the ground state is stabilized by the generator
\[
    -X_{2n} = 1,\quad -Z_{2n-1}X_{2n+1}Z_{2n+3} = 1
\]
Imposing the constraint $X_{2n} = -1$, the local Hilbert space on a unit cell (two nearby sites) is $\C\{r, sr\}$. The third model is similar. After applying the CZ gate, the ground state is stabilized by the generator
\[
    -X_{2n-1} = 1,\quad -Z_{2n-2}X_{2n}Z_{2n+2} = 1
\]
Imposing the constraint $X_{2n-1} = -1$, the local Hilbert space on a unit cell (two nearby sites) is $\C\{s,r\}$.
We further group two unit cells together to form the final unit cell, and the two non-trivial SPT phases are constructed from 
\[
    \C\{r, sr\}\ot\C\{r, sr\} = \C\la r^2,s\ra\,,
\]
and
\[
    \C\{s,r\}\ot\C\{s,r\} = \C\la r^2,sr\ra\,.
\]

In Appendix \ref{sec:reduction}, we explicitly show the equivalence between our Q-system construction with the above input algebras and those in \cite{seifnashri2024clusterstatenoninvertiblesymmetry} (up to a symmetry-preserving local unitary).

\section{Microscopic realization of anomaly-free fusion category symmetry} \label{sec:MPO}

Similar to imposing an internal group-like symmetry on the lattice, one should not only assign a global symmetry acting on the total Hilbert space but also a local action on each site and each local finite region.
Thus we need to provide a local version of the anomaly-free fusion category symmetry and its action on each lattice site in a 1D lattice model \cite{Inamura_2022,molnar2022matrixproductoperatoralgebras,inamura202411dsptphasesfusion,Garre_Rubio_2023}. The locality of the fusion category symmetry is realized by the fiber functor of a fusion category if there exists one. Specifically, a fiber functor ensures that the virtual bond of the fusion category symmetry matrix product operator (MPO) can be coherently assigned to a Hilbert space \cite{molnar2022matrixproductoperatoralgebras,Garre_Rubio_2023,Inamura_2022,inamura202411dsptphasesfusion}. This assignment allows the MPO to be decomposed into local tensors, where each tensor on a single site encodes the local data of the fusion category symmetry. Mathematically, a fusion category symmetry $(\cC,f)$ naturally gives a finite-dimensional Hopf C$^\star $-algebra from the coend \cite{Kitaev_2012,kong2012universalpropertieslevinwenmodels,lan2024tubecategorytensorrenormalization,etingof2015tensor} 
\[
    \sH := \int^{X\in\cC}\Hilb(fX,fX).
\]

We start with the Hopf algebra $\sH$ as the local version of the fusion category symmetry and define the charge category \cite{Lan_2024,lan2024categorysetorders} of the fusion category symmetry as the representation category of the Hopf algebra, isomorphic to $\cC_{\Hilb_f}^\vee$.
A generic lattice model with fusion category symmetry $(\cC, f)$ consists of: local Hilbert spaces that carry representations of the Hopf algebra and local interactions given by Hermitian intertwiners between representations on the local Hilbert space \cite{Inamura_2022}, which directly generalizes the construction with finite group symmetry in \cite{Lan_2024}.

Then we construct the fusion category symmetry operators as MPOs, which are derived from the action of Hopf algebra on local symmetry charges \cite{molnar2022matrixproductoperatoralgebras,Garre_Rubio_2023}. We emphasize two essential properties of fusion category symmetry MPOs: the ``onsiteness'' and dualizability \cite{molnar2022matrixproductoperatoralgebras,Garre_Rubio_2023,inamura202411dsptphasesfusion}. An onsite MPO means the MPO is in its minimal presentation. We define an \textit{onsite} fusion category symmetry MPO to be the one, the virtual bond dimension $v_s\in\Z_+$ of $s\in\cC$ is the same as the quantum dimension $d_s\in\Z_+$. Fixing the symmetry operator ``onsite'' is crucial to defining a trivial phase for SPTs with either invertible or non-invertible symmetries. Only after the symmetry operator is fixed onsite, a trivial state can be properly defined.

\subsection{Finite-dimensional Hopf C\texorpdfstring{$^\star $}{*}-algebra}
Here we briefly review Hopf C$^\star $-algebra \cite{bohm1999weakhopfalgebrasi} and the representation category of a finite-dimensional Hopf C$^\star $-algebra.

A Hopf alegbra $(\sH,m,\iota,\Delta,\epsilon)$ is a bialgebra over the field $\mathbb{C}$ together with an $\mathbb{C}$-linear map $S: \sH \rightarrow \sH$, called the antipode, satisfying the condition
\[
    S(a_{(1)})a_{(2)} = \epsilon(a)\one = a_{(1)}S(a_{(2)}).
\]

In a physical system, in order to represent a Hopf algebra on the Hilbert space, we should require a Hermitian representation. Thus we consider Hopf C$^\star $-algebra and require its representation to be a $\star$-representation.

\begin{definition}[C$^\star $-algebra]
    Let $(A,m,\iota)$ be a unital algebra over $\C$. The algebra $A$ is called a complex $\star$-algebra if there is an involution map $\star: A\rightarrow A$ that is conjugate-linear and anti-automorphic, meaning
    \begin{align*}
        &x^{\star\star} = x,\quad (\lambda x + \mu y)^\star = \overline{\lambda}x^\star  + \overline{\mu}y^\star ,\\ &(xy)^\star  = y^\star x^\star ,\quad \one^\star  = \one.
    \end{align*}
    Then $A$ a C$^\star $-algebra if there exists a norm $\|\cdot\|$ on $A$ which is submultiplicative ($\|ab\|\leq \|a\|\|b\|$) such that
    \[
        \|a^\star a\| =\|a\|^2,\quad \forall a\in A. 
    \]
\end{definition}

\begin{definition}
    Let $A, B$ be two complex $\star$-algebras. Then $f:A\rightarrow B$ is a $\star$-homomorphism if it is an algebra homomorphism and is compatible with the $\star$-structure: $f\circ \star_A = \star_B \circ f$.
\end{definition}

\begin{definition}
    The finite-dimensional Hopf algebra $(\sH, m, \iota, \Delta, \epsilon)$ is a Hopf C$^\star $-algebra if the algebra $(\sH, m, \iota)$ is a C$^\star $-algebra and $\Delta$ is a $\star$-homomorphism.
\end{definition}

\begin{remark}
    If $\sH$ is a Hopf C$^\star $-algebra, one has the property that for all $x\in \sH$ \cite{nikshych2000finitequantumgroupoidsapplications},
    \[
        \one^\star  = \one,\quad  \epsilon(x^\star ) = \epsilon(x)^\star ,\quad (S\circ \star)^2 = \id.
    \]
    Moreover, by the property of the finite-dimensional unital C$^\star $-algebra, $\sH$ is automatically semisimple.   
\end{remark}

Given a finite-dimensional Hopf algebra $(\sH, m, \iota, \Delta, \epsilon)$, its dual Hopf algebra is a tuple $(\sH^*, \Delta^*, \epsilon^*, m^*, \iota^*)$, where $\sH^* = \hom(\sH,\C)$ is the dual vector space of $\sH$, and the multiplication $\Delta^*$, the unit $\epsilon^*$, the comultiplication $m^*$, the counit $\iota^*$ are dual linear maps of $\Delta$, $\epsilon$, $m$, $\iota$, respectively. Its antipode $S^*$ is uniquely fixed by $S$ . If $\sH$ is a Hopf C$^\star $-algebra with involution $\star$, the C$^\star $ structure of $\sH^*$ is defined by
\[
    (f^\star )(x):= \overline{f(S(x)^\star )}
\]
for all $f\in \sH^*$ and $x\in \sH$ \cite{nikshych2000finitequantumgroupoidsapplications}.

\begin{example}
    Given a finite group $G$, the group algebra $(\C[G],\cdot, e, \Delta, \epsilon, \star)$ is a $\star$-Hopf algebra:
    \begin{itemize}
        \item[--] multiplication $\cdot$ : $g\cdot h = gh$ for $g,h\in G$;
        \item[--] unit: $e\in G$;
        \item[--] comultiplication: $\Delta g = g\ot g$ for $g \in G$;
        \item[--] counit: $\epsilon (g) = \delta_{ge}$ for $g \in G$;
        \item[--] antipode $S$: $S(g) = g^{-1}$ for $g\in G$;
        \item[--] $\star$-structure: $g^\star  = g^{-1}$ for $g\in G$.
    \end{itemize}
\end{example}
\begin{example}
    Given a finite group $G$ and $(\C[G],\cdot,e,\Delta,\epsilon)$ the group algebra of $G$, the tuple $(\Fun(G), \Delta^*, \epsilon^*,\cdot ^*, e^*)$, called the function algebra, is the dual Hopf algebra of $\C[G]$. 
    Fixing a dual basis $\{\delta_g\}_{g\in G}$ of $\Fun(G)$, where $\delta_g(h) = \delta_{g,h}$, then we have
    \begin{itemize}
        \item[--] multiplication: $\delta_g\delta_h = \delta_{g,h}\delta_g$;
        \item[--] comultiplication: $\Delta(\delta_g) = \sum_{kl = g}\delta_k\ot \delta_l$;
        \item[--] antipode: $S(\delta_g) = \delta_{g^{-1}}$;
        \item[--] $\star$-structure: given by the antilinear map $\delta_g^\star  = \delta_g$.
    \end{itemize}
\end{example}

\begin{example}
    The Kac-Paljutkin $\sH_8$ Hopf algebra is an $8$-dimensional algebra
    \begin{multline*}
        \sH_8 = \la x,y,z|x^2 = y^2 = z^2 = 1, \\xz = zx, yz = zy, xyz = yx\ra
    \end{multline*}
    along with 
    \begin{itemize}
        \item[--] comultiplication: 
            \begin{align*}
                &\Delta x = xe_0\ot x + xe_1\ot y,\\
                &\Delta y = ye_1\ot x + ye_0\ot y,\\
                &\Delta z = z\ot z;
            \end{align*}
        \item[--] counit: $\epsilon(x) = \epsilon(y) = \epsilon(z) = 1$;
        \item[--] antipode:
            \begin{align*}
                &S(x) = xe_0 + y e_1,\\
                &S(y) = x e_1 + y e_0,\\
                &S(z) = z.
            \end{align*}
        \item[--] $\star$-structure: $x^\star  = x,\quad y^\star  = y, \quad z^\star  = z$.
    \end{itemize}
    Thus $\sH_8$ is a Hopf C$^\star $-algebra.
\end{example}

\begin{definition}
    The representation category $\Rep^\dagger(\sH)$ of a finite-dimensional Hopf C$^\star $-algebra $(\sH, m, \iota, \Delta, \epsilon,\star)$ with antipode $S$ on finite-dimensional Hilbert spaces, is a strict unitary fusion category with the following data:
    \begin{itemize}
        \item objects: $(V,\rho)$, where $V$ is a finite-dimensional Hilbert space and $\rho:\sH\rightarrow \End(V)$ is a $\star$-homomorphism, i.e., an algebra homomorphism, such that the $\star$-structure and the inner product are compatible: $\rho(x)^\dagger = \rho(x^\star )$ for any $x\in\sH$;
        \item morphisms from $(V,\rho)$ to $(W,\sigma)$: linear maps $f:V\rightarrow W$, such that $f\rho(x) = \sigma(x)f$ for all $x\in\sH$. Moreover, the Hermitian conjugation $f^\dagger:W\rightarrow V$ is a morphism from $(W,\sigma)$ to $(V,\rho)$, by the $\star$-representation structure
            \[
                f^\dagger \sigma(x) = (\sigma(x^\star )f)^\dagger = (f\rho(x^\star ))^\dagger = \rho(x)f^\dagger
            \]
        for any $x\in\sH$;
        \item tensor product $\ot$: given two $\star$-representations $A = (V,\rho_V)$ and $B = (W,\rho_W)$, their tensor product $A\ot B$ is a pair $(V\ot W,\rho_{V\ot W})$, with
            \[
                \rho_{V\ot W} = (\rho_{V}\ot \rho_{W})\circ \Delta;
            \]
        \item the tensor unit $(C, \epsilon)$;
        \item the left dual $(V^*,(\rho\circ S)^*)$ of $(V,\rho)$;
        \item the right dual $(V^*,(\rho\circ S^{-1})^*)$ of $(V,\rho)$.
    \end{itemize}
\end{definition}
\begin{example}
    Let $\C[G]$ be the finite-dimensional C$^\star $ group algebra, then $\Rep^\dagger(\C[G]) = \Rep^\dagger(G)$, i.e. the category of finite-dimensional unitary $G$-representations, is a unitary fusion category.
\end{example}

\begin{example} \label{ex:Fun(G)}
    Let $\Fun(G)$ be the finite-dimensional C$^\star $ function algebra, then $\Rep^\dagger(\Fun(G))$ is monoidally equivalent to $\Hilb_G$, the category of finite-dimensional $G$-graded Hilbert spaces. To show this, take a $\Fun(G)$-representation $(V,\rho)$, where $V$ is a finite dimensional Hilbert space and $\rho:\Fun(G)\rightarrow \End(V)$ is a $\star$-homomorphism. Then for each $g\in G$, $\rho(\delta_g):V\rightarrow V$ is an orthogonal projector, making $V$ a $G$-graded Hilbert space $V = \oplus_{g\in G}V_g$. The intertwiners from $(V,\rho_V)$ to $(W,\rho_W)$ are the $G$-grading preserving linear maps from $V$ to $W$. Moreover, we can check that the tensor product space $V \ot W$ of two $G$-graded Hilbert spaces is $G$-grading preserving. Take $V = \oplus_g V_g$ and $W = \oplus_g W_g$, we have
    \begin{align*}
        \rho_{V\ot W}(\delta_g) &=(\rho_V\ot\rho_W) (\Delta\delta_g)\\ &= \sum_{kl = g}\rho_{V}(\delta_k)\ot \rho_{W}(\delta_l),
    \end{align*}
     meaning
    \[
        (V\ot W)_g = \oplus_{kl = g} V_k\ot W_l.
    \]
    The representation of the left dual space $V^*$ is given by
    \[
        \rho^*(\delta_g) = \rho(S(\delta_g))^* =\rho(\delta_{g^{-1}}).
    \]
    Thus $(V^*)_g = (V_{g^{-1}})^*$. The right dual space is defined similarly.
\end{example}

\subsection{\texorpdfstring{$\sH$}{H}-symmetric 1D quantum system} \label{section:lattice definition}
 With a properly defined local symmetry, we can now assign how the symmetry acts on a 1D quantum system. We define an abstraction of a 1D Hopf algebra symmetric quantum system independent of any specific model. Our definition is based on \cite{ogata2021classificationgappedgroundstate,Lan_2024} but with an additional emphasis on the ordering structure of lattice sites. Such an ordering is crucial to a quantum system with non-invertible symmetries as the Hopf algebra representations in general do not admit a braiding structure. For example, the charge category $\Hilb_G$ of a $\Rep^\dagger(G)$-symmetric lattice system is not a braided category. Therefore, the lattice sites should not be labeled by elements of a set, but a one with a relational structure.
\begin{definition}
    A linearly order set $(S,\prec)$ is a set $S$ with a binary relation $\prec$ satisfying the following conditions:
    \begin{itemize}
        \item[--] $x \prec x$ for $x \in S$ (reflexive);
        \item[--] for each $x, y, z \in S$, $x \prec y$ and $y \prec z \Rightarrow x \prec z$;
        \item[--] for each $x,y \in S$, at least one of the following holds: $x \prec y$, $x=y$, $y \prec x$.
    \end{itemize}
\end{definition}
Then we define an abstraction of a 1D lattice system together with their relations and operations, which inherit from the properties of a linearly ordered set.

\begin{definition}\label{def:chain}
    A 1D finite chain $\BI$, the elements of which are viewed as lattice sites, is a linearly ordered set $(\{0,1,\cdots, N-1\},\leq)$ \footnote{For simplicity, we omit the ordering ``$\leq$'' label when there is no ambiguity.}. We further define
    \begin{itemize}
        \item[--] region: a region $X$ of $ \BI$ is a subset of $\BI$, either being $\emptyset$ or the one of the form $\{m,m+1,m+2,\cdots, m+n\}$. The set of regions of $\BI$ is denoted by $\fS_{\BI}$; 
        \item[--] adjacency: given $X, Y \in \fS_{\BI}$, the region $X$ is left-adjacent to $Y$ if $\max(X)+1 = \min(Y)$, and is right-adjacent if $Y$ is left-adjacent to $X$. We further define that the $\emptyset$ is adjacent to any region.
    \end{itemize}
\end{definition}
\begin{remark}
    For an empty chain, the underlying linearly order set is $\emptyset$, and $(\Z,\leq)$ for an infinite chain, $(\N,\leq)$ for a right-infinite chain, $(-\N,\leq)$ for a left-infinite chain. In the infinite cases, the set of finite regions is denoted by $\fS^f_{\BI}$.
\end{remark}

Given a 1D chain $\BI$ and a finite-dimensional Hopf C$^\star $-algebra $\sH$, a 1D quantum system with ``onsite'' symmetry $\sH$ is labeled by the tuple $(\BI,\sH, \Psi, \Phi)$, which consists of the following data: 
\begin{itemize}
    \item [--] symmetry charge/local Hilbert space $\Psi(X) = (V_{X},\rho_{X}) \in \Rep^\dagger(\sH)$ for each $X\in\fS^f_{\BI}$;
    \item [--] an interaction $\Phi(X)\in\End_{\Rep^\dagger(\sH)}((V_{X},\rho_{X}))$ (a Hermitian symmetric operator);
\end{itemize}
such that for $X,Y\in\fS^f_{\BI}$ with $X$ left-adjacent to $Y$, $\Psi(X\cup Y) = \Psi(X)\ot \Psi(Y)$.
   
The empty region $\emptyset\in \fS^f_{\BI}$ is necessarily associated with the trivial $\sH$-representation. If the single-site region $ \{i\}\in \fS^f_{\BI}$ is associated with the representation $(V_{\{i\}},\rho_{\{i\}})$, then the Hilbert space defined over the finite region $X = \{m+1\}\cup \{m+2\} \cup \cdots \cup\{m+p\}$ is
\[
    V_{X} = V_{\{m+1\}}\ot V_{\{m+2\}}\ot \cdots \ot V_{\{m+n\}},
\]
and the corresponding representation is
\[
    \rho_{X} = (\rho_{\{m+1\}}\ot \rho_{\{m+2\}}\ot \cdots \ot \rho_{\{m+n\}})\circ \Delta^{n}.
\]
The interaction $\Phi$ is called $m$-local if there is a universal bound $m\in \N$ such that $\Phi(X) = 0$ for all $|X|>m$.
For an $m$-local interaction $\Phi$ and finite region $\Lambda\in\fS^f_{\BI}$, let $\cJ = \{\Xi\in\fS^f_{\BI}:\Phi(\Xi)\neq 0,\ \Xi\cap \Lambda\neq \emptyset\}$, we call
\[
    \overline{\Lambda}:=\bigcup_{\Xi\in \cJ}\Xi
\]
the surrounding of $\Lambda$.

The total Hamiltonian is not mathematically defined for an infinite chain, thus we only define the local Hamiltonian on a finite region $\Lambda\in\fS^f_{\BI}$ as
\[
    (H_{\Phi})_{\Lambda} := \sum_{X\in\fS^f_{\Lambda}}\Phi(X).
\]

\subsection{Hopf algebra symmetry MPO}\label{sec:Hopf algebra symmetry MPO}
Tannaka duality in Theorem \ref{thm:tannaka} relates the data between 1) a fusion category together with a fiber functor and 2) a Hopf algebra. Here we use a dual version of the Tannaka duality. More precisely, given a unitary fusion category $\cC$ and a unitary fiber functor $f:\cC\rightarrow \Hilb$, the coend 
\[
    \sH := \int^{X\in\cC}\Hilb(fX,fX)\simeq \bigoplus_{s\in\Irr(\cC)}f(s)\otimes f(s)^*
\]
gives a finite-dimensional Hopf C$^\star $-algebra and conversely $\cC\cong \Rep^\dagger(\sH^*)$ as a unitary fusion category, see 
Theorem 5.4.1 in \cite{etingof2015tensor} and applications in \cite{Kitaev_2012,molnar2022matrixproductoperatoralgebras,couvreur2022matrixquantumgroupsmatrix,Garre_Rubio_2023,inamura202411dsptphasesfusion}.

The isomorphism above decomposes $\sH$ into direct sums of matrix coalgebras.
We introduce the MPO basis for $\sH$ as the basis of the multi-matrix coalgebra \cite{Kitaev_2012}, labeled by the triple $(s,\, \alpha,\, \beta)$:
\begin{multline*}
    \Gamma = \bigg\{ s^{\alpha\beta}|s \in\Irr(\cC),\ \alpha,\beta = 1,2,\cdots, \dim(fs)\bigg\}.
\end{multline*}
This choice of basis makes the coalgebra structure explicit:
\[
    \Delta(s^{\alpha\beta}) =\sum_{\gamma}s^{\alpha\gamma}\ot  s^{\gamma\beta}.
\]
Let $\{e_i\}$ be a basis of $\sH$ and $\{e^i\}$ its dual basis in $\sH^*$, the basis vector $s^{\alpha\beta}\in\Gamma$ can be obtained by the transformation
\[
    s^{\alpha\beta} := \sum_{i}\tau_{s}(e^i)^{\alpha\beta}\,e_i\,,
\]
where $(V_s,\tau_s)$ is an irreducible representation of $\sH^*$, $\alpha$ and $\beta$ index the rows and columns of the irreducible representation $\tau_s$, and $\tau_s(e_i)^{\alpha\beta}$ is the basis transformation coefficient.

For any object $X = (V_X,\tau_X) \in \Rep^\dagger(\sH^*)$, we define a corresponding element in $\sH$:
\begin{equation*}
    X^{\alpha\beta}:=\sum_{i}\tau_X\left(e^i\right)^{\alpha\beta}e_i \quad \alpha,\beta = 1,\cdots, \dim(V).       
\end{equation*}
And they satisfy the identity equation
\begin{align*}
    &\Delta(X^{\alpha\beta}) =\sum_{\gamma}X^{\alpha\gamma}\ot  X^{\gamma\beta},\\
    &X^{\alpha\beta}Y^{\gamma\delta} = (X\ot Y)^{(\alpha\gamma)(\beta\delta)},\\
    &S(X^{\alpha\beta}) = (X^*)^{\beta\alpha}.
\end{align*}

We use the MPO basis to construct the $\Rep^\dagger(\sH^*)$ symmetry operator. For $\Lambda \in \fS^f_{\BI}$, let $\Psi(\Lambda) = (V_\Lambda, \rho_\Lambda)$. The local tensor of the $X$-MPO, denoted as $T^\Lambda_X$, $X \in \Rep^\dagger(\sH^*)$, on region $\Lambda$ is
\[
    (T^\Lambda_X)^{\alpha\beta}  :=  \rho_\Lambda\left(X^{\alpha\beta}\right) = \sum_{i}\tau_X(e^i)^{\alpha\beta}\rho_\Lambda(e_i).
\]
When $X$ is simple, we call such an MPO simple.   

Furthermore, the $\Rep^\dagger(\sH^*)$-symmetry MPO satisfies the following properties
\begin{enumerate}
    \item Extension: for $X\in\Rep^\dagger(\sH^*)$ and $\Lambda,\Xi\in\fS^f_{\BI}$, if $\Lambda$ is left-adjacent to $\Xi$, then
        \[
        \sum_{\beta}\left(T^{\Lambda}_X\right)^{\alpha\beta}\ot \left(T^{\Xi}_X\right)^{\beta\gamma} = \left(T^{\Lambda\cup \Xi}_X\right)^{\alpha\gamma};
        \]
    \item Composition: for $X,Y\in\Rep^\dagger(\sH^*)$ and $\Lambda\in\fS^f_{\BI}$,
        \[
        (T_{X}^\Lambda)^{\alpha\beta} (T_{Y}^\Lambda)^{\gamma\delta} = (T_{X\ot Y}^{\Lambda})^{(\alpha\gamma)(\beta\delta)}.
        \]
    \item Naturality: for $O:U\rightarrow V$ in $\Rep^\dagger(\sH)$, $f:X\rightarrow Y$ in $\Rep^\dagger(\sH^*)$, we have
        \[
        \sum_{\beta}(T_{Y}^{V})^{\alpha\beta}O \times f^{\beta\gamma} = \sum_{\beta} f^{\alpha\beta}\times O(T_{X}^{U})^{\beta\gamma}.
        \]
    \item Dualizability:
        \[
            \sum_{\beta}(T^V_{X^*})^{\beta\alpha}(T^V_X)^{\beta\gamma} = \delta_{\alpha\gamma}\one_V = \sum_{\beta}(T^V_{X})^{\alpha\beta}(T^V_{X^*})^{\gamma\beta}.
        \]
    \item The virtual bond dimension of the $X$-MPO equals the quantum dimension $d_X$ in $\Rep^\dagger(\sH^*)$\,,
\end{enumerate}
where we use $T_X^V$ to denote the local tensor of MPO on space $V$. The above properties can be proved by direct calculations. See a detailed discussion on the relation between MPOs and weak Hopf algebra in \cite{molnar2022matrixproductoperatoralgebras}.
On the other hand, these properties align with the intuition that a $\Rep^\dagger(\sH^*)$-symmetry operator should be a global operator indexed by objects in $\Rep^\dagger(\sH^*)$ satisfying the same fusion rule as objects of $\Rep^\dagger(\sH^*)$. 
\begin{example}
     For $\Rep^\dagger(G)$-symmetry MPO, i.e. $\sH=\Fun(G)$, $\sH^* = \mathbb{C}[G]$, recall Example \ref{ex:Fun(G)}, $\Rep^\dagger(\Fun(G))$ is monoidally equivalent to $\Hilb_G$.
     For any $X = (W,\tau)\in\Rep^\dagger(G)$, the local tensor $T_{X}$ on the charge space $V$ is \cite{inamura202411dsptphasesfusion}
    \[
        (T^V_X)^{\alpha\beta}  := \rho_V(X^{\alpha\beta}) = \sum_{g}\tau_X\left(g\right)^{\alpha\beta} \times \rho_V(\delta_g).
    \]
\end{example}
We give explicit local tensors for $\Rep^\dagger(\sH_8)$-MPOs in Appendix \ref{app:H8}. 

\section{Q-system model of \texorpdfstring{$\Rep^\dagger(\sH^*)$}{Rep(H*)} symmetry protected topological phases}\label{section:QSys}
In this section, we review the fixed-point lattice model constructed in \cite{Lan_2024}, see also \cite{Inamura_2022}.

Let $(A,m,\iota)$ be a Q-system in $\mathcal{C}_{\Hilb_f}^\vee$, $(M,m_M)$ an isometric left $A$-module, $(N,{}_{N}m)$ an isometric right $A$-module, 
\begin{itemize}
    \item $(\Z, \sH, A, 1- m^\dagger m)$ is called the Q-system model on an infinite chain, where the local Hilbert space is $\Psi(\{i\}) = A$ and the nearest-neighbor interaction is $\Phi(\{i,i+1\}) = 1-m^\dagger m$;
    \item $(\N, \sH, A, 1-m^\dagger m, M, 1-m_{M}^\dagger \circ m_{M})$ is called the Q-system model on a right half-infinite chain, where the local Hilbert space is $\Psi(\{0\}) = M$, $\Psi(\{i\}) = A$ for all $i>0$, and the nearest neighbor interaction is $\Phi(\{0,1\}) = 1-m_{M}^\dagger \circ m_{M}$ and $\Phi(\{i,i+1\}) = 1-m^\dagger \circ m$ for all $i>0$. And the definition of the left half-infinite chain is similar;
    \item $(\{1,2,\cdots,n\}, \Psi, \Phi)$ with 
    \begin{itemize}
        \item $\Psi(\{1\}) = M$, $\Psi(\{i\}) = A,\ \forall 2\leq i\leq n-1$ and $\Psi(\{n-1\}) = N$;
        \item $\Phi(\{0,1\}) = 1-m_{M}^\dagger \circ m_{M}$, $\Phi(\{i,i+1\}) = 1- m^\dagger m,\ \forall 2\leq i\leq n-2$ and $\Phi(\{n-1,n\}) = 1-{}_{M}m^\dagger \circ{}_{M}m$
    \end{itemize}
    defines a Q-system model on a finite chain. The Q-system $A$ is called its bulk, the left $A$-module $M$ is called the left boundary condition, and the right $A$-module $N$ is called the right boundary condition.
\end{itemize}
We may directly write the total Hilbert space as $\cH = \bigotimes_{i \in \Z} A$ and the total Hamiltonian as $H = \sum_{i \in \Z} \left(1 - m_{i,i+1}^\dagger m_{i,i+1}\right)$ for the Q-system model on an infinite chain if the mathematical rigor is not emphasized, and the case for the half-infinite chain is similar. We will also drop $\sH$ when there is no symmetry imposed.

By \cite{Lan_2024}, each local Hamiltonian of a Q-system model is a commuting projector. Moreover, for the Q-system model on a finite chain with bulk $A$, left boundary condition $M$, and right boundary condition $N$, the ground state subspace of it is $M\ot[A]N$.

For the ground state subspace of the infinite chain, Theorem 5.14 of \cite{Lan_2024} claims that the ground state subspace of the Q-system model on an infinite chain is the Q-system $A$. However, this approach has issues. Firstly, an infinite chain is fundamentally different from a finite chain; the product of projectors doesn't make sense in an infinite context and requires careful handling. Secondly, the ground state degeneracy on an infinite chain is an invariant of a phase, classified by the Morita class of the Q-system. Therefore, the ground state subspace should not depend on the choice of the representative of the Morita class.

\section{Characterization of SPT phases} \label{sec:gsd}
In this section, we address the characterization of an SPT phase with fusion category symmetry $(\mathcal{C},f)$. The defining feature of an SPT phase is the unique ground state on an infinite chain. In the operator algebraic framework, we formalize its defining feature by showing that a Q-system model describes a symmetric phase only if the Q-system in the charge category is a matrix algebra in $\Hilb$. 

\subsection{Ground state on an open chain} \label{sec:O-type gs}
A defining feature of the 1D SPT phase is the unique symmetric ground state on the infinite chain. The operator algebraic framework \cite{ogata2021classificationgappedgroundstate,Kapustin_2021,Haag:1963dh, bratteli1979operator,BratteliRobinson1981} allows a rigorous treatment of infinite lattices at the thermodynamic limit by using Heisenberg picture to encode a system's dynamics and calculate the time evolution of local observables. This locality ensures that only the time evolution near the region where the operator acts is significant, making the treatment of infinite lattice dynamics rigorous. In the operator algebraic framework, a quantum system is described by a pair $(\fA_{\BI},\tau)$, where $\fA_{\BI}$ is the C$^\star $-algebra formed by the physical observables and $\tau$, the strongly continuous one-parameter group of automorphisms, defines the dynamics on $\fA_{\BI}$. A state $\omega$ on $\fA_{\BI}$ is a positive linear functional, i.e. $\omega(A^\dagger A) \geq 0$ for $A \in \fA_{\BI}$ and is normalized $\omega(\mathbb{1})=1$.

\begin{definition}[$O$-type ground state, see for example \cite{ogata2021classificationgappedgroundstate}]\label{def:groundstate}
    Let $(\BI,\sH, \Psi, \Phi)$ be a 1D quantum system. If $\Phi$ is a uniformly bounded and finite-range interaction, the limit
    \[
        \tau_{\Phi}^t(A) := \lim_{\Lambda\rightarrow \BI}\ee^{\ii t(H_{\Phi})_{\Lambda}}A\ee^{-\ii t(H_{\Phi})_{\Lambda}},\quad \forall t \in \R,\quad A\in \fA_{\BI}
    \]
    exists and defines a dynamics $\tau_{\Phi}$ on $\fA_{\BI}$. Let $\delta_{\Phi}$ be the generator of $\tau_{\Phi}$. A state on $\fA_{\BI}$ is called a $\tau_{\Phi}$-ground state if the inequality 
    \[
        -\ii \,\omega(A^\dagger \delta_{\Phi}(A))\geq 0
    \]
    holds for all $A$ in the domain $\cD(\delta_{\Phi})$ of the generator $\delta_{\Phi}$.
\end{definition}

Note that the operator algebraic ground state, denoted as $O$-type ground state, can differ from the state that minimizes the energy of the system. The latter is a more well-known definition of the ground state, which we refer to as $H$-type ground state. These two types of ground state differ in the infinite-dimensional system, where superselection can be sharply defined \cite{Kong_2022}. $O$-type ground state only measures the lowest energy state within a certain superselection sector. For example, consider the 1D quantum Ising model in the symmetry breaking phase
\[
    H_{\Lambda} = \sum_{\{i,i+1\}\subset\Lambda}\frac{1}{2}\left(\one - Z_iZ_{i+1}\right)
\]
for any finite region $\Lambda \in\fS_{\Z}^f$. One can verify that $\omega_0(O)=\la \cdots 000 \cdots|O|\cdots 000\cdots\ra$, $\omega_1(O)=\la \cdots 111 \cdots|O|\cdots 111\cdots\ra$ and the domain wall state $\omega_{01}(O)=\la \cdots 000111 \cdots|O|\cdots 000111\cdots\ra$ each satisfies the $O$-type ground state definition, while contradicts the $H$-type ground state definition as the domain wall state is not the so-called lowest-energy state of $H_\Lambda$ when compared with the other two. For a detailed discussion, see Appendix \ref{app:O-type ground state}.

In fact, for a 1D quantum system with commuting projector interaction, i.e.
\begin{alignat*}{2}
    &[\Phi(X),\Phi(Y)] = 0,\quad &&\forall X, Y\in \fS^f_\BI,\\
    &\Phi(X)^2 = \Phi(X),\quad &&\forall X \in \fS^f_\BI,
\end{alignat*}
the $H$-type ground state is defined to be the one satisfying $\omega(\Phi(X)) = 0$ for any $X\in \fS^f_\BI$. In this case, one can prove that the $H$-type ground state is an $O$-type ground state. 
By the Cauchy-Schwarz inequality: for any $B\in \fA_{\Lambda}$ with $\Lambda \in \fS_{\Z}^{f}$, one has \cite{Alicki_2007}
\begin{gather*}
    |\omega(\Phi(X)B)|^2\leq \omega(\Phi(X) \Phi(X)^\dagger)\omega(B^\dagger B) = 0,\\
    |\omega(B\Phi(X))|^2\leq \omega(\Phi(X)^\dagger \Phi(X))\omega(B B^\dagger) = 0.
\end{gather*}

Thus
\begin{equation}\label{eq:rg}
    \omega(\Phi(X)B) = \omega(B\Phi(X)) = 0.
\end{equation}
The ground state condition reduces to
\[
    \omega(B^\dagger H_{\overline{\Lambda}}B) - \omega(B^\dagger BH_{\overline{\Lambda}}) = \omega(B^\dagger H_{\overline{\Lambda}}B) \geq 0,
\]
which holds by the positivity of $H_{\overline{\Lambda}}$.

For the purpose of this paper, we are only interested in the SPT phase, therefore $O$-type and $H$-type ground states are identical\footnote{Hereinafter, we do not distinguish between $O$-type and $H$-type ground states when referring to ``unique ground state'' associated with an SPT phase.}. Let $\chargef:\Rep^\dagger(\sH)\rightarrow \Hilb$ be the forgetful functor. Below, we will make this statement rigorous by showing that for the fusion category symmetry $(\cC,f)$
\begin{itemize}
    \item[--] when the Q-system $(\chargef A,\chargef m,\chargef \iota)$ is the $n$-by-$n$ dimensional matrix algebra in $\Hilb$, $O$-type and $H$-type ground states for the Q-system model $(\Z, \sH, A, 1- m^\dagger m)$ are identical;
    \item[--] the Q-system model $(\Z, \sH, A, 1- m^\dagger m)$ describes a symmetric phase only if $\chargef A$ is a matrix algebra in $\Hilb$.
\end{itemize}

\begin{proposition}\label{prop:matrix algebra}
    Let $(A,m,\iota)$ be the $n$-by-$n$ dimensional matrix algebra in $\Hilb$ and $(\Z, A, p)$ with $p = 1- m^\dagger m$ a quantum system on infinite chain without symmetry, the state $\omega:\fA_{\Z}\rightarrow \C$ is an $O$-type ground state if and only if $\omega(p_{i,i+1})=0$ for any $i\in\Z$.
\end{proposition}
\begin{proof}
    It suffices to prove the necessity. Given the decomposition $\{P_{\alpha}\}_{\alpha = 0}^{n^2-1}$:
     \begin{align*}
        &P_{\alpha}\in\hom_{A|A}(A\ot A, A),\quad \alpha = 0,1,\cdots, n^2-1,\\
        &P_{\alpha}P_{\beta}^\dagger = \delta_{\alpha\beta}\id_{A}, \quad P_0= m,
    \end{align*}
    such that
    \[
        \id_{A\ot A} = \sum_{\alpha = 0}^{n^2 - 1}P_{\alpha}^\dagger P_{\alpha},
    \] 
    we consider $O_\alpha = m^\dagger P_\alpha$. The ground state condition gives
    \begin{align*}
        &\omega\left(\tau_{i,i+1}(P_\alpha^\dagger m(1-m^\dagger m)m^\dagger P_\alpha)\right)\\
        &- \omega\left(\tau_{i,i+1}(P_\alpha^\dagger mm^\dagger P_\alpha(1-m^\dagger m))\right)\\
        &=-\omega(\tau_{i,i+1}(P_\alpha^\dagger P_\alpha))\geq 0, \quad \forall \alpha = 1,2,\cdots, n^2-1, 
    \end{align*}
    where $\tau_{i,i+1}:\End(A^2)\rightarrow \fA_{\{i,i+1\}}$. 
    Here we use $m(1-m^\dagger m)m^\dagger = 0$, $P_\alpha m^\dagger = 0$ and $mm^\dagger = \id_A$.
    By the positivity of $P^\dagger_\alpha P_\alpha$, we have
    \[
        \omega(\tau_{i,i+1}(P^\dagger_\alpha P_\alpha)) = 0,\quad \forall \alpha = 1,2,\cdots, n^2-1.
    \]
    By the decomposition condition, we have
    \begin{align*}
        \omega(p_{i,i+1}) &= \omega(1-m_{i,i+1}^\dagger m_{i,i+1})\\ &= \sum_{\alpha = 1}^{n^2-1}\omega\left(\tau_{i,i+1}(P_{\alpha}^\dagger P_\alpha)\right) = 0.
    \end{align*}
\end{proof}

Before establishing the connection between the condition that $(\Z, \sH, A, 1 - m^\dagger m)$ realizes the symmetric phase and $\chargef A$ being the matrix algebra in $\Hilb$, we first demonstrate the bijection between the space of $H$-type ground states and the space of states on the algebra of fixed-point operators. Such a connection arises from the property of the $H$-type ground state, which helps define a renormalization step on the space of local operators. 

Take a finite subregion $\Xi \in \fS_{\Z}^f$ and denote $m_\Xi:\,\chargef A^{\otimes|\Xi|} \rightarrow \chargef A$ as the multiplication map from $|\Xi|$ numbers of $\chargef A$ to a single $\chargef A$. By definition of the $H$-type ground state, we have
\[
    \omega(O) = \omega(m_{\Xi}^\dagger\, m_{\Xi}\,O \,m_{\Xi}^\dagger\, m_{\Xi}),\quad \forall O \in \fA_\Lambda,\quad \forall\Lambda\in\fS^f_\Z.
\]
Thus we define the screening map
\begin{align*}
    \scr:\fA_\Lambda &\rightarrow \fA_{\{0\}},\\
    O&\mapsto \scr(O):=m_{\Xi}\, O\, m_{\Xi}^\dagger,
\end{align*}
where $\Xi$ is the smallest finite region that includes $\overline{\Lambda}$ and $\overline{\{0\}}$ \footnote{
The fixed-point nature of the Q-system model ensures that the screening by any number of $\chargef A$ can be reduced to a single $\chargef A$ screening, thus it suffices for us to define the single $\chargef A$ screening. },
which is graphically represented as
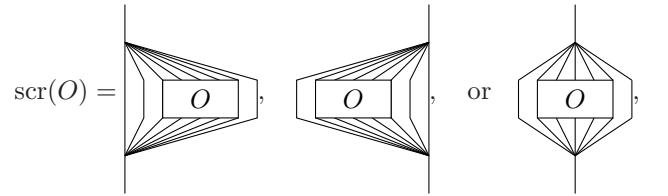
\[
    \scr(O) = \diagram{1}{
    \draw (0,0) -- (0,2.5);
        \foreach \i in {1,...,7}
        {
            \draw (0,.5) -- (\i*.25, 1) -- (\i*.25,1.5) -- (0,2);
        }
        \filldraw[fill = white] (.5,1) rectangle node{$O$} (1.5,1.5);
    },\quad 
    \diagram{1}{
        \draw (0,0) -- (0,2.5);
        \foreach \i in {1,...,7}
        {
            \draw (0,.5) -- (-\i*.25, 1) -- (-\i*.25,1.5) -- (0,2);
        }
        \filldraw[fill = white] (-.5,1) rectangle node{$O$} (-1.5,1.5);
    },\quad \text{or}\quad
    \diagram{1}{
        \draw (0,0) rectangle node{$O$} (1,.5);
        \draw (0,.5) -- (.5,1);
        \draw (.25,.5) -- (.5,1);
        \draw (.5,.5) -- (.5,1);
        \draw (.75,.5) -- (.5,1);
        \draw (1,.5) -- (.5,1);
        \draw (-.25,.5) -- (.5,1);
        \draw (1.25,.5) -- (.5,1);
        \draw (0,0) -- (.5,-.5);
        \draw (.25,0) -- (.5,-.5);
        \draw (.5,0) -- (.5,-.5);
        \draw (.75,0) -- (.5,-.5);
        \draw (1,0) -- (.5,-.5);
        \draw (-.25,0) -- (.5,-.5);
        \draw (1.25,0) -- (.5,-.5);
        \draw (-.25,0) -- (-.25,.5);
        \draw (1.25,0) -- (1.25,.5);
        \draw (.5,1) -- (.5,1.5);
        \draw (.5,-1) -- (.5,-.5);
    },
\]
depending on the position of $\Lambda$.

This map physically corresponds to the screening of the local operators by the condensed symmetry charges in the low energy subspace, such a screening defines a renormalization step on the space of local operators. The insertion of observables to the tensor network representation of the partition function is the so-called impurity tensor in the literature \cite{Evenbly_2016,wei2023tensornetworkrenormalizationapplication}. Along the tensor network renormalization process, the impurity tensor is constantly acted by tensors in the environment, and finally most details are screened out and only left with some finite data. Here, for a gapped system, the screening map is an idempotent $\scr^2 = \scr$, whose image $\im (\scr)=:\fZ(\chargef A)$ is called the algebra of fixed-point operators in the bulk.
Using the property of the Q-system, for $H$-type ground state $\omega$, we have
\[
    \omega = \omega\circ \scr.
\]
Thus the $H$-type ground state space is isomorphic to the space of states on $\fZ(\chargef A)$.

An explicit form of space of states on $\fZ(\chargef A)$ can be given.
In Proposition \ref{prop:fpo eq bimod map}, we show
$\fZ(\chargef A)\simeq \hom_{\chargef A|\chargef A}(\chargef A,\chargef A)\simeq Z(\chargef A)$.
We thus decompose the algebra as $\fZ(\chargef A)\simeq \oplus_{\alpha}\C p_\alpha$, with $\{p_\alpha\}$ being the set of orthogonal projectors. Suppose for any $O \in \fA_{\Z}$, $\scr(O) = \sum_\alpha c_\alpha p_\alpha$, then $\omega(O)$ is fully determined by $\{\omega\left(p_{\alpha,\{0\}}\right)\}_{\alpha}$, i.e. $\omega(O) = \sum_\alpha c_\alpha\, \omega (p_{\alpha,\{0\}})$. On the other hand, the $H$-type ground state of the Q-system model $(\Z, \sH, A, 1-m^\dagger m)$ form a simplex, the extreme points of which are $\omega_{\alpha}$, satisfying $\omega_{\alpha}(m_{i,i+1}^\dagger m_{i,i+1}) = 1$ for any $i$, and $\omega_{\alpha}\left(p_{\beta,\{0\}}\right) = \delta_{\alpha\beta}$. This establishes the bijection between the space of $H$-type ground states and the space of states on $\fZ(\chargef A)$ (See Appendix \ref{app:properties of fixed-point operators} for detailed proofs).

\begin{proposition}
    Let $\sH$ be a finite dimensional Hopf C$^\star $-algebra and $(A,m,\iota)$ a Q-system in $\Rep^\dagger(\sH)$. The Q-system model $(\Z,\sH, A, 1-m^\dagger m)$ has the unique ground state if and only if $\fZ(\chargef A)\simeq \C$, or equivalently, $\chargef A$ is a matrix algebra in $\Hilb$.
\end{proposition}
\begin{proof}
    Suppose the Q-system model has a unique ground state, then the $O$-type ground state should just be the $H$-type ground state. By proposition \ref{prop:lowest energy state}, the space of $H$-type ground states is isomorphic to the state space of $\fZ(\chargef A)$. Thus $\fZ(\chargef A)\simeq \C$.

    Conversely, suppose $\chargef A$ is a matrix algebra in $\Hilb$, by propositions \ref{prop:lowest energy state} and \ref{prop:matrix algebra}, its $H$-type ground state is unique and is also the unique $O$-type ground state.
\end{proof}

We obtain:
\begin{mdframed}
    The Q-system model $(\Z, \sH, A, 1-m^\dagger m)$ is in a symmetric phase only if $\chargef A$ is a matrix algebra in $\Hilb$.
\end{mdframed}

In our construction in Section \ref{sec:D8 lattices}, we choose $\cC_{\Hilb_f}^\vee = \Hilb_{D_8}$, $\omega = -1 \in H^2(\Z_2\times \Z_2, U(1))$ and define the Q-systems as:
\begin{itemize}
    \item $A_0 = \C e $,
    \item $A_1 = \C^\omega\langle r^2, s \rangle$,
    \item $A_2 = \C^\omega\langle r^2, sr \rangle$.
\end{itemize}
These Q-systems are all simple in $\Hilb$. In fact, they represent only three Q-systems (up to Morita equivalence) in $\Hilb_{D_8}$ and become simple Q-systems in $\Hilb$, i.e. when the $D_8$-grading structure is forgotten.

\subsection{Examples of matrix algebras in \texorpdfstring{$\Hilb_{D_8}$ and $\Hilb_{S_3 \times \Z_3}$}{Rep(D8) and Rep(S3xZ3)}} \label{sec:matrix}
Q-systems in $\Hilb_G$ are classified into Morita classes, each determined by a pair $(H, \psi)$, where $H \subset G$ is a subgroup up to conjugation and $\psi \in H^2(H, U(1))$ is a 2-cocycle up to conjugation $\psi(g,h)\sim \psi(kgk^{-1},khk^{-1})$ \cite{etingof2015tensor}, and can be represented by the twisted group algebra $\C^\psi[H]$.

The subgroups of $D_8$ up to conjugation are:
\begin{align*}
     &\la e\ra = \Z_1,\ \la r^2\ra = \Z_2^{(0)},\ \la s\ra = \Z_2^{(1)},\ \la sr\ra = \Z_2^{(2)},\\
     &\la r^2,sr\ra = (\Z_2\times \Z_2)^{(2)},\ \la r^2,s\ra = (\Z_2\times \Z_2)^{(1)},\\
     & \la r\ra = \Z_4,\ D_8\,.
\end{align*}
Among these, two subgroups isomorphic to $\Z_2 \times \Z_2$ admit non-trivial cocycles, thus can be matrix algebras in $\Hilb$. Additionally, $D_8$ itself admits non-trivial cocycles since $H^2(D_8, U(1)) \simeq \Z_2$. However, the twisted group algebra $\chargef\C^\psi[D_8]$ cannot be simple in $\Hilb$, even in the presence of a non-trivial cocycle, simply because there is no matrix algebra with dimension $8$. By central extension of $D_8$ to $D_{16}$, we know that $\chargef\C^{\psi}[D_8] = \mathbb{M}_2\oplus\mathbb{M}_2$ \cite{groupprops_d8}.
Thus, the algebras $A_0$, $A_1$, and $A_2$ are the only ones that satisfy the condition of the symmetric phase. These $\Rep^\dagger(D_8)$-symmetric phases can, in fact, be understood as gauging the $D_8$ symmetry of different SSB phases with $D_8$ symmetry: the $D_8$ fully SSB phase, the partial SSB partial SPT phase with a residual $(\mathbb{Z}_2 \times \mathbb{Z}_2)^{(1)}$ symmetry, and the partial SSB partial SPT phase with a residual $(\mathbb{Z}_2 \times \mathbb{Z}_2)^{(2)}$ symmetry.

More generally, we consider the twisted group algebra $\mathbb{C}^\psi[\mathbb{Z}_p \times \mathbb{Z}_p]$, where $\psi \in H^2(\mathbb{Z}_p \times \mathbb{Z}_p, U(1)) \simeq \mathbb{Z}_p$. We choose the ``bi-character gauge'' for the cocycles:
\[
    \psi(g,h) = \lambda^{g_2 h_1},
\]
where $\lambda$ is a $p$-th root of unity. For $\lambda \neq 1$ and $p$ being prime, it is straightforward to verify that 
\[
    Z(\chargef \mathbb{C}^\psi[\mathbb{Z}_p \times \mathbb{Z}_p]) \simeq \mathbb{C}.
\]
Thus, $\chargef\mathbb{C}^\psi[\mathbb{Z}_p \times \mathbb{Z}_p]$ is a matrix algebra in $\Hilb$ whenever $p$ is prime and $\psi$ is non-trivial.

Next we consider $\Rep^\dagger(S_3 \times \mathbb{Z}_3)$-symmetric phases, which has been discussed in the talk \cite{wen2024spt}. Up to conjugation, the possible subgroups are:
\begin{align*}
    &\mathbb{Z}_1, \quad \mathbb{Z}_2, \quad \mathbb{Z}_3^{(1)}, \quad \mathbb{Z}_3^{(2)}, \quad \mathbb{Z}_3^{(\mathrm{diag})}, \\
    &\mathbb{Z}_2\times \Z_3, \quad S_3, \quad \mathbb{Z}_3 \times \mathbb{Z}_3, \quad S_3 \times \mathbb{Z}_3\,.
\end{align*}
Since all cohomology groups of these subgroups are trivial except for $\mathbb{Z}_3 \times \mathbb{Z}_3$, the only possibility for symmetric phases arises when we consider the following twisted group algebras:
\[
    A_0 := \mathbb{C}_e, \quad A_\pm := \mathbb{C}^{\psi_\pm}[\mathbb{Z}_3 \times \mathbb{Z}_3],
\]
where
\[
    \psi_\pm(g,h) = e^{\pm \frac{2\pi i}{3} g_2 h_1}.
\]

While we notice that two non-trivial 2-cocycles are identified as they are related by conjugating the non-trivial $\mathbb{Z}_2$ element $s$. The 2-cocycle $\psi_-$ is related to $\psi_+$ via conjugation
\[
\psi_-(g,h) = \psi_+(sgs^{-1},shs^{-1}).
\]
Thus
\[
\mathbb{C}^{\psi_+}[\mathbb{Z}_3 \times \mathbb{Z}_3] \simeq_{\mathrm{Mor}}\mathbb{C}^{\psi_-}[\mathbb{Z}_3 \times \mathbb{Z}_3].
\]
There are two symmetric phases corresponding to:
\[
    A_0, \quad A_{+/-},
\]
where $A_0$ describes the trivial SPT phase and $A_{+/-}$ for the non-trivial SPT.

\subsection{Edge mode}
Edge modes are associated with the ``ground state subspace'' of a half-infinite chain. In this section, we focus on the $H$-type ground state subspace of the Q-system model on a half-infinite chain. Specifically, we consider the right half-infinite chain, noting that the analysis for the left half-infinite chain is analogous.

Similar to the discussion in the previous section, we consider the map of the $H$-type ground state on any local operator for any $\Lambda\in \fS^f_{\N}$. Take a finite subregion $\Xi = \{0,1,\cdots,\max(\Lambda)+1\}$ and denote $m_{\Xi}:M\ot A^{\otimes (|\Xi|-1)}\rightarrow M$ as the multiplication from $M\ot A^{\otimes (|\Xi|-1)}$ to $M$, the module associativity implies the notation $m_\Xi$ is independent of the multiplication order and thus well-defined. Let $O\in\fA_{\Lambda}$, by subsequently apply the Equation \eqref{eq:rg}, we have
\begin{equation*}\label{eq:fpobdy}
    \omega(O) = \omega(m_\Xi^\dagger\, m_\Xi\, O\, m_\Xi^\dagger\, m_\Xi).
\end{equation*}
We thus define the screening map
\begin{align*}
    \scr:\fA_{\Lambda}&\rightarrow \fA_{\{0\}},\\
    O &\mapsto \scr(O):= m_\Xi Om_\Xi^\dagger,
\end{align*}
and is graphically represented as
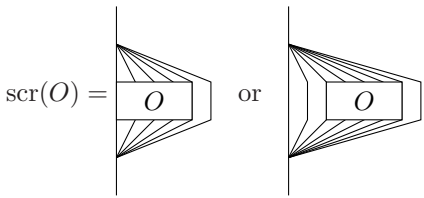
\[
    \scr(O) = 
    \diagram{1}{
        \draw (0,0) -- (0,.5);
        \foreach \i in {0,1,...,5}
        {
            \draw (0,.5) -- (\i*.25, 1) -- (\i*.25,1.5) -- (0,2);
        }
        \draw (0,2) -- (0,2.5);
        \filldraw[fill = white] (0,1) rectangle node{$O$} (1,1.5);
    }\quad 
    \text{or}\quad
    \diagram{1}{
        \draw (0,0) -- (0,.5);
        \foreach \i in {0,1,...,7}
        {
            \draw (0,.5) -- (\i*.25, 1) -- (\i*.25,1.5) -- (0,2);
        }
        \draw (0,2) -- (0,2.5);
        \filldraw[fill = white] (.5,1) rectangle node{$O$} (1.5,1.5);
    }
\]
depending on the position of $\Lambda$.

The screening map physically corresponds to the screening of the boundary local operators by the condensed symmetry charges in the low energy subspace in the bulk, such a screening defines a renormalization step on the space of boundary local operators. The $\scr$ is an idempotent and its image $\im(\scr)=:\fZ_{\chargef A}(\chargef M)$ is called the algebra of fixed-point operators on the edge. 
By the property of the Q-system model, we also have
\[
    \omega = \omega \circ \scr.
\]
Thus the $H$-type ground state space is isomorphic to the space of states on $\fZ_{\chargef A}(\chargef M)$.

\begin{proposition}\label{prop:bdy fpo eq rmod map}
    Let $(A,m,\iota)$ be a Q-system in $\Hilb$ and $(M,m_M)$ an isometric right $A$-module. Then $\fZ_{\chargef A}(\chargef M):=\im(\scr)\simeq \End_{\chargef A}(\chargef M)$.
\end{proposition}

When $\chargef A\simeq W\ot W^*$ is a matrix algebra, take $A$ itself as the boundary condition, then the algebra of fixed-point operators on the edge is $\End_{\chargef A}(\chargef A) = \End(W)$. The Hilbert space $W$ is the so-called edge mode. In other words, the edge mode is the effective ground state subspace of the half-infinity chain, which can be reconstructed from the representation of the algebra of boundary fixed-point operators:
\begin{mdframed}
    The edge mode of $(\N,\sH, A, 1-m^\dagger m, M, 1-m_M^\dagger m_M)$ is the minimal faithful module over the algebra of fixed-point operators on edge $\fZ_{\chargef A}(\chargef M)$.
\end{mdframed}
There is a natural bulk-to-edge map that sends the fixed-point operators in the bulk to that on the edge
\begin{align*}
    C_r: \fZ(\chargef A)&\rightarrow \fZ_{\chargef A}(\chargef M),\\
    f&\mapsto m_M(\id_M\ot f)m_M^\dagger,
\end{align*}
graphically represented as
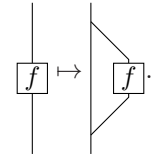
\[
    \diagram{1}{
        \draw (0,-1) -- (0,1);
        \filldraw[fill = white] (-.2,-.2) rectangle node{$f$} (.2,.2);
    }
    \mapsto 
    \diagram{1}{
        \draw (0,-1) -- (0,1);
        \draw (0,-.75) -- (.5,-.25) -- (.5,.25) --  (0,.75);
        \filldraw[fill = white] (-.2+.5,-.2) rectangle node{$f$} (.2+.5,.2);
    }.
\]

We can prove that $C_r$ factors through the center of $\fZ_{\chargef A}(\chargef M)$, i.e. all the elements in image of $C_r$ commute with those in $\fZ_{\chargef A}(\chargef M)$:
\[
    C_r = \bigg(\fZ(\chargef A)\rightarrow Z(\fZ_{\chargef A}(\chargef M))\hookrightarrow \fZ_{\chargef A}(\chargef M)\bigg).
\]
We prove it by considering the graph for all $f\in\fZ(\chargef A)$ and $g\in \fZ_{\chargef A}(\chargef M)$:
\[
    \diagram{1}{
        \draw (0,-2) -- (0,1);
        \draw (0,-.75) -- (.5,-.25) -- (.5,.25) --  (0,.75);
        \filldraw[fill = white] (-.2+.5,-.2) rectangle node{$f$} (.2+.5,.2);
        \filldraw[fill = white] (-.2,-1.5) rectangle node{$g$} (.2,-1.1);
    }
    =
    \diagram{1}{
        \draw (0,-2) -- (0,1);
        \draw (0,-1.75) -- (.5,-1.25) -- (.5,.25) --  (0,.75);
        \filldraw[fill = white] (-.2+.5,-.7) rectangle node{$f$} (.2+.5,-.3);
        \filldraw[fill = white] (-.2,-.7) rectangle node{$g$} (.2,-.3);
    }
    =
    \diagram{1}{
        \draw (0,-2) -- (0,1);
        \draw (0,-.75-1) -- (.5,-.25-1) -- (.5,.25-1) --  (0,.75-1);
        \filldraw[fill = white] (-.2+.5,-.2-1) rectangle node{$f$} (.2+.5,.2-1);
        \filldraw[fill = white] (-.2,.1) rectangle node{$g$} (.2,.5);
    }.
\]
This is a special case of the general mathematical principle called the bulk-edge correspondence for topological phases proposed in \cite{Kitaev_2012,Kong_2017}.

In Appendix \ref{app:generic edge modes}, the general edge mode $M$ between $A$ and $B$ phases is also considered.

The stability of the ground-state subspace requires that a symmetric operator, when restricted to this subspace, be proportional to the identity operator. 
Therefore, the edge mode of $(\N, \sH, A, 1-m^\dagger m, M,1-m_M^\dagger m_M)$ is stable against symmetric perturbations if and only if $\dim \left(\cC_{\Hilb_f}^\vee\right)_A\left(M,M\right) = 1$, i.e. $M$ is simple in $\left(\cC_{\Hilb_f}^\vee\right)_{A}$. And the edge mode $N$ between $A$-phase and $B$-phase is stable against symmetric perturbations if and only if $N$ is simple in ${}_{A}\left(\cC_{\Hilb_f}^\vee\right)_{B}$.

\section{Categorical aspects} \label{sec:auto}
Here, we investigate how the monoidal equivalence in the symmetry category $\cC$ induces an algebra in the charge category $\cC_{\Hilb_f}^\vee$, realizing a phase with the same charge category. Then we specialize to the case of an SPT phase and propose that to characterize an SPT phase with non-invertible symmetry, we need to consider its charge category in addition to the fiber functor of the symmetry category.

\begin{theorem} \label{thm:autoequivalence}
Let $\cC$ be a fusion category and $\cM$ be an indecomposable finite semisimple left $\cC$-module category. Let the dual category be ${\cC_\cM^\vee} :=\Fun_\cC(\cM,\cM)$. Suppose $\pi:\cC\to\cC$ is a monoidal equivalence, then $\pi$ induces an algebra $A_\pi$ in $\cC_\cM^\vee$ such that $_{A_\pi}(\cC_\cM^\vee)_{A_\pi} \cong \cC_\cM^\vee$;
\end{theorem}
\begin{proof}
$\cM$ induces a 2-equivalence between the 2-category of left $\cC$-module categories and the 2-category of right $\cC_\cM^\vee$-module categories, known as the categorical Morita equivalence. More explicitly, denote the 2-equivalence functor by $F_\cM$, then on objects, given a $\cC$-module $\cN$, 
$F_\cM(\cN)=\Fun_\cC(\cM,\cN)$ is a right $\cC_\cM^\vee$-module where the module action is simply the composition of module functors. On 1-morphisms and 2-morphisms (i.e. $\cC$-module functors and natural transformations), for any pair of $\cC$-modules $\cP,\cQ$, we have equivalence functors $\Fun_\cC(\cP,\cQ)\cong \Fun_{(\cC_\cM^\vee)^\rev}(F_\cM(\cP),F_\cM(\cQ))$, where $(\cC_\cM^\vee)^\rev$ is the same category as $\cC_\cM^\vee$ but with opposite tensor product. This can be expanded to
\begin{multline}\label{eq:functor identity}
    \Fun_\cC(\cP,\cQ)\\\cong \Fun_{\Fun_\cC(\cM,\cM)^\rev}(\Fun_\cC(\cM,\cP),\Fun_\cC(\cM,\cQ))
\end{multline}
and the equivalence is again induced by the post-composition of module functors. 

Now consider the monoidal equivalence $\pi:\cC\to\cC$, we define a new left $\cC$-module $_\pi\cM$ the module action of which is twisted by $\pi$: $c\rt[_\pi\cM] m:=\pi(c) \rt[\cM] m$. Automatically, we have $\Fun_\cC(\cM,\cM)\cong \Fun_\cC({}_\pi\cM,{}_\pi\cM)$, where module functors $(F,s_{-,-})$ are mapped to $(F,s_{\pi(-),-})$ and module natural transformations are mapped to themselves. Under the categorical Morita equivalence $F_\cM$, $F_\cM(_\pi\cM)=\Fun_\cC(\cM, {}_\pi\cM)$ and consequently
\begin{align*}
    \cC_\cM^\vee&=\Fun_\cC(\cM,\cM)
    \\&
    \cong \Fun_\cC({}_\pi\cM,{}_\pi\cM)
    \\&
    \cong \Fun_{(\cC_\cM^\vee)^\rev}(F_\cM(_\pi\cM),F_\cM(_\pi\cM)).
\end{align*} 
Therefore, $F_\cM(_\pi\cM)$ is an invertible $\cC_\cM^\vee$-$\cC_\cM^\vee$-bimodule. 

For any module functor $Y\in \Fun_\cC(\cM,{}_\pi\cM)$, making use of the following internal hom adjunction
\[ \Fun_\cC(\cM,\cM)(X,[Y,Y])\cong \Fun_\cC(\cM,{}_\pi\cM)(YX,Y),\]
we can construct an algebra $[Y,Y]$ in $\cC_\cM^\vee$ such that $_{[Y,Y]}(\cC_\cM^\vee)\cong F_\cM(_\pi\cM)$. Choose any algebra $A_\pi$ Morita equivalent to $[Y,Y]$ in $\cC_\cM^\vee$, we then have
$_{A_\pi}(\cC_\cM^\vee)\cong F_\cM(_\pi\cM)$ and $_{A_\pi}(\cC_\cM^\vee)_{A_\pi}\cong \Fun_{(\cC_\cM^\vee)^\rev}(F_\cM(_\pi\cM),F_\cM(_\pi\cM))\cong \cC_\cM^\vee$. 
\end{proof}

\begin{figure}
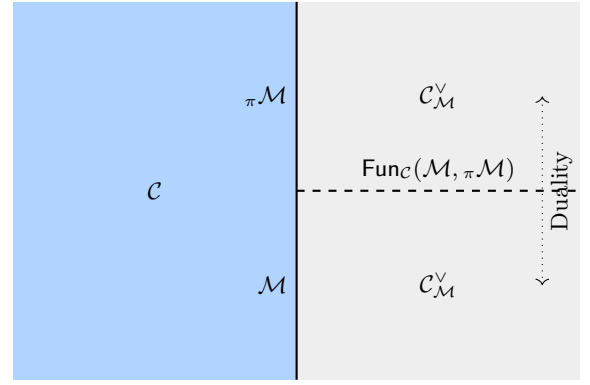

    \centering
    \diagram{2.5}{
        \fill[YaleLightBlue] (0,0) rectangle node[black]{$\cC$} (1.5,2);
        \fill[YaleLightGrey] (1.5,0) rectangle node[black]{$\cC_\cM^\vee$} (3,1);
        \fill[YaleLightGrey] (1.5,1) rectangle node[black]{$\cC_{\cM}^\vee$} (3,2);
        \draw[thick] (1.5,0) -- node[left]{$\cM$} (1.5,1);
        \draw[thick] (1.5,1) -- node[left]{${}_{\pi}\cM$} (1.5,2);
        \draw[dashed,thick] (1.5,1) -- node[above]{$\Fun_{\cC}(\cM,{}_{\pi}\cM)$} (3,1);
        \draw[dotted, <->] (2.8,.5) -- node[below,rotate = 90]{Duality} (2.8,1.5);
    }
    \caption{Illustration of the categorical relationship of SPT phases. $\cC$ represents the symmetry category, $\cC_\cM^\vee$ represents the charge category, and $\Fun_\cC(\cM,{}_{\pi}\cM)$ represents the category of projective charges. Left $\cC$ module $\cM$ describes the lattice construction. For tensor product Hilbert space, $\cM = \Hilb$.}
    \label{fig:2-moreq}
\end{figure}

Next, we prove a parallel result when considering $\Hilb$ as a $\cC$-module.
Firstly, we have monoidal equivalence
\begin{align*}
    \Hilb &\cong \Fun(\Hilb,\Hilb)\\
    W &\mapsto W\otimes -\\
    \beta &\mapsto \beta\otimes \id_-,
\end{align*}
which also means that any natural transformation $W\otimes -\Rightarrow W'\otimes -$ must be of the form $\beta\otimes\id_-$ for some linear map $\beta:W\to W'$.

A $\cC$-module structure on $\Hilb$ is the same as a monoidal functor $f:\cC\to\Hilb$. To emphasize the role of the monoidal functor, we denote its corresponding $\cC$-module by $\Hilb_f$. Then note that there is a natural functor that forgets the module functor structure.
\begin{align*}
    \mathrm{fgt}:\cC_{\Hilb_f}^\vee&=\Fun_\cC(\Hilb_f,\Hilb_f)\\&\to \Fun(\Hilb,\Hilb)\cong \Hilb.
\end{align*}

\begin{theorem} \label{lem:nontrivial fb}
    Consider two monoidal functors $f,h:\cC\to\Hilb$. There exist an algebra $A$ in $\cC^\vee_{\Hilb_f}$ such that
    \begin{enumerate}
       \item $_{A}(\cC_{\Hilb_f}^\vee)_{A}\cong \cC_{\Hilb_{h}}^\vee$;
        \item $\mathrm{fgt}(A)$ is a matrix algebra, or equivalently $\dim Z(\mathrm{fgt}(A))=1$.
    \end{enumerate}
\end{theorem}

\begin{proof}
    Let the linear dual of $W$ be $W^*$, and $W^*\ot -$ is automatically a left and right adjoint to $W\otimes-$ as linear functors $\Hilb\to\Hilb$. We want to show that such adjointness lifts to $\cC$-module functors. Suppose 
    $W\ot -$ equipped with natural isomorphisms $\beta_c\ot\id_-:W\otimes f(c)\otimes -\cong h (c)\otimes W \otimes -$ is a $\cC$-module functor in $\Fun_\cC(\Hilb_f,\Hilb_h)$. Then it is straightforward to check that $W^*\ot -$ equipped with natural isomorphisms $(\beta_{c^*})^*\ot\id_-: W^*\otimes h(c) \otimes -\cong f(c)\otimes W^*\otimes -$ is a $\cC$-module functor in $\Fun_\cC(\Hilb_h,\Hilb_f)$ and is moreover  left and right adjoint to $W\otimes -$. This fact further implies the adjunction
    \begin{multline*}
        \Fun_\cC(\Hilb_f,\Hilb_f)(X,W^*\otimes W\otimes -)\cong \\\Fun_\cC(\Hilb_f,\Hilb_h)(W\otimes X(-),W\otimes -).
    \end{multline*}
    We then know that 
    \begin{equation*}
        \Fun_\cC(\Hilb_f,\Hilb_h)\cong {}_{W^*\ot W\ot -} (\cC_{\Hilb_f}^\vee).
    \end{equation*}
    
    Clearly $\mathrm{fgt} (W^*\otimes W\otimes -)=W^*\otimes W\cong \End(W)$ is a matrix algebra. Since Morita equivalence is preserved by the linear monoidal functor $\mathrm{fgt}$, we conclude that for any algebra $A$ Morita equivalent to $W^*\ot W\ot -$ in $\cC_{\Hilb_f}^\vee$, $\mathrm{fgt}{A}$ is Morita equivalent to a matrix algebra in $\Hilb$, thus must also be a matrix algebra.
\end{proof}

\begin{corollary}
    Given a fusion category symmetry $\cC$ equipped with a fiber functor $f:\cC\to \Hilb$, and a monoidal equivalence $\pi:\cC\to\cC$. We can take an algebra $A_\pi$ in $\cC_{\Hilb_f}^\vee$ via $_{A_\pi}(\cC_{\Hilb_f}^\vee)\cong\Fun_\cC(\Hilb_f,\Hilb_{f \pi})$. Then we have the following
    \begin{itemize}
        \item[--] $_{A_\pi}(\cC_{\Hilb_f}^\vee)_{A_\pi}\cong \cC_{\Hilb_{f\pi}}^\vee\cong \cC_{\Hilb_f}^\vee$;
        \item[--] $\mathrm{fgt}(A_\pi)$ is a matrix algebra, or equivalently $\dim Z(\mathrm{fgt}(A_\pi))=1$.
    \end{itemize}
\end{corollary}

\begin{figure}
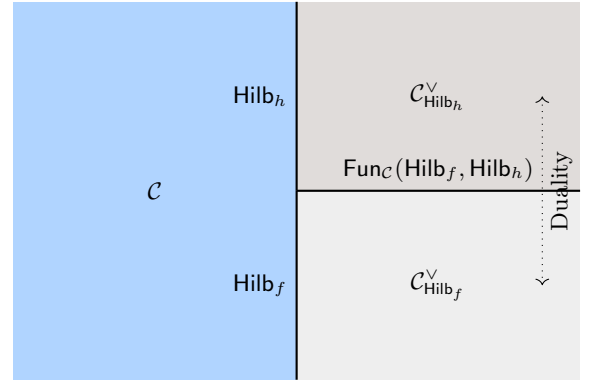

    \centering
    \diagram{2.5}{
        \fill[YaleLightBlue] (0,0) rectangle node[black]{$\cC$} (1.5,2);
        \fill[YaleLightGrey] (1.5,0) rectangle node[black]{$\cC_{\Hilb_f}^\vee$} (3,1);
        \fill[YaleMidGrey] (1.5,1) rectangle node[black]{$\cC_{\Hilb_h}^\vee$} (3,2);
        \draw[thick] (1.5,0) -- node[left]{$\Hilb_f$} (1.5,1);
        \draw[thick] (1.5,1) -- node[left]{$\Hilb_h$} (1.5,2);
        \draw[thick] (1.5,1) -- node[above]{$\Fun_{\cC}(\Hilb_{f},\Hilb_h)$} (3,1);
        \draw[dotted, <->] (2.8,.5) -- node[below,rotate = 90]{Duality} (2.8,1.5);
    }
    \caption{Illustration of the categorical relationship of SPT phases. Here the pair $(\cC,f)$ defines a symmetry as well as the trivial phase, and $\cC_{\Hilb_f}^\vee$ the charge category corresponding to this symmetry. The pair $(\cC,h)$ defines another symmetry, whose charge category is $\cC_{\Hilb_h}^\vee$, which may or may not be the same as $\cC_{\Hilb_f}^\vee$.}
    \label{fig:isocat}
\end{figure}

We show the converse by starting from the charge category $\cD:=\cC_\cM^\vee$ and the forgetful functor:
\begin{theorem}\label{thm:alg2fb}
    Fix a fiber functor $f:\cD\to\Hilb$. Any algebra $A\in \cD$ such that $fA$ is a matrix algebra in $\Hilb$, corresponds to a fiber functor $\cD_{\Hilb_{f}}^\vee\to \Hilb$.
\end{theorem}
\begin{proof}
    Let
    \[
        \underline{f}:{}_{A}\cD_{A}\rightarrow{}_{fA}\Hilb_{fA}
    \]
    be the monoidal functor induced from $f$.
    By the Lemma 4.2, Theorem 4.9, and Remark 4.10 of \cite{lan2024categorysetorders},
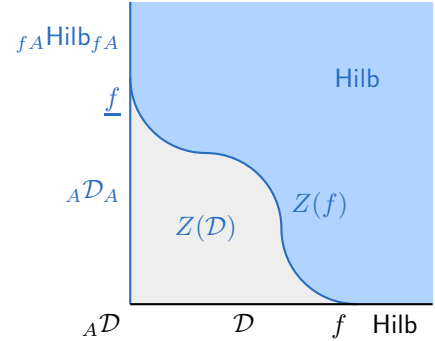
    \[
        \begin{tikzpicture}[baseline=(current bounding box.center),scale=1]
              \fill[YaleLightBlue] (0,0) rectangle (4,4);
              \fill[YaleLightGrey] (0,0)--(3,0) arc [start angle=-90, end angle=-180, radius=1] arc [start angle=0, end angle=90, radius=1]
              arc [start angle=-90, end angle=-180, radius=1] -- cycle;
              \draw[thick,YaleMidBlue] (3,0) arc [start angle=-90, end angle=-180, radius=1] node[above right]{$Z(f)$} arc  [start angle=0, end angle=90, radius=1] 
              arc [start angle=-90, end angle=-180, radius=1];
              \draw[thick,YaleMidBlue] (0,0)--node[left]{${}_{A}\cD_{A}$} (0,3) node[below left]{$\underline{f}$} --node[left]{${}_{fA}\Hilb_{fA}$} (0,4);
              \draw[thick] (4,0)--node[below]{$\Hilb$} (3,0)node[below left]{$f$} --node[below]{$\cD$} (0,0)node[below left]{${}_{A}\cD$};
              \node[YaleMidBlue] at (3,3) {$\Hilb$};
              \node[YaleMidBlue] at (1,1) {$Z(\cD)$};
        \end{tikzpicture}
    \]
    we have the canonical equivalence
    \[
        Z(\underline{f})\cong Z(f).
    \]
    Since $fA\in\Hilb$ is a simple algebra, there is a monoidal equivalence ${}_{fA}\Hilb_{fA}\cong \Hilb$.
    Thus we have another fiber functor
    \[
        Z(f)\cong Z(\underline{f})\xrightarrow{\fgt}{}_{fA}\Hilb_{fA}\cong \Hilb.
    \]
\end{proof}
An on-site symmetry is necessarily described by the pair $(\cC,f:\cC\to \Hilb)$, which determines the microscopic realization of the fusion category symmetry along with a trivial phase.
More specifically, $\cC$ is some abstract fusion category and $f$ describes its local structure. $\cC_{\Hilb_f}^\vee$ is the corresponding charge category. For an object $(V\otimes -,\, \beta)\in\cC_{\Hilb_f}^\vee$,  $V$ is the underlying Hilbert space of the symmetry charge, and $\beta_{c}: V \ot f(c) \cong f(c) \ot V$ is the local tensor of the $c$-MPO acting on the space $V$. In \cite{Lan_2024}, the defects between two such pairs $(\cC,f)$ and $(\cC,h)$ form the category $\Hom_{\Fun(\Sigma\cC,2\Hilb)}(\Sigma f,\Sigma h)=\Fun_\cC(\Hilb_f,\Hilb_h)$. Fixing $(\cC,f)$ as the trivial phase, $\Fun_\cC(\Hilb_f,\Hilb_h)$ is the category of edge-modes of $(\cC,h)$. For a module functor $(W\ot- , \gamma)\in \Fun_\cC(\Hilb_f,\Hilb_h)$, $W$ is the edge mode and $\gamma_c: W \ot f(c) \rightarrow h(c) \ot W$ is the local tensor of the $c$-MPO acting on the projective charge $W$. The internal hom tensors two opposite projective charges hence gives an algebra $A = (W^* \ot W\otimes -, \lambda)$ in the charge category $\cC_{\Hilb_f}^\vee$. And $\dim Z(\mathrm{fgt}(A)) = 1$ describes a symmetric state with a unique ground state in the thermodynamic limit.
In particular, when $\Hilb_f$ and $\Hilb_h$ are equivalent as $\cC$-modules, i.e., $W$ is taken to be $\C$, then there is no edge mode, suggesting that $f$ and $h$ realize the same SPT phase, the trivial phase. Its corresponding algebra $\End (W)$ is the trivial algebra in the charge category $\cC^\vee_{\Hilb_f}$. Thus we define

\begin{definition}
    Two pairs $(\cC,f)$ and $(\cC,h)$, or two fiber functors $f,h$ from $\cC$, are called phase-equivalent if the corresponding $\cC$-modules $\Hilb_f$ and $\Hilb_h$ are equivalent.
\end{definition}
Then there exists a one-to-one correspondence between 
\begin{itemize}
    \item the Morita classes of algebras in $\cC^\vee_{\Hilb_f}$ whose image under $\fgt$ is Morita trivial (i.e. matrix algebras) in $\Hilb$,
    \item and the phase-equivalent classes of fiber functors from $\cC$.
\end{itemize}

However, there are further subtleties on the macroscopic observables associated with the pairs $(\cC,f)$ and $(\cC,h)$. Two factors need to be taken into consideration: 
\begin{itemize}
    \item whether $f$ and $h$ are related by a monoidal equivalence $\pi:\cC\to \cC$, i.e. $h=f\pi$;
    \item whether the charge category is preserved $\cC^\vee_{\Hilb_f}\cong \cC^\vee_{\Hilb_h}$.
\end{itemize}
There are three possible scenarios:
\begin{enumerate}
    \item $h=f\pi$ induces a monoidal equivalence $\cC^\vee_{\Hilb_f}\cong \cC^\vee_{\Hilb_h}$ and $(\cC,f)$ and $(\cC,f\pi)$ are not necessarily phase-equivalent. All the $G$-SPT phases and three $\Rep^\dagger(D_8)$ SPTs belong to this case;
    \item $\cC^\vee_{\Hilb_f}\not\cong \cC^\vee_{\Hilb_h}$, where $h$ and $f$ are not related by monoidal equivalences. By Theorem \ref{lem:nontrivial fb}, there exists an algebra $A$ in $\cC^\vee_{\Hilb_f}$ such that $_A(\cC^\vee_{\Hilb_f})_A\cong \cC^\vee_{\Hilb_h}$. We can construct a Q-system model with such algebra $A$ and the fusion of defects is described by $\cC^\vee_{\Hilb_h}$. In such a model, the Hamiltonian is $(\cC,f)$ symmetric with a unique ground state. However, the fusion rules of the defects differ from the trivial phase, meaning that the charge category is not preserved. We view such a model as in a \textit{locally} spontaneous symmetry-breaking phase. Unlike invertible symmetries, where a symmetry is broken to a subgroup, here the symmetry $(\cC,f)$ is broken to $(\cC,h)$: the fusion category remains the same with the fiber functor changed. We emphasize that a fiber functor corresponds to the microscopic realization of the symmetry, consistent with the fact that in the locally SSB phase, only the fusion of defects is changed, which are local physical observables.

    
    \item There might be a situation landing in the middle of 1 and 2. It is possible that when $f$ and $h$ are not related by monoidal equivalences, it happens to be $\cC^\vee_{\Hilb_f}\cong \cC^\vee_{\Hilb_h}$. 
    It remains unclear what physical observables detect such an exotic case.
\end{enumerate}

\bigskip
In general, the charge category of a fusion category symmetry is not unique, as discussed in scenario 2. $\cC=\Rep^\dagger(S_3 \times \Z_3)$ symmetry in Wen's talk \cite{wen2024spt} is one such example; the charge category $\cC_{\Hilb_f}^\vee$ given by the forgetful functor $f$ is $\Hilb_{S_3 \times \Z_3}$. Following Theorems \ref{lem:nontrivial fb} and \ref{thm:alg2fb}, take $A_+ =  \mathbb{C}^{\psi_+}[\mathbb{Z}_3 \times \mathbb{Z}_3]$, we obtain a different charge category $_{A_+}(\Hilb_{S_3 \times \Z_3})_{A_+}\cong  \mathrm{TY}_{\Z_3 \times \Z_3}^{\chi, \epsilon}$, where $\chi(g,h) = \exp{\left[\frac{2\pi\ii}{3}\left(g_1h_2+g_2h_1\right)\right]},\;\epsilon = 1$ \footnote{Part of the result originated from a discussion with Conghuan Luo.}. See Appendix \ref{sect:bimodS3xZ3} for a detailed calculation. Another known example of scenario 2 is the isocategorical groups \cite{etingof2000isocategoricalgroups}. It states that when $\cC$ is a group representation category, its corresponding charge categories might have different underlying group structures, meaning that $\cC_{\Hilb_f}^\vee\cong \Hilb_G$, $\cC_{\Hilb_h}^\vee\cong \Hilb_{G'}$ for non-isomorphic groups $G,G'$. 

In this paper, we focus on the case of $\cC=\Rep^\dagger(D_8)$, where three fiber functors are related by monoidal equivalence. The module action of $\cC$ on $\Hilb$ is fixed by the forgetful functor $\mathrm{fgt}: \Rep^\dagger(D_8)\to \Hilb$. The charge category is $\Hilb_{D_8}$ and is preserved. Indeed we can show that $_{A_\pi}(\Hilb_{D_8})_{A_\pi} \cong \Hilb_{D_8}$. In Section \ref{sec:RepD8}, we give an explicit lattice realization of the monoidal equivalence $\pi$ via $_{A_\pi}(\Hilb_{D_8})_{A_\pi} \cong \Hilb_{D_8}$ and show that three SPT phases of $\Rep^\dagger(D_8)$ symmetry are permuted under such a construction.

By the classification result of $\Rep(G)$-module categories, the module category concerned in our $\Rep^\dagger(D_8)$ example is also equivalent to the category of projective representations $\Rep^{\omega\dagger}((\Z_2 \times \Z_2)^{(1),(2)}) \cong \Hilb$. Here, $\omega = -1 \in H^2(\Z_2 \times \Z_2, U(1))$ denotes the non-trivial cocycle.
The left $\Rep^\dagger(D_8)$-module structures on $\Rep^{\omega\dagger}((\Z_2 \times \Z_2)^{(1),(2)})$ are defined by $- \otimes W^{(1),(2)}$, where $W^{(1),(2)}$ is the irreducible projective representation of $(\Z_2 \times \Z_2)^{(1),(2)}$, respectively.

\section{Conclusion and discussion}
In this work, we develop a framework to classify and construct (1+1)d SPT phases with fusion category symmetries. We focus on the kinematic role of the fiber functor, meaning how it defines an on-site symmetry operator, the charge category, and the trivial phase. We emphasize that a trivial phase is an integral component of the microscopic realization of an anomaly-free fusion category symmetry. We provide an alternative characterization of an SPT phase using the Q-system in the charge category. As an example, we study the $\Rep^\dagger(D_8)$-SPTs and realize the $S_3$-automorphism between three fiber functors via a duality transformation. Our construction applies to spontaneous symmetry breaking phases of a fusion category symmetry as well.

Our construction provided a generic framework for understanding non-invertible symmetry in (1+1)d and in higher dimensions. In particular, there is still much to explore in topological phases enriched by fusion category symmetries. One natural generalization is anomaly-free fermionic fusion category symmetry (for example \cite{bhardwaj2024fermionicnoninvertiblesymmetries11d, wen2024topologicalholographyfermions, inamura2023fermionizationfusioncategorysymmetries}). Symmetry charges of a fermionic fusion category symmetry form a fusion category $\cD$ together with a monoidal functor to the category of super vector spaces \cite{Kong_2024,kelly1982enriched}:
\[
    \chargef:\cD\rightarrow {\sVec}.
\]

To characterize the fermonic SPTs in terms of Q-systems, a subtlety lies in the fact that fermion parity is a preserved symmetry thus the ground state degeneracy is observed by operators with even fermion parity \cite{Chen_2011_spinchain}. Therefore, a fermionic SPT is described by a Q-system $A$ in $\cD$, such that $\chargef(A)\in\sVec$ is simple.

We could also consider higher-dimensional generalizations \cite{Inamura_2024,Delcamp_2024}. In \cite{lan2024categorysetorders}, the charge category of a 2D system is a fusion 2-category \cite{douglas2018fusion2categoriesstatesuminvariant} $\mathfrak{D}$. However, it remains unclear how to categorically describe the underlying 2D bosonic spin system and to understand local symmetry charges of $\mathfrak{D}$. A 2D bosonic spin system should be described by $2\Ve$, and the local symmetry charge realization is given by the 2-fiber functor $\mathfrak{D}\rightarrow 2\Ve$. However, in tensor product Hilbert space lattice models, we only have $\Ve$, the looping of the $2\Ve$. Thus how to incorporate the higher structure in $2\Ve$ in a lattice model remains unclear.

On the other hand, if we relax the anomaly-free fusion category to those that do not admit a fiber functor, the anomalous symmetry category plays an essential role in understanding gapless phases. For example, anyonic chain construction can carry anomalous fusion category symmetries and can be used to explore critical phases \cite{Feiguin_2007, ning2023building1dlatticemodels,Kong_2020_gapless,Kong_2021_gapless}.

\section*{Acknowledgement}
CQM would like to express heartfelt gratitude to Liang Kong for many insightful discussions and for providing invaluable mental support throughout this work. Special thanks to David Penneys for introducing the Q-system and for his elegant and comprehensive lecture notes, which have been an indispensable resource. CQM also thanks Ansi Bai for patiently answering numerous questions about Hopf algebras, Gen Yue for instructive discussions on the category of SET orders, and Yizhou Ma for discussions on operator algebra. X.Y. thanks Meng Cheng, Da-Chuan Lu for helpful discussions, and is especially grateful to Conghuan Luo for discussions on $\Rep^\dagger(S_3 \times \Z_3)$ SPTs. This work is supported by funding from Hong Kong’s Research Grants Council (RGC Research Fellow Scheme (RFS) 2023/24,  No. RFS2324-4S02). X.Y. acknowledges the support from NSF under the grant number DMR-2424315. TL is supported
by start-up funding from The Chinese University of Hong Kong, and by funding from Research Grants Council, University Grants Committee of Hong Kong (ECS No. 24304722).

\textit{Note added}: as we were finalizing this work for posting, we became aware of a separate work that shares similar results as ours \cite{warman2024categoricalsymmetriesspinmodels}.

\onecolumngrid

\appendix

\section{Comparison with the cluster state SPT}\label{sec:reduction}
In \cite{seifnashri2024clusterstatenoninvertiblesymmetry,Tantivasadakarn_2024, Seiberg_2024,Seiberg_2024_Majorana}, the Kramers-Wannier operator $D_{KW}$ written in MPO form is
\begin{gather*}
    (D_{KW})_{00} = \frac{1}{\sqrt{2}}\begin{pmatrix}
        1 & 0\\ 1 & 0
    \end{pmatrix},\quad 
    (D_{KW})_{10} = \frac{1}{\sqrt{2}}\begin{pmatrix}
        1 & 0\\ -1 & 0
    \end{pmatrix},\\
    (D_{KW})_{01} = \frac{1}{\sqrt{2}}\begin{pmatrix}
        0 & 1\\ 0 & 1
    \end{pmatrix},\quad
    (D_{KW})_{11} = \frac{1}{\sqrt{2}}\begin{pmatrix}
        0 & -1\\ 0 & 1
    \end{pmatrix}.
\end{gather*}

Following \cite{seifnashri2024clusterstatenoninvertiblesymmetry}, the non-invertible symmetry operator $D$ of $\Rep^\dagger(D_8)$ is composed by
\[
D = T_LD^e_{KW}D^o_{KW}.
\]
We group two nearby sites to form a unit cell. A direction composition of three MPO's gives an MPO with an eight-dimensional virtual bond, which is not injective and can be further reduced \cite{Cirac_2021,perezgarcia2007matrixproductstaterepresentations}.

After reduction, the MPO is written as
\begin{gather*}
D_{(00)(00)} = 
\frac{1}{2}\left(
\begin{array}{cccc}
 1 & 0 & 0 & 0 \\
 1 & 0 & 0 & 0 \\
 1 & 0 & 0 & 0 \\
 1 & 0 & 0 & 0 \\
\end{array}
\right),\quad 
D_{(00)(01)} =
\frac{1}{2}\left(
\begin{array}{cccc}
 0 & 1 & 0 & 0 \\
 0 & 1 & 0 & 0 \\
 0 & 1 & 0 & 0 \\
 0 & 1 & 0 & 0 \\
\end{array}
\right),\\
D_{(00)(10)} =
\frac{1}{2}\left(
\begin{array}{cccc}
 1 & 0 & 0 & 0 \\
 1 & 0 & 0 & 0 \\
 -1 & 0 & 0 & 0 \\
 -1 & 0 & 0 & 0 \\
\end{array}
\right),\quad
D_{(00)(11)} =
\frac{1}{2}\left(
\begin{array}{cccc}
 0 & 1 & 0 & 0 \\
 0 & 1 & 0 & 0 \\
 0 & -1 & 0 & 0 \\
 0 & -1 & 0 & 0 \\
\end{array}
\right),\\
D_{(01)(00)} = 
\frac{1}{2}\left(
\begin{array}{cccc}
 0 & 0 & 1 & 0 \\
 0 & 0 & 1 & 0 \\
 0 & 0 & 1 & 0 \\
 0 & 0 & 1 & 0 \\
\end{array}
\right),\quad
D_{(01)(01)} =
\frac{1}{2}\left(
\begin{array}{cccc}
 0 & 0 & 0 & 1 \\
 0 & 0 & 0 & 1 \\
 0 & 0 & 0 & 1 \\
 0 & 0 & 0 & 1 \\
\end{array}
\right),\\
D_{(01)(10)} =
\frac{1}{2}\left(
\begin{array}{cccc}
 0 & 0 & -1 & 0 \\
 0 & 0 & -1 & 0 \\
 0 & 0 & 1 & 0 \\
 0 & 0 & 1 & 0 \\
\end{array}
\right),\quad 
D_{(01)(11)} =
\frac{1}{2}\left(
\begin{array}{cccc}
 0 & 0 & 0 & -1 \\
 0 & 0 & 0 & -1 \\
 0 & 0 & 0 & 1 \\
 0 & 0 & 0 & 1 \\
\end{array}
\right),\\
D_{(10)(00)} =
\frac{1}{2}\left(
\begin{array}{cccc}
 1 & 0 & 0 & 0 \\
 -1 & 0 & 0 & 0 \\
 1 & 0 & 0 & 0 \\
 -1 & 0 & 0 & 0 \\
\end{array}
\right),\quad
D_{(10)(01)} =
\frac{1}{2}\left(
\begin{array}{cccc}
 0 & -1 & 0 & 0 \\
 0 & 1 & 0 & 0 \\
 0 & -1 & 0 & 0 \\
 0 & 1 & 0 & 0 \\
\end{array}
\right),\\
D_{(10)(10)} =
\frac{1}{2}\left(
\begin{array}{cccc}
 1 & 0 & 0 & 0 \\
 -1 & 0 & 0 & 0 \\
 -1 & 0 & 0 & 0 \\
 1 & 0 & 0 & 0 \\
\end{array}
\right),\quad
D_{(10)(11)} =
\frac{1}{2}\left(
\begin{array}{cccc}
 0 & -1 & 0 & 0 \\
 0 & 1 & 0 & 0 \\
 0 & 1 & 0 & 0 \\
 0 & -1 & 0 & 0 \\
\end{array}
\right),\\
D_{(11)(00)} =
\frac{1}{2}\left(
\begin{array}{cccc}
 0 & 0 & 1 & 0 \\
 0 & 0 & -1 & 0 \\
 0 & 0 & 1 & 0 \\
 0 & 0 & -1 & 0 \\
\end{array}
\right),\quad
D_{(11)(01)} =
\frac{1}{2}\left(
\begin{array}{cccc}
 0 & 0 & 0 & -1 \\
 0 & 0 & 0 & 1 \\
 0 & 0 & 0 & -1 \\
 0 & 0 & 0 & 1 \\
\end{array}
\right),\\
D_{(11)(10)} =
\frac{1}{2}\left(
\begin{array}{cccc}
 0 & 0 & -1 & 0 \\
 0 & 0 & 1 & 0 \\
 0 & 0 & 1 & 0 \\
 0 & 0 & -1 & 0 \\
\end{array}
\right)
,\quad
D_{(11)(11)} =
\frac{1}{2}\left(
\begin{array}{cccc}
 0 & 0 & 0 & 1 \\
 0 & 0 & 0 & -1 \\
 0 & 0 & 0 & -1 \\
 0 & 0 & 0 & 1 \\
\end{array}
\right).
\end{gather*}
Since the virtual bond dimension of $D$ is four instead of two, not equal to the quantum dimention of $\sigma\in\Rep^\dagger(D_8)$, this MPO $D$ is not in the onsite form based on our definition \eqref{eq:onsite}. We can further reduce this MPO by conjugating $D$ with the CZ gate. Simultaneously, we conjugate the cluster Hamiltonian with the CZ gate. After conjugation, $D$ MPO can be reduced to its minimal form while the cluster Hamiltonian becomes $-\sum_i X_i$. Explicitly, $D$ is written as

\begin{gather*}
    D_{(00)(00)} = \frac{1}{2}\begin{pmatrix}
    1 & 0\\ 1 & 0
\end{pmatrix},\quad 
D_{(00)(01)} = \frac{1}{2}\begin{pmatrix}
    0 & 1\\ 0 & 1
\end{pmatrix}, \quad
D_{(00)(10)} = \frac{1}{2}\begin{pmatrix}
    1 & 0\\ -1 & 0
\end{pmatrix},\quad
D_{(00)(11)} = \frac{1}{2}\begin{pmatrix}
    0 & -1\\ 0 & 1
\end{pmatrix},\\
D_{(01)(00)} = \frac{1}{2}\begin{pmatrix}
    0 & 1\\ 0 & 1
\end{pmatrix},\quad
D_{(01)(01)} = \frac{1}{2}\begin{pmatrix}
    1 & 0\\ 1 & 0
\end{pmatrix},\quad
D_{(01)(10)} = \frac{1}{2}\begin{pmatrix}
    0 & -1\\ 0 & 1
\end{pmatrix},\quad
D_{(01)(11)} = \frac{1}{2}\begin{pmatrix}
    1 & 0\\ -1 & 0
\end{pmatrix},
\\
D_{(10)(00)} = \frac{1}{2}\begin{pmatrix}
    1 & 0\\ -1 & 0
\end{pmatrix},\quad 
D_{(10)(01)} = \frac{1}{2}\begin{pmatrix}
    0 & -1\\ 0 & 1
\end{pmatrix},\quad
D_{(10)(10)} = \frac{1}{2}\begin{pmatrix}
    1 & 0\\ 1 & 0
\end{pmatrix},\quad
D_{(10)(11)} = \frac{1}{2}\begin{pmatrix}
    0 & 1\\ 0 & 1
\end{pmatrix},\\
D_{(11)(00)} = \frac{1}{2}\begin{pmatrix}
    0 & -1\\ 0 & 1
\end{pmatrix},\quad
D_{(11)(01)} = \frac{1}{2}\begin{pmatrix}
    1 & 0\\ -1 & 0
\end{pmatrix},\quad
D_{(11)(10)} = \frac{1}{2}\begin{pmatrix}
    0 & 1\\ 0 & 1
\end{pmatrix},\quad
D_{(11)(11)} = \frac{1}{2}\begin{pmatrix}
    1 & 0\\ 1 & 0
\end{pmatrix}.
\end{gather*}
Under the basis $\{|+\ra,|-\ra\}$, non-zero components of $D$ are
\begin{gather*}
    D_{(++)(++)} = \begin{pmatrix}
    1 & 0\\ 0 & 1
\end{pmatrix},\quad
D_{(+-)(+-)} = \begin{pmatrix}
    1 & 0\\ 0 & -1
\end{pmatrix} = Z,\\
D_{(-+)(-+)} = \begin{pmatrix}
    0 & 1\\ 1 & 0
\end{pmatrix} = X,\quad
D_{(--)(--)} = \begin{pmatrix}
    0 & -1\\ 1 & 0
\end{pmatrix} = -\ii Y.
\end{gather*}
Meanwhile, we can check that the invertible symmetry operators of $\Rep^\dagger(D_8)$
\[
    \eta^e = \prod_{j:\,even} X_j,\quad \eta^o = \prod_{j:\,odd} X_j
\]
remain onsite on the closed chain. The Hamiltonian is transformed to
\[
    H = -\sum_{i}X_i,
\]
unique ground state of which is a product state. Upon fixing the symmetry operators onsite, the cluster state can be chosen to be the trivial phase. Without the stacking structure of non-invertible SPTs, this choice of trivial phase is not canonical. However, we can still define a trivial phase for non-invertible SPTs once we assign a proper onsite symmetry operator.
Denoting
\[
    S = \frac{1}{\sqrt{2}}\begin{pmatrix}
    1 & -\ii\\ -\ii & 1
\end{pmatrix},
\]
we perform a similar transformation to $\sigma$, such that
\[
    \rho'_\sigma(r) =  S\rho_\sigma(r) S^{-1}= S\ii Z S^{-1} = -\ii Y,\quad \rho'_\sigma(s) = X.
\]
We identify
\[
    |++\ra = |e\ra,\quad |-+\ra = |s\ra,\quad |--\ra = |r\ra,\quad |+-\ra = |sr\ra.
\]
Thus the local Hilbert space in \cite{seifnashri2024clusterstatenoninvertiblesymmetry} is
\[
    V = \mathrm{span}( |e\ra,\ |s\ra,\ |r\ra,\ |sr\ra ).
\]
The Hamiltonian for the trivial phase is
\[
 H = -\sum_{i}P^1_i.
\]
where $P_i$ is the local projection to $\C_e$.

In the second model (see Eqn.16 of \cite{seifnashri2024clusterstatenoninvertiblesymmetry}), after applying the CZ gate, the ground state is stabilized by the generator
\[
    -X_{2i} = 1,\quad -Z_{2n-1}X_{2n+1}Z_{2n+3} = 1
\]
Imposing the constraint $X_{2n} = -1$, the local Hilbert space on a two-site unit cell is $\C\{r, sr\}$.
In $|r\ra, |sr\ra$ basis, the second operator can be written as
\begin{gather*}
    Z_1X_2Z_3|g,h,k\ra = (-1)^{n_s(h)+1}|sg,h,sk\ra.
\end{gather*}
 We further group 2 cells into a single one, such that the local Hilbert space is
\[
    \C\{r, sr\}\ot\C\{r, sr\} = \C\{e, s, r^2, sr^2\}
\]

\[
    e = sr\,sr,\quad s = r\,sr,\quad r^2 = r\,r,\quad sr^2 = sr\,r.
\]
Thus the final unit cell is composed with 4 original sites, labeled by the capital letters $I, J, K,\cdots$.
The Hamiltonian is equivalent to
\[
    H = -\sum_{I}\frac{1}{4} (1 - Z_{4I-3}X_{4I-1}Z_{4I+1})(1 - Z_{4I-1}X_{4I+1}Z_{4I+3}) = -\sum_{I}\tilde{P}_{I, I+1}\,.
\]
with the local projector 
\[
    \tilde{P}_{I, I+1} =  \frac{1}{4}(1 - Z_{4I-3}X_{4I-1}Z_{4I+1})(1 - Z_{4I-1}X_{4I+1}Z_{4I+3}).
\]
This model in \cite{seifnashri2024clusterstatenoninvertiblesymmetry} is equivalent to our construction when the sign is flipped.
Notice that 
\[
    U_{i} = Z_{8i+1}Z_{8i+3}Z_{8i+5}Z_{8i+7}
\]
is a symmetric operator, we have
\[
    \left(-\sum_{I}\tilde{P}_{I, I+1}\right)\prod_{i = 1}^{[L/8]}U_{i} =
    \begin{cases}
        \prod_{i = 1}^{[L/8]}U_{i}\left(-\sum_{I}P_{I, I+1}\right) & L \equiv 0\mod{8}\\
        \prod_{i = 1}^{[L/8]}U_{i}\left(-\sum_{I}P_{I, I+1} - \bar{P}_{L/4-1,L/4} - \bar{P}_{L/4,L/4+1}\right) & L\equiv 4\mod{8}
    \end{cases}
\]
where
\begin{gather*}
    P_{I,I+1} = \frac{1}{4}(1 + Z_{4I-3}X_{4I-1}Z_{4I+1})(1 + Z_{4I-1}X_{4I+1}Z_{4I+3}),\\
    \bar{P}_{L/4-1,L/4} = \frac{1}{4}(1 + Z_{L-7}X_{L-5}Z_{L-3})(1 - Z_{L-5}X_{L-3}Z_{L-1}),\\
    \bar{P}_{L/4,1} = \frac{1}{4}(1 - Z_{L-3}X_{L-1}Z_{1})(1 + Z_{L-1}X_{1}Z_{3}).
\end{gather*}
Thus the local projector is $P_{I, I+1}$ with exceptions on the boundary, which correspond to local symmetric defects. Since a phase is at the thermodynamic limit, the phase realized in \cite{seifnashri2024clusterstatenoninvertiblesymmetry} is equivalent to
\[
    H = - \sum_{I}P_{I,I+1}.
\]
The component of $P_{I,I+1}$ in each sector can be expanded as
\begin{gather*}
    P_e = \frac{1}{4}
    \begin{array}{c|cccc}
        & e\,e & s\,s & r^2\,r^2 & sr^2\,sr^2\\
        \hline
        e\,e &  1 & 1 & -1 & 1\\
        s\,s & 1 & 1 & -1 & 1\\
        r^2\,r^2 & -1 & -1 & 1 & -1\\
        sr^2\,sr^2 & 1 & 1 & -1 & 1\\
    \end{array},\quad 
    P_{s} = \frac{1}{4}
    \begin{array}{c|cccc}
         &  e\,s & s\,e & sr^2\,r^2 & r^2\,sr^2\\
         \hline
          e\,s &   1 & 1 & -1 & 1\\
          s\,e & 1 & 1 & -1 & 1\\
          sr^2\,r^2 & -1 & -1 & 1 & -1\\
          r^2\,sr^2 & 1 & 1 & -1 & 1
    \end{array},\\
    P_{r^2} = \frac{1}{4}
    \begin{array}{c|cccc}
         & e\,r^2 & r^2\,e & s\,sr^2 & sr^2\,s \\
         \hline
        e\,r^2 & 1 & 1 & 1 & -1\\
        r^2\,e & 1 & 1 & 1 & -1\\
        s\,sr^2 & 1 & 1 & 1 & -1\\
        sr^2\,s & -1 & -1 & -1 & 1
    \end{array},\quad
    P_{sr^2} = \frac{1}{4}
    \begin{array}{c|cccc}
         & e\,sr^2 & sr^2e\, & s\,r^2 & r^2\,s \\
         \hline
        e\,sr^2 & 1 & 1 & 1 & -1\\
        sr^2\,e & 1 & 1 & 1 & -1\\
        s\,r^2 & 1 & 1 & 1 & -1\\
        r^2\,s & -1 & -1 & -1 & 1
    \end{array}
\end{gather*}
Define
\begin{gather*}
    m_e=
    \begin{array}{c|cccc}
         & e\,e & s\,s & r^2\,r^2 & sr^2\,sr^2\\
         \hline
         e & 1 & 1 & -1 & 1
    \end{array}, \quad 
    m_s=
    \begin{array}{c|cccc}
         & e\,s & s\,e & sr^2\,r^2 & r^2\,sr^2\\
         \hline
         s & 1 & 1 & -1 & 1
    \end{array},\\
    m_{r^2}=
    \begin{array}{c|cccc}
         & e\,r^2 & r^2\,e & s\,sr^2 & sr^2\,s \\
         \hline
         r^2 & 1 & 1 & 1 & -1
    \end{array},\quad
        m_{sr^2}=
    \begin{array}{c|cccc}
         & e\,sr^2 & sr^2\,e & s\,r^2 & r^2\,s \\
         \hline
         sr^2 & 1 & 1 & 1 & -1
    \end{array}.
\end{gather*}
Then
\[
    P = \frac{1}{4}m^\dagger m.
\]
We can then extract the cocycle gauge by $m$:
\[
    \omega(r^2,r^2) = \omega(sr^2,r^2) = \omega(r^2,s) = \omega(sr^2, s) = -1,\quad \omega(\text{others}) = 1.
\]
Thus the Hamiltonian of \cite{seifnashri2024clusterstatenoninvertiblesymmetry}, for either $L = 0\mod{8}$ or $L = 4\mod{8}$, matches our model locally
\[
    H  = -\frac{1}{4}\sum_{i}m_{i,i+1}^\dagger m_{i,i+1}\,.
\]
And the MPO $A_\sigma$ acts on ground state as $\Tr(A_\sigma)|e\ra = 2|e\ra$.

The third model is constructed similarly. After applying the CZ gate, the ground state is stabilized by the generator
\[
    -X_{2n-1} = 1,\quad -Z_{2n-2}X_{2n}Z_{2n+2} = 1\,.
\]
Imposing the constraint $X_{2n-1} = -1$, the local Hilbert space on a two-site unit cell is $\C\{s,r\}$.
In $|s\ra, |r\ra$ basis, 
\[
    Z = |s\ra\la s| - |r\ra\la r|,\quad X = |s\ra\la r| + |r\ra\la s|\,.
\]
We then group 2 cells to form the final unit cell and
\[
    \C\{s,r\}\ot\C\{s,r\} = \C\{e,sr,r^2, sr^3\}
\]
with
\[
    e = s\,s,\quad sr = s\,r,\quad r^2 = r\,r,\quad sr^3 = r\,s.
\]
After reversing the sign, the local projector between two nearby sites to the unique ground state is
\[
    P_{I,I+1} = \frac{1}{4}\left(1 + Z_{4I}X_{4I+2}Z_{4I+4}\right)\left(1 + Z_{4I+2}X_{4I+4}Z_{4I+6}\right)\,,
\]
with
\begin{gather*}
    P_e = \frac{1}{4}
    \begin{array}{c|cccc}
        & e\,e & sr\,sr & r^2\,r^2 & sr^3\,sr^3\\
        \hline
        e\,e &  1 & 1 & -1 & 1\\
        sr\,sr & 1 & 1 & -1 & 1\\
        r^2\,r^2 & -1 & -1 & 1 & -1\\
        sr^3\,sr^3 & 1 & 1 & -1 & 1\\
    \end{array},\quad 
    P_{sr} = \frac{1}{4}
    \begin{array}{c|cccc}
         &  e\,sr & sr\,e & sr^3\,r^2 & r^2\,sr^3\\
         \hline
          e\,sr &   1 & 1 & 1 & -1\\
          sr\,e & 1 & 1 & 1 & -1\\
          sr^3\,r^2 & 1 & 1 & 1 & -1\\
          r^2\,sr^3 & -1 & -1 & -1 & 1
    \end{array},\\
    P_{r^2} = \frac{1}{4}
    \begin{array}{c|cccc}
         & e\,r^2 & r^2\,e & sr^3\,sr & sr\,sr^3 \\
         \hline
        e\,r^2 & 1 & 1 & 1 & -1\\
        r^2\,e & 1 & 1 & 1 & -1\\
        sr^3\,sr & 1 & 1 & 1 & -1\\
        sr\,sr^3 & -1 & -1 & -1 & 1
    \end{array},\quad
    P_{sr^3} = \frac{1}{4}
    \begin{array}{c|cccc}
         & e\,sr^3 & sr^3\,e & sr\,r^2 & r^2\,sr \\
         \hline
        e\,sr^3 & 1 & 1 & -1 & 1\\
        sr^3\,e & 1 & 1 & -1 & 1\\
        sr\,r^2 & -1 & -1 & 1 & -1\\
        r^2\,sr & 1 & 1 & -1 & 1
    \end{array}
\end{gather*}
Define
\begin{gather*}
    m_e=
    \begin{array}{c|cccc}
         & e\,e & s\,s & r^2\,r^2 & sr^2\,sr^2\\
         \hline
         e & 1 & 1 & -1 & 1
    \end{array}, \quad 
    m_{sr}=
    \begin{array}{c|cccc}
         &  e\,sr & sr\,e & sr^3\,r^2 & r^2\,sr^3\\
         \hline
         s & 1 & 1 & 1 & -1
    \end{array},\\
    m_{r^2}=
    \begin{array}{c|cccc}
         & e\,r^2 & r^2\,e & sr^3\,sr & sr\,sr^3 \\
         \hline
         r^2 & 1 & 1 & 1 & -1
    \end{array},\quad
    m_{sr^3}=
    \begin{array}{c|cccc}
         & e\,sr^3 & sr^3\,e & sr\,r^2 & r^2\,sr \\
         \hline
         sr^3 & 1 & 1 & -1 & 1
    \end{array}.
\end{gather*}
Then
\[
    P = \frac{1}{4}m^\dagger m.
\]
We can then extract the cocycle gauge by $m$:
\[
    \omega(r^2,r^2) = \omega(r^2,sr^3) = \omega(sr,sr^3) = \omega(r^2, sr) = -1,\quad \omega(\text{others}) = 1.
\]

\section{\texorpdfstring{$\Rep^\dagger(\sH^*_8)$}{Rep(H8)}-MPOs on irreducible charges}\label{app:H8}
We first list the data for $\sH_8$:
\begin{equation*}
    \sH_8= \Big \la x,y,x|x^2=y^2=z^2=\one,\, xz=zx,\,yz=zy,\,xyz=yx \Big \ra \,,
\end{equation*}
with basis
\begin{equation*}
    \Big \{\one,\,x,\,y,\,z,\,xz,\,yz,\,xy,\,yx \Big\},
\end{equation*}
where $z$ is the central element. Thus we identify two central orthogonal idempotents
\begin{equation*}
    p_0 = \frac{1}{2}(1+z),\quad p_1 = \frac{1}{2}(1-z).
\end{equation*}
The counits are $\epsilon(x) = \epsilon(y) = \epsilon (z)=1$ and the comultiplications are given by
\begin{align*}
    & \Delta 1 = 1 \ot 1,\quad \Delta x = x p_0 \ot x+ x p_1 \ot y,\quad \Delta y = yp_1 \ot x + y p_0 \ot y,\quad \Delta z = z \ot z,\\
    & \Delta(xz) = (\Delta x)(\Delta z) = x p_0 \ot xz-x p_1 \ot yz, \quad \Delta(yz) = (\Delta y)(\Delta z) = -yp_1 \ot xz+y p_0 \ot yz, \\
    &\Delta (xy) = (\Delta x)(\Delta y)=xyp_0 \ot xy + xyp_1 \ot yx,\quad \Delta(yx) = yxp_1 \ot xy+yxp_0 \ot yx.
\end{align*}
The antipode $S$ is defined as
\begin{align*}
    &S(x) = xp_0 + yp_1,\quad S(y) = x p_1+yp_0,\quad S(z)=z,\\
    &S(xy) = yx,\quad S(xz) = \frac{1}{2}(x-y+xz+yz),\quad S(yz) = \frac{1}{2}(-x+y+yz+xz), \quad S(yx)=xy.
\end{align*}

The MPO basis $\Gamma=\Big\{\bl{\one},\, \bl{a},\, \bl{b},\, \bl{c},\, \bl{\sigma}^{11},\, \bl{\sigma}^{12},\,\bl{\sigma}^{21},\,\bl{\sigma}^{22} \Big\}$ is explicitly given by the basis transformation
\begin{align*}
    &\bl{\one}=\one,\quad
    \bl{a}=\frac{1+\ii}{2}(xy -\ii yx),\quad
    \bl{b}=\frac{1+\ii}{2}(-\ii xy + yx),\quad
    \bl{c}=z,\\
    &\bl{\sigma}^{11} = xp_0,\quad \bl{\sigma}^{12} = xp_1\quad \bl{\sigma}^{21} = yp_1,\quad \bl{\sigma}^{22} = yp_0.\\
\end{align*}
Here we use hat to denote objects in $\Rep^\dagger(\sH_8^*)$ and those without hat to denote objects in $\Rep^\dagger(\sH_8)$.
As a sanity check, we can show that the above basis satisfies MPO basis properties in \ref{sec:Hopf algebra symmetry MPO}, meaning
\begin{itemize}
    \item [--] multiplication: $\bl{a}\bl{a} = \bl{\one},\; \bl{b} = \bl{a}\bl{c} = \bl{c}\bl{a}$;
    \item [--] comultiplication: $\Delta \bl{a} = \bl{a}\ot \bl{a},\; \Delta \bl{c} = \bl{c}\ot \bl{c}, \; \Delta \bl{\sigma}^{\alpha\beta} = \sum_{\gamma}\bl{\sigma}^{\alpha\gamma}\ot \bl{\sigma}^{\gamma\beta}$;
    \item [--] antipode: $S(\bl{a}) = \bl{a},\; S(\bl{c}) = \bl{c},\; S(\bl{\sigma}^{\alpha\beta}) = \bl{\sigma}^{\beta\alpha}$.
\end{itemize}
Representation category of $\sH_8$ is isomorphic to the Tambara-Yamagami category $\TY_{\Z_2\xt \Z_2}^{\chi(g,h)=(-1)^{g_1h_1+g_2h_2},\epsilon=1}$. The irreducible representations are listed in Table \ref{tab:RepH8}.
\begin{table}
    \centering
    \begin{tabular}{|c|c|c|c|c|c|}
    \hline
        \diagbox{$\sH_8$}{Irr} & $\one$ & $a$ & $b$ & $c$ & $\sigma$  \\
        \hline
       $x$  & $1$ & $-1$ & $1$ & $-1$ & $X$\\
       \hline
       $y$  & $1$ & $1$ & $-1$ & $-1$ & $Y$\\
       \hline
       $z$  & $1$ & $1$ & $1$ & $1$ & $-\one$\\
       \hline
    \end{tabular}
    \caption{Irreducible representations of $\sH_8$ on generators.}
    \label{tab:RepH8}
\end{table}
We then calculate $\Rep^\dagger(\sH_8^*)$-MPOs on irreducible $\sH_8$-charges. Invertible MPOs are:
\begin{align*}
    &(T_\bl{\one})^{11}_\one = (T_\bl{\one})^{11}_a = (T_\bl{\one})^{11}_b = (T_\bl{\one})^{11}_c = 1,\quad (T_\bl{\one})^{11}_\sigma = \one,\\
    &(T_\bl{a})^{11}_\one=(T_\bl{a})^{11}_c = 1,\quad  (T_\bl{a})^{11}_a =  (T_\bl{a})^{11}_b = -1,\quad (T_\bl{a})^{11}_\sigma = -Z,\\
    &(T_\bl{b})^{11}_\one=(T_\bl{b})^{11}_c = 1,\quad  (T_\bl{b})^{11}_a =  (T_\bl{b})^{11}_b = -1,\quad (T_\bl{b})^{11}_\sigma = Z,\\
    &(T_\bl{c})^{11}_\one=(T_\bl{c})^{11}_a = (T_\bl{c})^{11}_b =  (T_\bl{c})^{11}_c = 1,\quad (T_\bl{c})^{11}_\sigma = -\one.
\end{align*}
And components of the non-invertible MPO are:
\begin{align*}
    &(T_\bl{\sigma})^{11}_\one =(T_\bl{\sigma})^{11}_b = \one,\quad (T_\bl{\sigma})^{11}_a =(T_\bl{\sigma})^{11}_c =  -\one,\quad (T_\bl{\sigma})^{11}_\sigma = \mathbf{0},\\
    &(T_\bl{\sigma})^{12}_\one =(T_\bl{\sigma})^{12}_a = (T_\bl{\sigma})^{12}_b =(T_\bl{\sigma})^{12}_c =  \mathbf{0},\quad (T_\bl{\sigma})^{12}_\sigma = X,\\
    &(T_\bl{\sigma})^{21}_\one =(T_\bl{\sigma})^{21}_a = (T_\bl{\sigma})^{21}_b =(T_\bl{\sigma})^{21}_c =  \mathbf{0},\quad (T_\bl{\sigma})^{21}_\sigma = Y,\\
    &(T_\bl{\sigma})^{22}_\one =(T_\bl{\sigma})^{22}_a = \one,\quad (T_\bl{\sigma})^{22}_b =(T_\bl{\sigma})^{22}_c =  -\one,\quad (T_\bl{\sigma})^{22}_\sigma = \mathbf{0}.
\end{align*}

\section{Some basics of category theory}
\begin{definition}[Module Category] 
Let $(\cC,\ot,\mathbf{1},\alpha,\lambda,\rho)$ be a linear monoidal category, where $\mathbf{1}$ is the tensor unit, $\alpha,\,\lambda,\,\rho$ are natural isomorphisms called associator, left unitor and right unitor respectively. A left module category over $\cC$ is a linear category $\cM$ equipped with a linear monoidal functor $\cC \rightarrow \Fun(\cM,\cM)$, or equivalently, a bilinear functor $\rhd:\cC \times \cM \rightarrow \cM$ with
\begin{itemize}
    \item An associator: a natural transformation $m$: $(-\ot-)\rhd- \Rightarrow - \rhd(- \rhd -)$ $\forall \,X,\,Y,\,Z \in \cC$, $M \in \cM$ such that the functor $M \mapsto \mathbf{1} \rhd M: \cM \rightarrow \cM$ is an autoequivalence, and such that the pentagon diagram commutes
    \[
    \begin{tikzcd}[cramped]
	{((X \otimes Y)\otimes)\rhd M} &&&& {(X \otimes Y) \rhd (Z \rhd M)} \\
	\\
	{(X \otimes(Y \otimes Z)) \rhd M} && {X \rhd((Y \otimes Z) \rhd M)} && {X          \rhd(Y \rhd(Z \rhd M))}
	\arrow["{m_{X \otimes Y,Z,M}}", from=1-1, to=1-5]
	\arrow["{\alpha_{X,Y,Z} \rhd \text{id}_M}"', from=1-1, to=3-1]
	\arrow["{m_{X,Y,Z \rhd M}}", from=1-5, to=3-5]
	\arrow["{m_{X,Y\otimes Z, M}}"', from=3-1, to=3-3]
	\arrow["{\text{id}_X \rhd m_{Y,Z,M}}"', from=3-3, to=3-5]
    \end{tikzcd} \,;\]
    \item A unitor: a natural transformation $\mu:\mathbf{1} \rhd - \rightarrow \text{id}_{\cM}$ such that $\forall\,X \in \cC$, $M \in \cM$,
    \[\begin{tikzcd}[cramped]
	{(\mathbf{1} \otimes X) \rhd M} && {\mathbf{1} \rhd (X \rhd M)} \\
	& {X \rhd M}
	\arrow["{m_{\mathbf{1},X,M}}", from=1-1, to=1-3]
	\arrow["{\lambda_X \rhd \text{id}_M}"', from=1-1, to=2-2]
	\arrow["{\mu_{X \rhd M}}", from=1-3, to=2-2]
    \end{tikzcd} \,,\]
    \[\begin{tikzcd}[cramped]
	{(X \otimes \mathbf{1}) \rhd M} && {X \rhd (\mathbf{1} \rhd M)} \\
	& {X \rhd M}
	\arrow["{m_{X,\mathbf{1},M}}", from=1-1, to=1-3]
	\arrow["{\rho_X \rhd \text{id}_M}"', from=1-1, to=2-2]
	\arrow["{\text{id}_X \rhd \mu_M}", from=1-3, to=2-2]
    \end{tikzcd} \,.\] 
\end{itemize}
A right $\cC$-module is defined similarly.
    
\end{definition}
\begin{definition}[$\cC$-module functor]
    Let $\cM$ and $\cN$ be two module categories over $\cC$ with associativity constraints $m$ and $n$, respectively. A $\cC$-module functor from $\cM$ to $\cN$ consists of a linear functor $F:\cM \rightarrow \cN$ and a natural isomorphism
\begin{equation*}
    s_{X,M}:F(X \rhd M) \rightarrow X \rhd F(M),\quad X \in \cC,\; M \in\cM \,,
\end{equation*}
such that the following diagrams
\[\begin{tikzcd}[cramped]
	{F(X \rhd (Y \rhd M))} &&& {F((X \otimes Y)\rhd M)} &&& {(X \otimes Y) \rhd F(M)} \\
	\\
	{X \rhd F(Y \rhd M)} &&&&&& {X \rhd (Y \rhd F(M))}
	\arrow["{F(m_{X,Y,M})}"', from=1-4, to=1-1]
	\arrow["{s_{X \otimes Y,M}}", from=1-4, to=1-7]
	\arrow["{n_{X,Y,F(M)}}", from=1-7, to=3-7]
	\arrow["{s_{X,Y\rhd M}}", from=3-1, to=1-1]
	\arrow["{\text{id}_X \rhd s_{Y,M}}", from=3-1, to=3-7]
\end{tikzcd}\,,\]
\[\begin{tikzcd}[cramped]
	{F(\mathbf{1}\rhd M)} && {\mathbf{1} \rhd F(M)} \\
	& {F(M)}
	\arrow["{s_{\mathbf{1},M}}", from=1-1, to=1-3]
	\arrow["{F(\mu_M)}"', from=1-1, to=2-2]
	\arrow["{\mu_{F(M)}}", from=1-3, to=2-2]
\end{tikzcd}\]
commute $\forall\,X,\,Y\in \cC$ and $M \in \cM$.
\end{definition}

\begin{definition}[Internal Hom]
    Let $\cC$ be a monoidal category and $\cM$ be a $\cC$-module category. Given $M,\,N \in \cM$. If the functor $X \mapsto \cM(X \rhd M,N):\cC \rightarrow \Ve$ is representable, i.e., there exists an object $[M,N] \in \cC$ and a natural isomorphism
    \begin{equation*}
       \cM(X \rhd M,N) \cong {\cC}(X,[M,N])\,,
    \end{equation*}
    then $[M,N]$ is called the internal Hom.
\end{definition}

\begin{definition}[$\Hilb_G$]\label{def:HilbG}
    Let $G$ be a finite group. The category of finite-dimensional $G$-graded Hilbert spaces, denoted by $\Hilb_G$, is a unitary fusion category consists of the following data:
    \begin{itemize}
        \item[--] objects: $G$-graded Hilbert spaces $V \simeq \oplus_{g\in G}V_g$, and for all $u_g\in V_g$ and $v_h\in V_h$, $\la u_g|v_h \ra =0$ if $g\neq h$;
        \item[--] morphisms: $G$-grading preserving linear maps;
        \item[--] the tensor product: given $G$-graded Hilbert spaces $V = \oplus_{g\in G}V_g$ and $W = \oplus_{g\in G}W_g$, the $g$-graded subspace of $V\ot W$ is $(V\ot W)_{g} = \oplus_{h\in G}V_h\ot W_{h^{-1}g}$;
        \item[--] dual: the $g$-graded subspace of $V^*$ is $(V^*)_{g} = \left(V_{g^{-1}}\right)^*$.
    \end{itemize}
\end{definition}

\begin{definition}
    Let $(\cC,\ot)$ and $(\cD,\odot)$ be unitary fusion categories. A monoidal functor $F:\cC\rightarrow \cD$ with monoidal structure $\eta:\odot\circ (F\times F) \Rightarrow F\circ\ot$ is called unitary if $F$ is a $\star$-functor (i.e., $F$ is linear and $F(f^\dag)=F(f)^\dag$), $F(\one_\cC)$ is unitarily isomorphic to $\one_\cD$ and $\eta_{X,Y}$ is unitary for every $X,Y\in\cC$.
\end{definition}

\begin{definition}[Unitary fiber functor]\label{def:fiber functor} For a tensor category $\cC$, a fiber functor is an exact faithful monoidal functor $\cC\to\Ve$.
   For a unitary fusion category $\cC$, we focus on unitary fiber functor which is a unitary monoidal functor $F:\cC\rightarrow \Hilb$ (note that the exactness and faithfulness are automatic for a linear functor from a semisimple category).
\end{definition}
We now expand the definition of the fiber functor. A unitary fiber functor is a unitary monoidal functor $F:\cC\rightarrow \Hilb$, with
\begin{itemize}
     \item a map from objects of $\cC$ to objects of $\Hilb$;
    \item for each $X,Y$ in $\cC$, an injective linear map from $\cC(X,Y)$ to $\Hilb(F(X), F(Y))$, ;
    \item for each $X, Y$ in $\cC$, a unitary natural isomorphism $F_{X,Y}$ from $F(X)\ot F(Y)$ to $F(X\ot Y)$,
\end{itemize}
satisfying the following conditions:
\begin{itemize}
    \item for any $X\xrightarrow{f}Y\xrightarrow{g}Z$ in $\cC$, we have $F(g)F(f) = F(gf)$ and $F(\id_X) = \id_{F(X)}$;
    \item $F(\one)\simeq \C$;
    \item for any $X, Y, Z$ in $\cC$, the following diagram commutes:
\[\begin{tikzcd}
	& {F(X)\ot F(Y)\ot F(Z)} \\
	{F(X\ot Y)\ot F(Z)} && {F(X)\ot F(Y\ot Z)} \\
	{F((X\ot Y)\ot Z)} && {F(X\ot (Y\ot Z))}
	\arrow["{F_{X,Y}\ot \id_{F(Z)}}"', from=1-2, to=2-1]
	\arrow["{\id_{F(X)}\ot F_{Y,Z}}", from=1-2, to=2-3]
	\arrow["{F_{X\ot Y, Z}}", from=2-1, to=3-1]
	\arrow["{F_{X, Y\ot Z}}"', from=2-3, to=3-3]
	\arrow["{F(\alpha_{X, Y,Z})}", from=3-1, to=3-3]
\end{tikzcd}\]
\end{itemize}
In a word, \textit{a fiber functor is a consistent assignment of Hilbert spaces to objects}. This interpretation is used throughout the paper: each virtual bond of the fusion category MPO assumes a \textit{dual role}—it simultaneously serves as an object in the symmetry category and as a Hilbert space. Similarly, each Hilbert space on a site also plays a \textit{dual role}—it is both an object in the charge category and a Hilbert space.

\begin{definition}[Balanced dual]\label{def:balanced dual}
    Let $\cC$ be a unitary fusion category and $X\in\cC$, $(X^*,\ev_X,\coev_X)$ is called a balanced dual of $X$ if
    \[
    \ev_X(\id_{X^*}\ot f)\ev_X^\dagger  = \coev^\dagger_X (f\ot \id_{X^*})\coev_X
    \]
    is satisfied for any $f:X\rightarrow X$.
\end{definition}

\begin{theorem}[Tannaka Duality, see \cite{etingof2015tensor}]\label{thm:tannaka}
    The assignments
    \begin{equation*}
        (\cC, F) \mapsto \sH = \categoricalEnd(F),\qquad \sH \mapsto(\Rep(\sH),\mathrm{fgt})
    \end{equation*}
    are mutually inverse bijections between
    \begin{enumerate}
        \item equivalence classes of finite tensor categories $\cC$ with a fiber functor $F$, up to tensor equivalence and isomorphism of tensor functors, and
        \item isomorphism classes of finite-dimensional Hopf algebras.
    \end{enumerate}
\end{theorem}

\section{Algebra}
\begin{definition}(Algebra)
  Let $\cC$ be a monoidal category. An (associative unital) algebra in $\cC$ is a triple $(A,m,\iota)$, consisting of an object $A \in \cC$ together with a multiplication $m:A \ot A \rightarrow A$ and a unit morphism $\iota:\mathbf{1} \rightarrow A$ satisfying associativity and identity:
  \[\begin{tikzcd}[cramped]
	{(A \otimes A) \otimes A} && {A \otimes(A \otimes A)} \\
	{A \otimes A} && {A \otimes A} \\
	& A
	\arrow["{\alpha_{A,A,A}}", from=1-1, to=1-3]
	\arrow["{m \otimes \text{id}_A}"', from=1-1, to=2-1]
	\arrow["{\text{id}_A\otimes m}", from=1-3, to=2-3]
	\arrow["m"', from=2-1, to=3-2]
	\arrow["m", from=2-3, to=3-2]
\end{tikzcd}\,,\]
\[\begin{tikzcd}[cramped]
	{\mathbf{1} \otimes A} && {A \otimes \mathbf{1}} \\
	{A \otimes A} & A & {A \otimes A}
	\arrow["{\iota \otimes \text{id}_A}"', from=1-1, to=2-1]
	\arrow["{\text{id}_A}", from=1-1, to=2-2]
	\arrow["{\text{id}_A}"', from=1-3, to=2-2]
	\arrow["{\text{id}_A \otimes \iota}", from=1-3, to=2-3]
	\arrow["m"', from=2-1, to=2-2]
	\arrow["m", from=2-3, to=2-2]
\end{tikzcd}\]
where $\text{id}_A:A \rightarrow A$ is the unique identity map on $A$.
\end{definition}

\begin{definition}[Module over an algebra]
Given an algebra $(A,m,\iota)$ in $\cC$, a right $A$-module is a pair $(M,\rho)$, where $M \in \cC$ and $\rho: M \otimes A \rightarrow M$ is an action morphism such that
\[\begin{tikzcd}[cramped]
	{(M \otimes A)\otimes A} && {M \otimes (A \otimes A)} \\
	{M \otimes A} && {M \otimes A} \\
	& M
	\arrow["{\alpha_{M,A,A}}", from=1-1, to=1-3]
	\arrow["{\rho \otimes \text{id}_A}"', from=1-1, to=2-1]
	\arrow["{\text{id}_M \otimes m}", from=1-3, to=2-3]
	\arrow["\rho"', from=2-1, to=3-2]
	\arrow["\rho", from=2-3, to=3-2]
\end{tikzcd}\,,\]
\[\begin{tikzcd}[cramped]
	{M \otimes \mathbf{1}} \\
	\\
	{M \otimes A} && M
	\arrow["{\text{id}_M \otimes \iota}"', from=1-1, to=3-1]
	\arrow["{\text{id}_M}", from=1-1, to=3-3]
	\arrow["\rho"', from=3-1, to=3-3]
\end{tikzcd}\,.\]
A left $A$-module is defined similarly.
\end{definition}

\begin{remark}
    Given $A$-modules $M,\,N$ in $\cC$, module homomorphisms between them form a subspace $\cC_A(M,N)$ of the vector space $\cC(M,N)$. The composition of homomorphisms is a homomorphism, thus right $A$-modules in $\cC$ form a category $\cC_A$.
\end{remark}

\begin{definition}
    Given a fusion category $\cC$, and an indecomposible semisimple $\cC$-module $\cM$,  denote category $\mathsf{Fun}_\cC(\cM,\cM)$ as $\cC^\vee_\cM$ and call it the dual fusion category to $\cC$ with respect to $\cM$.
\end{definition}

\begin{remark}
    Let $\cC$ be a fusion category. The collection of all $\cC$-module categories forms a 2-category $\mathsf{Mod}(\cC)$.
\end{remark}

\begin{theorem}
    Let $\cM$ be a faithful exact module category over a multitensor category $\cC$. The 2-functor
    \begin{equation*}
        \cN \mapsto \mathsf{Fun}_\cC(\cM,\cN):\mathsf{Mod}(\cC) \rightarrow \mathsf{Mod}((\cC_\cM^\vee)^{\text{rev}})
    \end{equation*}
    is a 2-equivalence.
\end{theorem}

\begin{definition}[Categorical Morita equivalence]
    Let $\cC$, $\cD$ be fusion categories. $\cC$ and $\cD$ are categorically Morita equivalent if there is an indecomposible semisimple $\cC$-module category $\cM$ and a monoidal equivalence $\cD^{\text{rev}} \cong \cC^\vee_\cM$.
\end{definition}
    
\begin{definition}[Hopf Algebra]
    Let $\sH$ be a finite-dimensional vector space over $\C$, $m:\sH\ot \sH\rightarrow \sH$, $\iota:\C\rightarrow\sH$, $\Delta:\sH\rightarrow\sH\ot \sH$, and $\epsilon:\sH\rightarrow \C$ be linear maps, such that $(\sH,m,\iota)$ is an algebra with multiplication $m$ and unit $\iota$ and $(\sH,\Delta, \epsilon)$ is a coalgebra with comultiplication $\Delta$ and counit $\epsilon$. If they further satisfy the conditions: 
    \begin{itemize}
        \item $\Delta$ and $\epsilon$ are algebra homomorphisms with the multiplication on $\sH\ot \sH$ be defined as $(m\ot m)(\id\ot \mathrm{swap} \ot \id)$;
        \item Comultiplication of the unit $\one$ gives $\one \ot \one$;
        \item The counit of the unit is $1$: $\epsilon(\one) = 1$;
    \end{itemize}
    the tuple $(\sH, m,\iota, \Delta, \epsilon)$ is called a bialgebra. If there exists an invertible map (necessarily unique) $S:\sH\rightarrow \sH$ called the antipode satisfying the condition
    \[
        S(a_{(1)})a_{(2)} = \epsilon(a)\one = a_{(1)}S(a_{(2)}),
    \]
    the bialgebra $(\sH, m, \iota, \Delta, \epsilon)$ is called a Hopf algebra.
\end{definition}

\begin{theorem}[See for example, theorem 1.6.6 of \cite{Penneys2023}]\label{thm:Cstar}
    The following conditions are equivalent for a finite-dimensional unital complex $\star$-algebra $A$:
    \begin{enumerate}
        \item $A$ is a C$^\star $-algebra;
        \item There exists a $\star$-homomorphism $A\simeq \oplus_{i = 1}^{n}\M_{a_i}$, where each summand has the usual conjugate transpose $\dagger$-operation;
        \item For every $a\in A$, $aa^\star  = 0$ implies $a = 0$.
    \end{enumerate}
\end{theorem}
Hopf C$^\star $-algebra is defined in \cite{vaes1999hopfcalgebras}, and a version of C$^\star $-quantum groupoid is defined in \cite{nikshych2000finitequantumgroupoidsapplications}.

\section{Operator algebraic description of ground state} \label{app:O-type ground state}
\begin{example}\label{example:ising}
Consider the 1D quantum Ising model in symmetry breaking phase, characterized by the local interaction:
\[
    H_{\Lambda} = \sum_{\{i,i+1\}\subset\Lambda}\frac{1}{2}\left(\one - Z_iZ_{i+1}\right)
\]
for any finite region $\Lambda \in\fS_{\Z}^f$. Let $p_{i,i+1} = \frac{1}{2}\left(\one - Z_iZ_{i+1}\right)$. It can be verified that the $H$-type ground state $\omega$ satisfying $\omega(p_{i,i+1}) = 0$ is an $O$-type ground state \cite{Alicki_2007} and the space of $H$-type ground states forms a simplex $[0,1]$ whose extreme points $0$ and $1$ are two superselection sectors, corresponding to $\la \cdots 00000\cdots|(-)|\cdots 00000\cdots\ra$ and $\la \cdots 11111\cdots|(-)|\cdots 11111\cdots\ra$, respectively. 
What might be counterintuitive is that the single domain wall state, with the domain wall located at the link $( 0,1)$, defined by
\[\omega = \la \cdots 00001111\cdots|(-)|\cdots 00001111\cdots\ra\]
is also a ground state in \ref{def:groundstate}. We now verify it.

For an arbitrary $O\in\fA_{\Lambda}$, if $\Lambda \subset \Z_{>1}$ (or $\Lambda \subset \Z_{<0}$), the ground state condition is given by
\[
    \omega(O^\dagger H_{\overline{\Lambda}}O)-0\geq 0.
\]
which is trivially satisfied by the positivity of $H_{\overline{\Lambda}}$. When $\overline{\Lambda}$ crosses the regions $\Z_{\leq 0}$ and $\Z_{>0}$, the ground state condition becomes
\[
    \omega\left(O^\dagger [H_{\overline{\Lambda}},O]\right) 
     = \omega\left(O^\dagger H_{\overline{\Lambda}}O\right) - \omega\left(O^\dagger O\right)\geq 0.
\]
If $\omega\left(O^\dagger O\right) = 0$, the inequality is satisfied by the positivity of $H_{\overline{\Lambda}}$. If $\omega\left(O^\dagger O\right) > 0$, define 
\[|\eta\ra := O\bigg|0[00\cdots 0011\cdots 11]1\bigg\ra \in (\C^2)^{\otimes |\overline{\Lambda}|},\]
above condition is then equivalent to
\[
    \frac{\la \eta| H_{\overline{\Lambda}}|\eta\ra}{\la \eta|\eta\ra}\geq 1.
\]
We use the square bracket to denote the region $\Lambda$.
This inequality is satisfied because operators supported on $\Lambda$ can not annihilate a single domain wall; they can only move or create a pair of the domain walls within $\overline{\Lambda}$, thus the energy of $|\eta\ra$ can only be larger or equal to the energy of $\bigg|0[0\cdots 01\cdots 1]1\bigg\ra$. This ground state property is fundamentally linked to the topological nature of the single domain wall state.
\end{example}

\section{Q-system}
Following \cite{chen2024qsystemcompletionc2categories}, the Q-system is defined as follows:
\begin{definition}[Q-system]
    Let $\cC$ be a unitary fusion category. A Q-system in $\cC$ is a triple $(A,m,\iota)$, where
    \begin{enumerate}
        \item $(A,m,\iota)$ is a unital associative isometric algebra in $\cC$ and;
        \item $m$ and $m^\dagger$ satisfy the Frobenius condition;
        \item $\ev_A:=\iota^\dagger m:A\ot A\rightarrow\C$ and $\coev_A:=m^\dagger \iota:\C\rightarrow A\ot A$ is the balanced dual (see \ref{def:balanced dual}) for $A$.
    \end{enumerate}
\end{definition}
It has been proved in \cite{K-th,zito20052ccategoriesnonsimpleunits,Qsys,Lan_2024} that 1$\Rightarrow$2.
And condition 3 implies that $A$ is symmetrically self-dual, see for example \cite{penneys2018unitarydualfunctorsunitary,chen2024qsystemcompletionc2categories}. We denote $m_{(n)}$ by multiplying $n$ copies of $A$'s to a single $A$.

\begin{definition}[Irreducible Q-system, see for example Lemma 3.13 of \cite{Qsys}]
    A Q-system $(A,m,\iota)$ in a unitary fusion category $\cC$ is called irreducible if $\dim \cC(\one, A) = 1$.
\end{definition}

Motivated by Lemma 3.16 of \cite{Qsys}, we define
\begin{definition}\label{def:unit alg}
    Let $(A,m,\iota)$ be a simple Q-system in a unitary fusion category $\cC$, then the tuple $(\cC(\one, A),\cdot,\iota,\star)$ is a finite-dimensional unital associative $\star$-algebra with
\begin{itemize}
    \item multiplication: $f\cdot g:= m(f\ot g)$ for any $f,g\in \cC(\one, A)$;
    \item unit: $\iota\cdot f = f\cdot\iota = f$ for any $f$;
    \item star: $f^\star := (f^\dagger\ot \id_A)m^\dagger\iota = (\id_A\ot f^\dagger)m^\dagger\iota$ for any $f$.
\end{itemize}
\end{definition}

It is a C$^\star $-algebra by the theorem \ref{thm:Cstar}
\[
    f\cdot f^\star = 0\Rightarrow \iota^\dagger m(f\ot \id_A)(\iota^\dagger m(f\ot \id_A))^\dagger = 0\Rightarrow \iota^\dagger m(f\ot \id_A) = 0\Rightarrow f = 0.
\]
By Lemma 3.17 of \cite{Qsys}:
\begin{proposition}
    There is a $\star$-isomorphism
    \begin{gather*}
        \cC(\one,A)\simeq \cC_{\one|A}(A,A)\\
        f\mapsto m(f\ot \id_A)\\
        g\circ \iota \mapsfrom g .    
    \end{gather*}
    Similarly, there is also a $*$-isomorphism $\cC(\one,A)\simeq \cC_{A|\one}(A,A)$.
\end{proposition}

\begin{definition}[Modules]
    Let $\cC$ be a unitary fusion category and $(A,m,\iota)$ a Q-system in $\cC$. The left (right) $A$-module $(M,m_M)$ ($(N,{}_{N}m)$) is called a Q-system module if $m_M$ (${}_{N}m$) is partial isometric. 
\end{definition}
By the Definition 3.37 of \cite{Qsys}
\begin{definition}[Relative tensor product]
    Let $A, B,C$ be Q-systems in $\cC$ and $(M,\mu_A,\mu_B)$ is an $A$-$B$-bimodule, $(N,\nu_B,\nu_C)$ is a $B$-$C$-bimodule. Suppose
    \[
        P_\otimes:=
        \diagram{1}{
            \draw (0,0) -- (0,1.5);
            \draw (1,0) -- (1,1.5);
            \draw (0,.9) node[left]{$M$} -- (.5,.6) -- (1,.9) node[right]{$N$};
            \draw (.5,.4) -- (.5,.6) node[above]{$B$};
            \filldraw[fill = white, draw = black] (.5,.4) circle [radius = .05];
        }
    \]
    splits as
    \[
        P_\otimes = S^\dagger S,\quad S:M\ot N\rightarrow M\ot[B]N,
    \]
    then $M\ot[B]N$ is an $A$-$C$-bimodule, whose left $A$-action is
    \[
        m_{L}:=S(\mu_A\ot \id_N)(\id_A\ot S^\dagger)
    \]
    and right $C$-action is
    \[
        m_R := S(\id_M\ot \nu_C)(S^\dagger\ot \id_C).
    \]
\end{definition}
We now review the Morita equivalence between Q-systems.
\begin{definition}[Morita equivalence]
    Let $A,B$ be Q-systems in $\cC$. Then $A$ and $B$ are called Morita equivalent if there is an $A$-$B$-bimodule $M$ and a $B$-$A$-bimodule $N$ such that $M\ot[B]N\simeq A$ and $N\ot[A]M\simeq B$.
\end{definition}
\begin{remark}[\cite{kong2019semisimple}]\label{remark:matrix algebra}
    Suppose $A$ and $B$ are Morita equivalent by an $A$-$B$-bimodule $M$ and a $B$-$A$-bimodule $N$, then one can construct a larger Q-system
    \[
    X = \begin{pmatrix}
        A & M\\ N & B
    \end{pmatrix}:=A\oplus M\oplus N\oplus B,
    \]
    whose algebra structure is the multiplication by block.
\end{remark}
We review the definition of the simple Q-system.
\begin{definition}[Simple Q-system]
    The Q-system $(A,m,\iota)$ in $\cC$ is called simple if $A$ is simple in $\cC_{A|A}$.
\end{definition}
As a special case, we consider Q-systems in $\Hilb$:
\begin{definition}[Center]
    Let $(A,m,\iota)$ be the Q-system in $\Hilb$, the center of $A$ is a subalgebra of $\hom(\C,A)\simeq A$ defined by
    \[
        Z(A):=\{a\in A: ab = ba,\ \forall b\in A\}.
    \]
\end{definition}
It can be proved that 
\begin{proposition}
    The C$^\star $-structure of $\hom(\C,A)$ induces a C$^\star $-structure of $Z(A)$.
    And $Z(A)$ is $\star$-isomorphic to $\hom_{A|A}(A,A)$.
\end{proposition}
\begin{proof}
    We denote $a\in A$ by
    \[
        a = 
        \diagram{1}{
            \draw (0,0) -- (0,.5);
            \fill[black] (-.1,-.1) rectangle (.1,.1);
        },\quad a^\dagger = 
        \diagram{1}{
            \draw (0,0) -- (0,.5);
            \fill[black] (-.1,.4) rectangle (.1,.6);
        }.
    \]
    
    Then $a$ is in $Z(A)$ if and only if
    \begin{equation}\label{eq:central}
        \diagram{1}{
            \draw (0,0) -- (0,1);
            \fill[black] (-.1,-.1) rectangle (.1,.1);
            \draw (0,.5) -- (.5,0) -- (.5,-1);
        }\quad =\quad 
        \diagram{1}{
            \draw (0,0) -- (0,1);
            \fill[black] (-.1,-.1) rectangle (.1,.1);
            \draw (0,.5) -- (-.5,0) -- (-.5,-1);
        }.
    \end{equation}
    Then $a^\star$ is central:
    \[
        \diagram{1}{
            \fill[black] (-.1,-.1) rectangle (.1,.1);
            \draw (0,0) -- (.25,-.25) -- (.5,0) -- (.5,1);
            \draw (.5,.5) -- (1,0) -- (1,-1);
            \draw (.25,-.5) -- (.25,-.25);
            \filldraw[fill = white,draw = black] (.25,-.5) circle [radius = .05];
        }=
        \diagram{1}{
            \fill[black] (-.1,-.1) rectangle (.1,.1);
            \draw (0,0) -- (.5,-.5);
            \draw (.5,-1) -- (.5,1);
        }
        =
        \diagram{1}{
            \fill[black] (-.1,-.1) rectangle (.1,.1);
            \draw (0,0) -- (-.5,-.5);
            \draw (-.5,-1) -- (-.5,1);
        }
        =
        \diagram{1}{
            \fill[black] (-.1,-.1) rectangle (.1,.1);
            \draw (0,0) -- (-.25,-.25) -- (-.5,0) -- (-.5,1);
            \draw (-.5,.5) -- (-1,0) -- (-1,-1);
            \draw (-.25,-.5) -- (-.25,-.25);
            \filldraw[fill = white,draw = black] (-.25,-.5) circle [radius = .05];
        }
    \]
    And the isomorphism $Z(A)\rightarrow \hom_{A|A}(A,A)$ is exactly by Equation \eqref{eq:central} and the converse direction is by $-\circ \iota$.
\end{proof}

By Lemma 4.1 of \cite{Qsys},
\begin{proposition}\label{prop:qsys reduction}
    Let $(A,m,\iota)$ be a simple Q-system and $p$ a projector in $(\cC(\one, A),\cdot,\iota,\star)$. Suppose $m(p\ot \id_A)$ splits as $r^\dagger r$ with $r:A\rightarrow A_r$, $m(\id_A\ot p)$ splits as $l^\dagger l$ with $l:A\rightarrow A_l$, and $m_{(2)}(p\ot\id_A\ot p)$ splits as $P^\dagger P$ with $P:A\rightarrow A_P$. Then
    \begin{itemize}
        \item $(A_P,m_P,\iota_P)$ is a Q-system in $\cC$, where
            \[
                m_P:=Pm(P^\dagger\ot P^\dagger)\times \sqrt{\frac{d_A}{|p^\dagger p|}},\quad \iota_P := P\iota\times \sqrt{\frac{|p^\dagger p|}{d_A}}.
            \]
        \item  $(A_l,\mu,\mu_l)$ is an $A$-$A_P$-bimodule, $(A_r,\nu_r,\nu)$ is an $A_P$-$A$-bimodule, where
            \[
                \mu:=lm(\id_A\ot l^\dagger),\quad \mu_r:=lm(l^\dagger\ot P^\dagger)\times \sqrt{\frac{d_A}{|p^\dagger p|}},\quad \nu_r:=rm(P^\dagger\ot r^\dagger)\times \sqrt{\frac{d_A}{|p^\dagger p|}},\quad \nu:= rm(r^\dagger\ot \id_A).
            \]
        \item $A$ and $A_P$ are Morita equivalent by $A_r$ and $A_l$.
    \end{itemize}
\end{proposition}

\section{More properties of Q-system models}
\subsection{Proofs of properties associated to fixed-point operators}\label{app:properties of fixed-point operators}
\begin{proposition}\label{prop:fpo eq bimod map}
    Let $(A,m,\iota)$ be a Q-system in $\Hilb$. The space of fixed-point operators $\fZ(A)$ is isomorphic to the space of $A$-$A$-bimodule operators $\End_{A|A}(A)$.
\end{proposition}
\begin{proof}
    Suppose $f\in\fZ(A)$ is a fixed-point operator, then $f$ is a left $A$-module map:
    \[
        \diagram{.5}{
            \draw (0,-.5) -- (0,2.5);
            \filldraw[fill = white] (-.3,.7) rectangle node{$f$} (.3,1.3);
            \draw (-1.5,.5) -- (0,2);
        } = 
        \diagram{.5}{
            \draw (0,-.5) -- (0,2.5);
            \filldraw[fill = white] (-.3,.7) rectangle node{$f$} (.3,1.3);
            \draw (-1.5,.5) -- (0,2);
            \draw (0,1.75) -- (-.75,1) -- (0,.25);
            \draw (0,1.75) -- (.75,1) -- (0,.25);
        } 
        =
        \diagram{.5}{
            \draw (0,-.5) -- (0,2.5);
            \filldraw[fill = white] (-.3,.7) rectangle node{$f$} (.3,1.3);
            \draw (-.5,-.5) -- (0,0);
            \draw (0,1.75) -- (-.75,1) -- (0,.25);
            \draw (0,1.75) -- (.75,1) -- (0,.25);
        } 
        =
        \diagram{.5}{
            \draw (0,-.5) -- (0,2.5);
            \filldraw[fill = white] (-.3,.7) rectangle node{$f$} (.3,1.3);
            \draw (-.5,-.5) -- (0,0);
        }.
    \]
    The right $A$-module condition can be proved similarly. Conversely, given $g\in\End_{A|A}(A)$, it is straightforward to check for the fixed-point equation.
\end{proof}

For any fixed-point operators $f,g\in \fZ(A)$, the commutativity of their composition $f\circ g = g \circ f$ can be proved by studying the following diagram:
\[
    \diagram{1}{
        \draw (0,0) -- (0,2);
        \filldraw[fill = white] (-.2, .3) rectangle node{$g$} (.2,.7);
        \filldraw[fill = white] (-.2, 1.3) rectangle node{$f$} (.2,1.7);
    } 
    =
    \diagram{1}{
        \draw (0,0) -- (0,.5) -- (.5,1) -- (0,1.5) -- (0,2);
        \draw (0,.5) -- (-.5,1) -- (0,1.5);
        \filldraw[fill = white] (-.7, .8) rectangle node{$g$} (-.3,1.2);
        \filldraw[fill = white] (.7, .8) rectangle node{$f$} (.3,1.2);
    } 
    =
    \diagram{1}{
        \draw (0,0) -- (0,2);
        \filldraw[fill = white] (-.2, .3) rectangle node{$f$} (.2,.7);
        \filldraw[fill = white] (-.2, 1.3) rectangle node{$g$} (.2,1.7);
    }
    =
    \diagram{1}{
        \draw (0,0) -- (0,.5) -- (.5,1) -- (0,1.5) -- (0,2);
        \draw (0,.5) -- (-.5,1) -- (0,1.5);
        \filldraw[fill = white] (-.7, .8) rectangle node{$f$} (-.3,1.2);
        \filldraw[fill = white] (.7, .8) rectangle node{$g$} (.3,1.2);
    }
\]
Such a circular commutativity does not hold for any local operators but only for those that are in the fixed point already.
Moreover, given any $f\in \fZ(A)$, its Hermitian conjugate $f^\dagger$ is also in $\fZ(A)$. Thus $\fZ(A)$ is a finite-dimensional commutative C$^\star $-algebra.

Since any finite-dimensional commutative C$^\star $-algebra is isomorphic to direct sums of $\C$, we decompose the algebra to $\fZ(A)\simeq \oplus_{\alpha}\C p_\alpha$, with the set $\{p_\alpha\}$ be the set of orthogonal projectors satisfying $p_\alpha p_\beta = \delta_{\alpha\beta}p_\alpha$ and $p_\alpha^\dagger = p_\alpha$. These projectors satisfy the property that they can freely proliferate on any $A$-mesh:
\[
    \diagram{1.3}{
        \draw (0,0) -- (0,2);
        \draw (1,0) -- (1,2);
        \draw (2,0) -- (2,2);
        \draw (3,0) -- (3,2);
        \draw (0,.25) -- (1,.75);
        \draw (1,1.25) -- (2,1.75);
        \draw (2,.25) -- (3,.75);
        \filldraw[fill = white] (1-.15, 1-.15) rectangle node{$p_{\alpha}$} (1+.15, 1+.15);
    }
    =
    \diagram{1.3}{
        \draw (0,0) -- (0,2);
        \draw (1,0) -- (1,2);
        \draw (2,0) -- (2,2);
        \draw (3,0) -- (3,2);
        \draw (0,.25) -- (1,.75);
        \draw (1,1.25) -- (2,1.75);
        \draw (2,.25) -- (3,.75);
        \filldraw[fill = white] (1-.15, 1-.15) rectangle node{$p_{\alpha}$} (1+.15, 1+.15);
        \filldraw[fill = white] (2-.15, 1-.15) rectangle node{$p_{\alpha}$} (2+.15, 1+.15);
        \filldraw[fill = white] (0-.15, 1-.15) rectangle node{$p_{\alpha}$} (0+.15, 1+.15);
        \filldraw[fill = white] (3-.15, 1-.15) rectangle node{$p_{\alpha}$} (3+.15, 1+.15);
        \filldraw[fill = white] (1-.15, .4-.15) rectangle node{$p_{\alpha}$} (1+.15, .4+.15);
        \filldraw[fill = white] (1-.15, 1.6-.15) rectangle node{$p_{\alpha}$} (1+.15, 1.6+.15);
        \filldraw[fill = white] (3-.15, .4-.15) rectangle node{$p_{\alpha}$} (3+.15, .4+.15);
        \filldraw[fill = white] (.5-.15, .5-.15) rectangle node{$p_{\alpha}$} (.5+.15, .5+.15);
        \filldraw[fill = white] (2.5-.15, .5-.15) rectangle node{$p_{\alpha}$} (2.5+.15, .5+.15);
        \filldraw[fill = white] (1.5-.15, 1.5-.15) rectangle node{$p_{\alpha}$} (1.5+.15, 1.5+.15);
    }
\]
\begin{corollary}\label{prop:lowest energy state}
    The $H$-type ground state space $\mathcal{G}_\Z$ of the Q-system model $(\Z, A, 1-m^\dagger m)$ is a simplex, whose extreme points are $\omega_{\alpha}$, satisfying $\omega_{\alpha}(m_{i,i+1}^\dagger m_{i,i+1}) = 1$ for any $i$ and $\omega_{\alpha}\left(p_{\beta}\right) = \delta_{\alpha\beta}$ for $p_\beta$ on any site.
\end{corollary}
\begin{proof}
     Denote $p_{\alpha,\{0\}}$ the projector on site $0$, let $\lambda_{\alpha} = \omega(p_{\alpha, \{0\}})\geq 0$, for any finite support operator $O\in\fA_{\Lambda}$ and $\omega$ a ground state,
    \begin{align*}
        \omega(O) &= \sum_{\alpha}c_{\alpha}\omega(p_{\alpha,\{0\}})\\
        &=\sum_{\alpha\beta}c_{\alpha}\lambda_{\beta}\omega_{\beta}(p_{\alpha,\{0\}})\\
        &= \sum_{\beta}\lambda_{\beta}\omega_\beta(O).
    \end{align*}
    Thus any ground state $\omega$ can be expressed as a finite convex combination of $\omega_\alpha$:
    \[
        \omega = \sum_{\alpha}\lambda_\alpha \omega_\alpha,\quad \lambda_\alpha\geq 0,\quad \sum_{\alpha} \lambda_\alpha = 1.
    \]
\end{proof}

Similar to the $G$-SPT case, the ground state of the $(\cC, f)$-symmetric phase can also be represented by a matrix product state (MPS) \cite{inamura202411dsptphasesfusion}. 
Suppose $(A,m,\iota)$ is a Q-system in $\cC_{\Hilb_f}^\vee$ and $\chargef A\simeq W\ot W^*$ is a matrix algebra in $\Hilb$. The ground state $\omega$ of the Q-system model $(\Z, \sH, A, 1-m^\dagger m)$ is an MPS, whose local tensor in a unit cell is 
\[
    \diagram{1}{
        \draw[string] (0,0) -- node[above right]{$A$} (0,1);
        \draw[string] (0,0) -- node[above]{$W$} (1.1,0);
        \draw[string] (-1.1,0) -- node[above]{$W$} (0,0);
        \filldraw[fill = white] (-.15,-.15) rectangle node {$T$} (.15,.15);
    }
     =\frac{1}{\sqrt{\dim W}}\times
    \diagram{1}{
        \draw[string] (.1,1) -- node[right]{$W$} (.1,0);
        \draw[string] (.1,0) -- (1.1,0);
        \draw[string] (-1.1,0) -- (-.1,0);
        \draw[string] (-.1,0) -- node[left]{$W$} (-.1,1);
    },
\]
and its map on a local operator is \cite{Fannes1992}
\[
    \omega(O) = \frac{1}{\dim W}\times 
    \diagram{0.7}{
        \foreach \i in {0,1,...,6}
        {
            \draw[string] (\i,0) -- (\i+1,0);
            \draw[string] (\i+1,0) -- (\i+1,1);
            \draw[string] (\i+1,1) -- (\i,1);
        }
        \draw[string] (0,0) -- (0,1);
        \filldraw[fill = white] (.75,.25) rectangle node{$O$} (6.25,.75);
        \foreach \j in {1,2,...,6}
        {
            \filldraw[fill = white] (\j-.1,-.1) rectangle (\j+.1,.1);
            \filldraw[fill = white] (\j-.1,.9) rectangle (\j+.1,1.1);
        }
    }
\]
It is straightforward to check the axiom of a quantum state in the operator algebra framework. We check it is the $H$-type ground state
\[
    \omega(m^\dagger m) = \frac{1}{\left(\dim W\right)^{4}}\times
    \diagram{.5}{
        \draw (.1,0) -- (.5,.4) -- (.9,0) -- cycle;
        \draw (.1,2) -- (.5,1.6) -- (.9,2) -- cycle;
        \draw (-1,0) -- (-.1,0) -- (.4,.5) -- (.4,1.5) -- (-.1,2) -- (-1,2) -- cycle;
        \draw (2,0) -- (1.1,0) -- (.6,.5) -- (.6,1.5) -- (1.1,2) -- (2,2) -- cycle;
    } = 1.
\]
For simplicity, we drop the arrow of strings here.

\subsection{Generic edge modes}\label{app:generic edge modes}
Suppose $(A,m_A,\iota_A), (B,m_B,\iota_B)$ are Q-systems in $\Hilb$ and $(M,\mu_A,\mu_B)$ a $A$-$B$ bimodule. Let
\[
    \Psi(\{i\}) = \begin{cases}
        A & i<0\\
        M & i = 0\\
        B & i>0
    \end{cases}
\]
and
\[
    \Phi(\{i,i+1\}) = \begin{cases}
        1- m_A^\dagger m_A & i<-1\\
        1- \mu_A^\dagger \mu_A & i = -1\\
        1- \mu_B^\dagger \mu_B & i = 0\\
        1- m_B^\dagger m_B & i > 0
    \end{cases}.
\]
The triple $(\Z, \Psi, \Phi)$ defines a 1D quantum system on an infinite chain describing the domain wall $M$ between the phase $A$ and $B$.

Similar to the bulk and edge case, the screening map
\begin{gather*}
    \scr:\fA_{\Lambda}\rightarrow \fA_{\{0\}},\\
    O\mapsto \scr(O):= m_\Xi Om_\Xi^\dagger,
\end{gather*}
can be defined for the following three cases of supporting regions of $O$:
\[
    \scr(O)=
    \diagram{1}{
    \draw (0,0) -- (0,2.5);
        \foreach \i in {1,...,7}
        {
            \draw[YaleBlue] (0,.5) -- (\i*.25, 1) -- (\i*.25,1.5) -- (0,2);
        }
        \filldraw[draw = YaleBlue,fill = white] (.5,1) rectangle node{$O$} (1.5,1.5);
    },\quad 
    \diagram{1}{
        \draw (0,0) -- (0,2.5);
        \foreach \i in {1,...,7}
        {
            \draw[YaleGreen] (0,.5) -- (-\i*.25, 1) -- (-\i*.25,1.5) -- (0,2);
        }
        \filldraw[draw = YaleGreen,fill = white] (-.5,1) rectangle node{$O$} (-1.5,1.5);
    }\quad \text{or}\quad 
    \diagram{1}{
        \draw (0,0) -- (0,2.5);
        \foreach \i in {1,2,3}
        {
            \draw[YaleBlue] (0,.5) -- (\i*.25, 1) -- (\i*.25,1.5) -- (0,2);
        }
        \foreach \i in {1,2,3}
        {
            \draw[YaleGreen] (0,.5) -- (-\i*.25, 1) -- (-\i*.25,1.5) -- (0,2);
        }
        \filldraw[fill = white] (-.5,1) rectangle node{$O$} (.5,1.5);
    }.
\]
The map $\scr$ is an idempotent and its image $\im(\scr) =: \fZ_{A|B}(M)$ is called the algebra of fixed-point operators on the domain wall. By the property of the Q-system model, we have
\[
    \omega = \omega \circ \scr.
\]
Thus the $H$-type ground state space is isomorphic to the space of states on $\fZ_{A|B}(M)=\End_{A|B}(M)$. And $\fZ(A) = \fZ_{A|A}(A)$, $\fZ_{A}(M) = \fZ_{\one|A}(M)$ as special cases. We thus have:
\begin{mdframed}
    The zero mode of $(\Z, \Psi,\Phi)$ is the minimal faithful module over the algebra of fixed-point operators on the domain wall $\fZ_{A|B}(M)$.
\end{mdframed}

\subsection{Ground state degeneracy by taking the thermodynamic limit}
We now calculate the ground state degeneracy of the Q-system model by taking the thermodynamic limit. 

Let $\{\cL_L\}_{L=4}^\infty$ be a sequence of quantum systems, where each $\cL_L := (\{1, \cdots, L\}, \Psi, \Phi)$ represents a quantum system defined on the lattice $\{1, \cdots, L\}$. This sequence corresponds to the process of taking the thermodynamic limit.
In general, both the bulk and the edge contribute to the ground state degeneracy of each finite chain $\cL_L$. Our goal is to distinguish these contributions.

Both the bulk and the edge include regions containing a macroscopic number of sites, growing proportionally with the system size as the system itself grows. This perspective is what we refer to as ``macroscopically local'', and we provide a more detailed characterization below.

A macroscopic local region in one dimension is a sequence of local regions $\{\Lambda_{L}\}_{L=4}^\infty$, where $\Lambda_L\subset \{1,\cdots, L\}$ such that the following limits exist
\[
    \lim_{L\rightarrow \infty}\min(\Lambda_L)/L := d_l,\quad \lim_{L\rightarrow \infty}|\Lambda_L|/L  := |\Lambda| > 0.
\]
We also define 
\[d_r := 1 - d_l - |\Lambda|.\]
The $d_l$ ($d_r$) characterizes the left (right) distance between the left (right) boundary of the region and that of the system and $|\Lambda|$ characterizes the length of the region.
The macroscopic local region $\{\Lambda_{L}\}_{L=4}^\infty$ is
\begin{itemize}
    \item a bulk region if $0<d_l<1$ and $0<d_r<1$;
    \item a left boundary region if $\min(\Lambda_L) = 0$ for any $L$ and $0<d_r<1$;
    \item a right boundary region if $\max(\Lambda_L) = L-1$ for any $L$ and $0<d_l<1$.
\end{itemize}
\begin{figure}
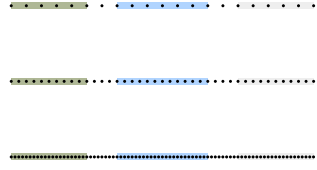

    \centering
    \diagram{2}{
        \fill[YaleLightBlue] (.7,-.02) rectangle (1.3,.02);
        \fill[YaleLightGreen] (0,-.02) rectangle (0.5,.02);
        \fill[YaleLightGrey] (1.5,-.02) rectangle (2,.02);
        \fill[YaleLightBlue] (.7,-.02-.5) rectangle (1.3,.02-.5);
        \fill[YaleLightGreen] (0,-.02-.5) rectangle (0.5,.02-.5);
        \fill[YaleLightGrey] (1.5,-.02-.5) rectangle (2,.02-.5);
        \fill[YaleLightBlue] (.7,-.02-1) rectangle (1.3,.02-1);
        \fill[YaleLightGreen] (0,-.02-1) rectangle (0.5,.02-1);
        \fill[YaleLightGrey] (1.5,-.02-1) rectangle (2,.02-1);
        \foreach \i in {0,1,2,...,20}
        {
            \fill (\i * .1,0) circle [radius = .01];
        }
        \foreach \i in {0,1,2,...,40}
        {
            \fill (\i * .05,-.5) circle [radius = .01];
        }
        \foreach \i in {0,1,2,...,80}
        {
            \fill (\i * .025,-1) circle [radius = .01];
        }   
    }
    \caption{Illustration of the macroscopic local regions. Here the green region represents the left boundary, the blue region represents the bulk and the grey region represents the right boundary.}
    \label{fig:enter-label}
\end{figure}

To calculate the bulk and boundary ground state degeneracies separately, we adopt the principle that \textit{the ground state degeneracy in a macroscopically local region is the effective degeneracy observed by its local observables}. This principle is inspired by the philosophical exploration of the principle that \textit{a phase is defined on an open disk} presented in \cite{kong2022invitationtopologicalorderscategory} and the entanglement bootstrap framework discussed in \cite{Shi_2020,Shi_2021}.

The observable within a macroscopic local region is a sequence of C$^\star $-algebras $\{\fA_{\Lambda_L}\}_{L = 4}^{\infty}$\footnote{An ideal formulation of the thermodynamic limit requires specifying embedding maps between the C$^\star $-algebras for different $L$ values, so that models with varying numbers of sites can be meaningfully compared. The total algebra in the thermodynamic limit is then the inductive limit of this inductive net.
}, with the natural embedding of algebras $\fA_{\Lambda_L}\hookrightarrow\fA_{\{1,\cdots, L\}}$ for each $L$. Any states on $\fA_{\{1,\cdots,L\}}$ can be restricted to $\fA_{\Lambda_L}$. In particular, the restriction of the global ground states
\footnote{Note that for a non-fixed-point model, the ground state is generically not exactly degenerate. Still, for a fixed $L$ there exist energy gaps $E_{i, L}$ between the ground state and several excitation states labeled by $i$, such that $\lim_{L\rightarrow \infty}E_{i,L} = 0$ for all $i$. In this case, an error is needed to include the energy states that are sufficiently low as the ground state subspace when taking the thermodynamic limit \cite{kitaev2024almostidempotentquantumchannelsapproximate}. We will not deal with this case here.} 
defines an effective space of ground states on local regions \cite{nachtergaele2010liebrobinsonboundsquantummanybody}
\begin{gather*}
    \cG_{\Lambda_L} :=  \left\{\omega|_{\fA_{\Lambda_L}}: \omega\in\cG_{\{1,\cdots, L\}}\right\}.
\end{gather*}
This restriction may not be injective; there may exist two distinct ground states $\omega \neq \omega'$ on $\fA_{\{1, \cdots, L\}}$ that become identical on $\fA_{\Lambda_L}$ after restriction:
\[
    \omega|_{\fA_{\Lambda_L}} = \omega'|_{\fA_{\Lambda_L}}.
\]
If the geometric property of the space $\cG_{\Lambda_L}$ stabilizes as $L\to +\infty$, then the ground state degeneracy and the associated superselection sectors can be inferred from this stabilized geometry.

The argument above may be too vague. Now we consider the Q-system model as a specific example. Since it is a fixed-point model, this case is clean and almost trivial. 
Let $(A,m,\iota)$ be a Q-system in $\cC_{\Hilb_f}^\vee$ and $(M,m_M)$, $(N,{}_{N}m)$ are standard $A$-modules. Consider a sequence $\left\{\mathcal{L}_{L}\right\}_{L = 4}^{\infty}$, where $\cL_{L} = (\{1,2,\cdots, L\}, \Psi_L,\Phi_L)$ is the Q-system model on a finite chain
with left boundary condition $M$, right boundary condition $N$ and the bulk $A$.
We know that the ground state subspace of $\cL_L$ is isomorphic to $\chargef (M\ot[A]N)$ for all $L$. 

We denote the map
\[
    m_{\{1,\cdots, L\}} := \bigg(M\ot A^{\otimes(L-2)}\ot N\rightarrow M\ot A\ot N\rightrightarrows \\
    M\ot N\rightarrow M\ot[A]N\bigg).
\]
Given a fixed $L$, ground states of $\cL_L$ are equivalent to density matrices over $M\ot[A]N$. 
Suppose $\rho\in\End(M\ot[A]N)$ is a density matrix, the map
\[
    \rho\mapsto m_{\{1,\cdots, L\}}^\dagger\rho m_{\{1,\cdots, L\}}
\]
carries a density matrix on $M\ot[A]N$ to a density matrix on $M\ot A^{\otimes(L-2)}\ot N$. And the ground state valued on $O \in \fA_{\Lambda_L}$ is
\[
    \omega(O) = \Tr\bigg(m_{\{1,\cdots, L\}}^\dagger\rho m_{\{1,\cdots, L\}}O\bigg) \\= \Tr\bigg(\rho m_{\{1,\cdots, L\}}Om_{\{1,\cdots, L\}}^\dagger\bigg).
\]

\begin{itemize}
    \item When $\{\Lambda_L\}$ is a bulk region, there exists a number $x > 0$ and $N\in \N$, such that $\min(\Lambda_L)> Lx$, and $\max(\Lambda_L)<L(1-x)$ for all $L>N$. Thus for sufficiently large $L$, the region $\Lambda_L$ will be surrounded by $A$'s. Thus
    \[
        \omega(O) = \Tr(\rho m_{\{1,\cdots, L\}}Om_{\{1,\cdots, L\}}^\dagger) \\= \Tr\left(\rho m_{\{1,\cdots, L\}}\scr(O)m_{\{1,\cdots, L\}}^\dagger\right)
    \]
    Suppose $\scr(O) = \sum_{\alpha}\lambda_\alpha p_\alpha$ is the decomposition of bimodule maps, the equation above can be reduced to
    \[
        \omega(O) =\omega\left(\scr(O)\right) =\sum_{\alpha}\lambda_\alpha\omega(p_\alpha),\quad \forall L > N.
    \]
    Thus the value of $\omega$ on local operators in $\Lambda_L$ is completely fixed by $\omega(p_\alpha)$, implying the state space $\cG_{\Lambda_L}$ is a simplex for all $L>N$. We thus conclude that the ground state space on the macroscopic local region is a simplex, whose dimension is $\dim(\fZ(\chargef A)) - 1$. Thus the dimension of $\fZ(\chargef A)$ is the ground state subspace in the bulk, and whose extreme points are superselection sectors;
    \item When $\{\Lambda_L\}$ is a left boundary region, there exists a number $x>0$ and $N\in N$, such that $\max(\Lambda_L)<L(1-x)$ for all $L>N$. Thus for sufficiently large $L$, the region $\Lambda_L$ will be surrounded by $A$'s on the right-hand side. By the Equation \eqref{eq:fpobdy}
    \[
        \omega(O) = \omega\left(\scr(O)\right),\quad \forall L > N.
    \]
    Since $\scr(O) \in \fZ_{\chargef A}(\chargef M)$, the value of $\omega$ on local operators on $\Lambda_L$ is totally given by its value on operators in $\fZ_{\chargef A}(\chargef M)$. Thus, the edge mode is the minimal faithful module over $\fZ_{\chargef A}(\chargef M)$.
    \item The right boundary region case is the same.
\end{itemize}
These results are consistent with the infinite chain calculations in the last section.

\subsection{Deformation}\label{section: deformation}
In this section, we show that a Q-system model $(\Z,\sH,A,-m^\dagger m)$ can be deformed into another smaller one $(\Z,\sH,A_P,-m_P^\dagger m_P)$ constructed from the reduced Q-system $(A_P,m_P,\iota_P)$ without causing phase transition.

Let $(A,m,\iota)$ be a simple Q-system in $\cC_{\Hilb_f}^\vee$. Define the parameterized multiplication:
\[
    m(f) := 
    \begin{tikzpicture}[baseline = (current bounding box), scale = .7]
        \draw (-1,-1) -- (0,0) -- (1,-1);
        \draw (0,0) -- (0,1);
        \draw (0,-.5) -- (-.25,-.25);
        \filldraw[fill = white,draw = black] (0-.2,-.5-.4) rectangle node{$f$} (+.2,-.5);
    \end{tikzpicture}.
\]
where $f\in\cC_{\Hilb_f}^\vee(\one, A)$.
Since $A$ is simple, the following map only gives a constant factor.
\[
    m(f) m(f)^\dagger = \lambda_f \id_A,
\]
Here we denote
\[
    f = 
    \diagram{1}{
        \draw (0,0) -- (0,.5);
        \fill[black] (-.1,-.1) rectangle (.1,.1);
    },\quad f^\dagger = 
    \diagram{1}{
        \draw (0,0) -- (0,.5);
        \fill[black] (-.1,.4) rectangle (.1,.6);
    }.
\]
Thus $\lambda_f$ can be calculated via
\[
    \lambda_f d_A = 
    \diagram{.5}{
        \draw (0,-1) -- (1,0) -- (0,1) -- (-1,0) -- cycle;
        \draw (0,-1.5) -- (0,-1);
        \draw (0,1) -- (0,1.5);
        \draw (0,.5) -- (-.25,.75);
        \draw (0,-.5) -- (-.25,-.75);
        \fill[black] (0-.1,.5-.1) rectangle (0+.1,.5+.1);
        \fill[black] (0-.1,-.5+.1) rectangle (0+.1,-.5-.1);
        \filldraw[fill = white] (0,1.5) circle [radius = .05];
        \filldraw[fill = white] (0,-1.5) circle [radius = .05];
    }
    =
    \diagram{.5}{
        \draw (0,-1) -- (0,1);
        \draw (0,.75) -- (.25,.5);
        \draw (0,-.75) -- (.25,-.5);
        \fill[black] (.25-.1,.5-.1) rectangle (.25+.1,.5+.1);
        \fill[black] (.25-.1,-.5+.1) rectangle (.25+.1,-.5-.1);
        \filldraw[fill = white] (0,1) circle [radius = .05];
        \filldraw[fill = white] (0,-1) circle [radius = .05];
    }
    =
    |f^\dagger f|.
\]
In the second and third steps, we used the spherical property and the pivotal property of the standard dual.
Thus $\lambda_A = d_A^{-1}|f^\dagger f| \geq 0$ and $\lambda_f = 0$ if and only if $f = 0$.

The quantum system
\[
    \cL_f:=\left(\Z,\sH,A,-m(f)^\dagger m(f)\right)
\]
is a gapped commuting model for each $f\in\cC(\one,A)\setminus \{\mathbf{0}\}$. By
\[
    (-m(f)^\dagger m(f))^2 = \frac{|f^\dagger f|}{d_A}m(f)^\dagger m(f),
\]
the energy gap is $\frac{|f^\dagger f|}{d_A}>0$. 
Thus a continuous class of the non-isolated fixed-point models can be realized whenever $A$ is simple and reducible. The phenomenon of the non-isolated fixed-point first appeared in \cite{levin_ppt} by the corner-double-line tensor and further been seriously treated in \cite{Gu_2009} by the so-called \textit{entanglement filtering} procedure. Our local interaction here can be viewed as a type of the corner-double-line tensor in Hamiltonian mechanics and $f$ is the the tensor in the corner.

Given a projector $p$ in $\cC_{\Hilb_f}^\vee(\one, A)$ (definition \ref{def:unit alg}), i.e. 
\[
    p = \ 
    \diagram{1}{
        \draw (0,0) -- (0,.5);
        \fill[black] (0,0) circle[radius = .05];
    },\quad 
    p^\dagger = \
    \diagram{1}{
        \draw (0,0) -- (0,.5);
        \fill[black] (0,.5) circle[radius = .05];
    },
\]
with
\[
    \diagram{1}{
        \draw (-.5,-.5) -- (0,0) -- (0,.5);
        \draw (.5,-.5) -- (0,0);
        \fill[black] (-.5,-.5) circle[radius = .05];
        \fill[black] (.5,-.5) circle[radius = .05];
    }
    \ =\ 
    \diagram{1}{
        \draw (0,-.5) -- (0,.5);
        \fill[black] (0,-.5) circle[radius = .05];
    }\ = \ 
    \diagram{1}{
        \draw (0,0) -- (.2,-.2) -- (.4,0) -- (.4,.6);
        \draw (.2,-.4) -- (.2,-.2);
        \filldraw[fill = white] (.2,-.4) circle[radius = .05];
        \fill[black] (0,0) circle[radius = .05];
    }\ 
    =\ 
    \diagram{1}{
        \draw (0,0) -- (-.2,-.2) -- (-.4,0) -- (-.4,.6);
        \draw (-.2,-.4) -- (-.2,-.2);
        \filldraw[fill = white] (-.2,-.4) circle[radius = .05];
        \fill[black] (0,0) circle[radius = .05];
    }.
\]
It is always possible to choose a path $\gamma:[0,1]\rightarrow \cC_{\Hilb_f}^\vee(\one,A)$, such that $\gamma(0) = \iota$ and $\gamma(1) = p$.
Then $\cL\circ\gamma$ is a path of the commuting models, with $\cL\circ \gamma(0) = \cL_{\iota} = \left(\Z,\sH,A,-m^\dagger m\right)$ be the original model and $\cL\circ \gamma(1) = \cL_p = (\Z,\sH, A, -m(p)^\dagger m(p))$ be the deformed model. 

Suppose $p$ reduces the Q-system to $(A_P,m_P,\iota_P)$ (proposition \ref{prop:qsys reduction}), we show that the $H$-type ground state of $\cL_{p}$ is one-to-one correspondence to that of $(\Z,\sH,A_P, -m_P^\dagger m_P)$. Suppose $\omega$ 
 is an $H$-type ground state for $(\Z,\sH, A, 1-m(p)^\dagger m(p))$, then
\begin{align*}
    \frac{|p^\dagger p|}{d_A} &=\omega(m^\dagger_{i,i+1}(p)m_{i,i+1}(p))\\
    &= \omega\bigg(m^\dagger_{i-1,i}(p)m_{i-1,i}(p) m^\dagger_{i,i+1}(p)m_{i,i+1}(p)
    \times m^\dagger_{i+1,i+2}(p)m_{i+1,i+2}(p)\bigg) \\
    &= \omega\bigg(m^\dagger_{i-1,i}(p)m_{i-1,i}(p) \frac{|p^\dagger p|}{d_A}m^\dagger_{P,i,i+1}m_{P,i,i+1} m^\dagger_{i+1,i+2}(p)m_{i+1,i+2}(p)\bigg)\\
    &=\frac{|p^\dagger p|}{d_A}\omega\left(m^\dagger_{P,i,i+1}m_{P,i,i+1}\right)
\end{align*}
for any $i\in\Z$.
Conversely, if $\omega$ is an $H$-type ground state for $(\Z,\sH,A_P, 1-m_P^\dagger m_P)$, we have
\begin{align*}
    1 &=\omega(m^\dagger_{P,i,i+1}m_{P,i,i+1})\\
    &= \omega\bigg(m^\dagger_{P,i-1,i}m_{P,i-1,i}m^\dagger_{P,i,i+1}m_{P,i,i+1}m^\dagger_{P,i+1,i+2}m_{P,i+1,i+2}\bigg) \\
    &= \frac{d_A}{|p^\dagger p|}\omega\bigg(m^\dagger_{P,i-1,i}m_{P,i-1,i}m^\dagger_{i,i+1}(p)m_{i,i+1}(p)m^\dagger_{P,i+1,i+2}m_{P,i+1,i+2}\bigg)\\
    &=\frac{d_A}{|p^\dagger p|}\omega\bigg(m^\dagger_{i,i+1}(p)m_{i,i+1}(p)\bigg)
\end{align*}
for any $i\in\Z$. Combined with Remark \ref{remark:matrix algebra}, any Morita equivalent Q-systems realize the same phase.

\section{Bimodule category \texorpdfstring{${}_{A_+}(\Hilb_{S_3\xt \Z_3})_{A_+}$}{A+(HilbS3xZ3)A+}}\label{sect:bimodS3xZ3}
First, we can show a result that will be used later. For a finite group $G$ and a 2-cocycle $\psi\in Z^2(G,U(1))$ with $\psi(e,g) = \psi(g,e) = 1$ for any $g\in G$, 
    \[
    \psi(g,h)\psi(h^{-1},g^{-1}) = \frac{\psi(g,g^{-1})\psi(h,h^{-1})}{\psi(gh,h^{-1}h^{-1})}.
    \]
    Thus $\psi(g,h)\psi(h^{-1},g^{-1})$ is cohomologously trivial.



We fix the presentation $S_3 = \la  s, r|s^2 = r^3 = e,\ srs = r^2\ra$ and $\Z_3 = \la a|a^3 = e\ra$. We consider $A_+ = \C^{\psi_+}\la r,a\ra \in \Hilb_{S_3\xt \Z_3}$, where $\psi_+(g,h) = \ee^{\frac{2\pi \ii }{3}g_2h_1}$. Following \cite{ostrik2006modulecategoriesdrinfelddouble}, we compute the bimodule category ${}_{A_+}(\Hilb_{S_3\xt \Z_3})_{A_+}$. Below, we use $1_g$ to represent a basis of $\C_g$, such that $1_g\ot 1_h = 1_{gh}$. And the multiplication of $A_+$ is defined on basis as
\[
1_g\cdot 1_h = \psi_+(g,h)1_{gh}.
\]

\subsection{As a linear category}
The double cosets $\la r,a \ra\setminus (S_3\xt \Z_3)/\la r,a \ra$ are
\[
\la r,a\ra, \quad s\la r,a\ra.
\]

The groups and cocycles defined in \cite{ostrik2006modulecategoriesdrinfelddouble} are:
\begin{itemize}
    \item $\la r,a\ra^e = \la r,a\ra$:
    \[
    \psi^e(c,d) = \psi_+(c,d)\psi_+(d^{-1},c^{-1})\sim_{\mathrm{coh}} 1
    \]
    which is a trivial 2-cocycle following from the result discussed above. The projective representations of $\la r,a\ra$ with projective phase $\psi^e$ are 1-dimensional. Thus
    \[
    \left|\Irr\left(\Rep^{\psi^e}\la r,a\ra^e\right)\right| = 9.
    \]
     We choose the $1$-dimensional projective representations $\rho_{ij}:\Z_3\xt\Z_3\rightarrow \mathrm{PGL}(U_{ij})$ of $\la r,a\ra^e$ with projective phase $\psi^e$:
\[
\rho_{ij}(r^m) = \ee^{\frac{2\pi\ii }{3}im},\quad \rho_{ij}(a^n) = \ee^{\frac{2\pi\ii }{3}jn},\quad i,\, j,\, m,\, n = 0,\,1,\, 2.
\]
And we define
\[
\rho_{ij}(r^ma^n):=\ee^{-\frac{2\pi\ii mn }{3}}\rho_{ij}(a^m)\rho_{ij}(r^n) = \ee^{\frac{2\pi\ii}{3}(im+jn-mn)}.
\]
We denote the simple $A_+$-bimodule related to the projective representation $(U_{ij},\rho_{ij})$ as $[ij]$.
    \item $\la r,a\ra^s= \la r,a\ra$:
    \begin{align*}
      \psi^s(c,d) &= \psi_+(c,d)\psi_+(s^{-1}d^{-1}s,s^{-1}c^{-1}s) \\
      &= \ee^{\frac{2\pi\ii}{3}(c_2d_1-c_1d_2)} \\
      &= \ee^{-\frac{2\pi\ii}{3}c_2d_1}\frac{\ee^{\frac{2\pi\ii}{3}c_1c_2}\ee^{\frac{2\pi\ii}{3}d_1d_2}}{\ee^{\frac{2\pi\ii}{3}(c_1+d_1)(c_2+d_2)}} \\
      &\sim_{\mathrm{coh}} \ee^{-\frac{2\pi\ii}{3}c_2d_1}
\end{align*}
which is a non-trivial 2-cocycle. The representation of $\C^{\psi^s}\la r,a\ra^s$ is $3$-dimensional. Thus
\[
\left|\Irr\left(\Rep^{\psi^s}\la r,a\ra^s\right)\right| = 1
\]
The $3$-dimensional projective representation $\rho_V:\Z_3\xt \Z_3\rightarrow \mathrm{PGL}(V)$ of $\la r,a\ra^s$ can be chosen as
\[
\rho_V(r^m) = \begin{pmatrix}
    \ee^{-2\pi\ii/3} & 0 & 0\\
    0 & 1 & 0\\
    0 & 0 & \ee^{2\pi\ii/3}
\end{pmatrix}^m,\quad 
\rho_V(a^n) = \begin{pmatrix}
    0 & 0 & 1\\
    1 & 0 & 0\\
    0 & 1 & 0
\end{pmatrix}^n.
\]
We then have
\[
\rho_V(r^ma^n) :=\ee^{\frac{2\pi\ii}{3}mn}\rho_V(r^m)\rho_V(a^n),
\]
\[
\rho_V(a^n)\rho_V(r^m) = \ee^{\frac{2\pi\ii}{3}mn}\rho_V(r^ma^n),
\]
\end{itemize}

Denote the simple bimodule related to the projective representation $(V,\rho_V)$ as $\sigma$. Since $\sigma$ is the unique non-invertible object in the bimodule category with the quantum dimension $d_\sigma = 3$, the fusion rule is Tambara-Yamagami-type. There are several groups of order $18$:
\[
D_{18},\quad \Z_{9}\times \Z_2,\quad S_3\times \Z_3,\quad (\Z_3\times \Z_3)\rtimes_{\varphi}\Z_2,\quad \Z_3\times \Z_3\times \Z_2,
\]
Here $\varphi_1(a,b) = (-a, -b)$.
None of these groups have irreducible representation of dimension-$3$. Thus the bimodule category here does not correspond to any group representation category.

From the above setup, the bimodule category is isomorphic to $\Rep(\sH)$ as a linear category, where $\sH$ is an algebra
\[
\sH = \C^{\psi^e}\la r,a\ra\oplus \C^{\psi^s}\la r,a\ra.
\]
We denote the basis in $\C^{\psi^e}\la r,a\ra$ as $\{g_e\}_{g\in\la r,a\ra}$, and in $\C^{\psi^s}\la r,a\ra$ as $\{g_s\}_{g\in\la r,a\ra}$.

\subsection{As a fusion category}
The Tambara-Yamagami category $\TY_{G}^{\chi,\epsilon}$ is defined by
\begin{itemize}
    \item an abelian group $G$ and a symmetric non-degenerate bi-character $\chi:G\xt G\rightarrow U(1)$;
    \item the Frobenius-Schur indicator $\epsilon = \pm 1$.
\end{itemize}
Below we will determine $\chi$ and $\epsilon$.

We first fix the left $A$-action ``$\rt$'' and the right $A$-action ``$\lt$''.

\begin{itemize}
    \item $[ij]$: The invertible objects are $[ij] = U_{ij}\otimes A$, the right $A$-module action on which is given by the free $A$-multiplication, and
\[
1_{g}\rt (u_{ij}\ot 1_e)\lt 1_{g^{-1}} = \rho_{ij}(g)u_{ij}\ot 1_e,\quad \forall u_{ij}\in U_{ij}.
\]

\item $\sigma$: the non-invertible bimodule is $\sigma = V\ot s\ot A$, the right $A$-module action is given by the free $A$-multiplication and
\[
1_g\rt (v\ot s\ot 1_e)\lt 1_{s^{-1}gs} = \rho_V(g)(v)\ot s\ot 1_e,\quad \forall v\in V.
\]
\end{itemize}
 
The relative tensor product is defined as
\[
U\ot[A]V := U\ot V/\la u\lt a\ot v- u\ot a\rt v \ra.
\]
Thus 
\[
u\lt a\ot[A]v = u\ot[A]a\rt v.
\]

\begin{itemize}
    \item $[ij]\ot[A][kl]$: The adjoint action on $u_{ij}\ot 1_e\ot[A] u_{kl}\ot 1_e$ is
    \[
     1_g\rt (u_{ij}\ot 1_e)\ot[A]( u_{kl}\ot 1_e)\lt 1_{g^{-1}} = \psi_+(g^{-1},g)^{-1} (\rho_{ij}(g)u_{ij}\ot 1_e)\ot[A]( \rho_{kl}(g)u_{kl}\ot 1_e).
    \]
    Thus the effective action on the vector space $U_{ij}\ot U_{kl}$ is encoded by the map
\[
\rho_{U_{ij}\ot U_{kl}}(g) := \psi_+(g^{-1},g)^{-1} \rho_{ij}(g)\ot \rho_{kl}(g).
\]
\item $[ij]\ot[A]\sigma$: The adjoint action is
\begin{align*}
    &\quad 1_g\rt (u_{ij}\ot 1_e)\ot[A] (v\ot 1_s)\lt 1_{s^{-1}g^{-1}s}\\
    &= \psi_+(g^{-1},g)^{-1}1_g\rt (u_{ij}\ot 1_e)\lt 1_{g^{-1}}\ot[A]1_g\rt (v\ot 1_s)\lt 1_{s^{-1}g^{-1}s}\\
    &=\psi_+(g^{-1},g)^{-1} \bigg(\rho(g)(u_{ij})\ot 1_e\bigg)\ot[A] \bigg(\rho_V(g)(v)\ot 1_s\bigg).
\end{align*}
The effective action on the vector space $U_{ij}\ot V$ is encoded by the map
\[
\rho_{U_{ij}\ot V}(g) := \psi_+(g^{-1},g)^{-1} \rho_{ij}(g)\ot \rho_{V}(g).
\]
    \item $\sigma\ot[A][ij]$: 
    The adjoint action on $v\ot 1_s \ot[A] u_{ij}\ot 1_{e}\in\sigma\ot[A][ij]$ is
\begin{align*}
&\quad 1_g\rt (v\ot 1_s)\ot[A]( u_{ij}\ot 1_e)\lt 1_{s^{-1}g^{-1}s}\\
    &=\psi_+(s^{-1}g^{-1}s,s^{-1}gs)^{-1}1_g\rt (v\ot 1_s)\lt 1_{s^{-1}g^{-1}s}\ot[A]1_{s^{-1}gs}( u_{ij}\ot 1_e)\lt 1_{s^{-1}g^{-1}s}\\
    &=\psi_+(s^{-1}g^{-1}s,s^{-1}gs)^{-1} \bigg(\rho_V(g)(v)\ot 1_s\bigg)\ot[A] \bigg(\rho_{ij}(s^{-1}gs)(u_{ij})\ot 1_e \bigg).
\end{align*}
We find that the effective action on the vector space $V\ot U_{ij}$ is encoded by the map
\[
\rho_{V\ot U_{ij}}(g) := \psi_+(s^{-1}g^{-1}s,s^{-1}gs)^{-1} \rho_{V}(g)\ot \rho_{ij}(s^{-1}gs).
\]
\item $\sigma\ot \sigma$: The adjoint action is
\begin{align*}
    &\quad 1_g\rt \bigg(u\ot 1_s\bigg)\ot[A]\bigg(v\ot 1_s\bigg)\lt 1_{g^{-1}}\\
    &=\psi_+(s^{-1}g^{-1}s,s^{-1}gs)^{-1} 1_g\rt \bigg(u\ot 1_s\bigg)\lt 1_{s^{-1}g^{-1}s}\ot[A]1_{s^{-1}gs}\rt\bigg(v\ot 1_s\bigg)\lt 1_{g^{-1}}\\
    &=\psi_+(s^{-1}g^{-1}s,s^{-1}gs)^{-1} \bigg(\rho_V(g)(u)\ot 1_s\bigg)\ot[A]\bigg(\rho_{V}(s^{-1}gs)(v)\ot 1_s\bigg).
\end{align*}
The effective action on the vector space $V\ot V$ is encoded by the map
\[
\rho_{V\ot V}(g) := \psi_+(s^{-1}g^{-1}s,s^{-1}gs)^{-1}\rho_V(g)\ot \rho_V(s^{-1}gs).
\]
\end{itemize}

With the conjugation action on the relative tensor product space, we now fix the intertwiner of the fusion.
\begin{itemize}
    \item $[(i+k)(j+l)]\rightarrow [ij]\ot[] [kl]$: 
    \begin{align*}
         &\quad \psi_+((r^ma^n)^{-1},r^ma^n)^{-1}\rho_{ij}(r^ma^n)\rho_{kl}(r^ma^n)\\
         &= \ee^{\frac{2\pi\ii}{3}mn}\ee^{\frac{2\pi\ii}{3}(im+jn-mn)}\ee^{\frac{2\pi\ii}{3}(km+ln-mn)}\\
         &= \ee^{\frac{2\pi\ii}{3}((i+k)m+(j+l)n-mn)}\\
         &= \rho_{i+k,j+l}(r^ma^n).
    \end{align*}
     Thus the abelian group $G$ of the Tambara-Yamagami category is $\Z_3\xt\Z_3$;
     \item $\sigma\rightarrow [ij]\ot \sigma$: the bimodule map condition implies that we only need to find the intertwiner $f_{ij}:V\rightarrow U_{ij}\ot V$, that is, a linear map $f_{ij}$ such that 
     \[
     \rho_{U_{ij}\ot V}(g) f_{ij} = f_{ij}\rho_{V}(g),\quad \forall g\in \la r,a\ra.
     \]
     We find
     \[
     f_{10} = \begin{pmatrix}
        0 & 1 & 0\\
        0 & 0 & 1\\
        1 & 0 & 0
    \end{pmatrix},\quad 
    f_{01} =
    \begin{pmatrix}
        \ee^{-2\pi\ii/3} & 0 & 0\\
        0 & 1 & 0\\
        0 & 0 & \ee^{2\pi\ii/3}
    \end{pmatrix}.
     \]
     \item $\sigma\rightarrow \sigma\ot[][ij]$: we need to determine the intertwiner $k_{ij}: V\rightarrow V\ot U_{ij}$
     \[
     \rho_{V\ot U_{ij}}(g)f_{ij} = f_{ij}\rho_{V}(g),\quad \forall g\in \la r,a\ra.
     \]
     We find
     \[
k_{10}  = 
\begin{pmatrix}
    0 & 0 & 1\\
    1 & 0 & 0\\
    0 & 1 & 0
\end{pmatrix},\quad 
k_{01} = 
\begin{pmatrix}
    \ee^{-2\pi\ii/3} & 0 & 0\\
    0 & 1 & 0\\
    0 & 0 & \ee^{2\pi\ii/3}
\end{pmatrix}.
\]
\item $[ij]\rightarrow \sigma\ot \sigma$:
we find the intertwiner $m_{ij}:U_{ij}\rightarrow V\ot V$, such that
\[
\rho_{V\ot V}(g)m_{ij} =  m_{ij}\rho_{ij}(g).
\]
\end{itemize}

Now we compute the bi-character with the intertwiners. First, it is easy to see that
\[
\chi(r,r) = \chi(a,a) = 1.
\]
Since
\[
f_{10}k_{01} = \ee^{\frac{2\pi\ii}{3}}k_{01}f_{10},
\]
we have
\[
\chi(r,a) = \chi(a,r) = \ee^{\frac{2\pi\ii}{3}}.
\]
Thus, the bicharacter is
\[
\chi(g,h) = \ee^{\frac{2\pi\ii}{3}(g_1h_2+g_2h_1)}.
\]

And since the $m_{00}$ coincides with the coevaluation map, the Frobenius-Schur indicator $\epsilon = 1$.

We have the result:
\[
{}_{A_+}(\Hilb_{S_3\xt \Z_3})_{A_+}\cong \TY_{\Z_3\xt\Z_3}^{\chi,\epsilon},\quad \chi(g,h) = \exp{\left[\frac{2\pi\ii}{3}\left(g_1h_2+g_2h_1\right)\right]},\quad \epsilon = 1.
\]

The fusion structure gives $\sH$ a Hopf algebra structure, which can be directly read off from above
\begin{align*}
    \one &= e_e+e_s,\\
    \Delta g_e &= \psi_+(g^{-1},g)^{-1} g_e\ot g_e + \psi_+(s^{-1}g^{-1}s,s^{-1}gs)^{-1}g_s\ot s^{-1}g_s s,\\
    \Delta g_s &= \psi_+(g^{-1},g)^{-1}g_e\ot g_s + \psi_+(s^{-1}g^{-1}s,s^{-1}gs)^{-1} g_s\ot s^{-1}g_e s,\\
    \epsilon(g_e) &= \psi_+(g^{-1},g),\quad \epsilon(g_s) = 0,\\
    S(g_e)&=g_e^{-1},\quad S(g_s) = s^{-1}g_s^{-1}s.
\end{align*}

\twocolumngrid

\bibliography{ref}
\end{document}